\documentclass{article}
\usepackage[a4paper,margin=2cm,footskip=.5cm]{geometry}
\usepackage{setspace}
\usepackage{graphicx, epstopdf, parskip, latexsym, amsmath, amssymb, mathtools, times, amsthm, url, multicol, rotating, natbib, algorithm, algpseudocode}
\usepackage{fancyhdr}

\newtheorem{thm}{Theorem}
\newtheorem{Lem}{Lemma}
\newtheorem{Def}{Definition}

\newtheorem{Eg}{Example}
\newtheorem{Rem}{Remark}
\DeclareMathOperator{\Var}{Var}
\DeclareMathOperator{\Cov}{Cov}
\DeclareMathOperator{\Corr}{Corr}

\allowdisplaybreaks
\title{Spatio-temporal Ornstein-Uhlenbeck processes: \\ theory, simulation and statistical inference}
\author{Michele Nguyen \footnote{Address correspondence to: M. Nguyen, Department of Mathematics, Imperial College London, 180 Queen's Gate, SW7 2AZ London, UK.} \\
	Imperial College London \\
	michele.nguyen09@imperial.ac.uk
	\and 
	Almut E. D. Veraart \\
	Imperial College London \\
	 a.veraart@imperial.ac.uk
	}
\providecommand{\keywords}[1]{\textbf{\textit{Keywords:}} #1}

\date{}

\begin{document}

\maketitle

\begin{abstract}
Spatio-temporal modelling is an increasingly popular topic in Statistics. Our paper contributes to this line of research by developing the theory, simulation and inference for a spatio-temporal Ornstein-Uhlenbeck process. We conduct detailed simulation studies and demonstrate the practical relevance of these processes in an empirical study of radiation anomaly data. Finally, we describe how predictions can be carried out in the Gaussian setting.
\end{abstract}

\keywords{Spatio-temporal modelling, Ornstein-Uhlenbeck processes, stochastic simulation, moments-based inference.}

\section{Introduction}
\label{sec:Intro}

Due to advances in storage and computational efficiencies, more data with spatial and temporal information are being collected and shared. Taking the view that a pure temporal or spatial analysis of such data is insufficient, many scientists have proposed statistical models to study the spatio-temporal interactions and dependencies (see for example, Section 6.7 in \cite{CW2011} for an overview, \cite{BSR2012}, \cite{SH2012} and \cite{DKS2013}). We contribute to this research area by extending the theory of a spatio-temporal L\'evy-driven Ornstein-Uhlenbeck ($\mathrm{OU}$) process, and pioneering its simulation and inference.
\\
This spatio-temporal $\mathrm{OU}$ process, which is referred to as the $\mathrm{OU}_{\wedge}$ process in \cite{BN2004}, can be written as:
\begin{equation}
Y_{t}(\mathbf{x}) = \int_{A_{t}(\mathbf{x})} \exp(-\lambda(t-s)) L(\mathrm{d}\xi, \mathrm{d}s),
\label{eqn:OUh}
\end{equation}
where $\{Y_{t}(\mathbf{x}): \mathbf{x}\in\mathcal{X}, t\in\mathcal{T}\}$ is a random field in space-time $S = \mathcal{X} \times \mathcal{T}$. Usually, we have $\mathcal{X} = \mathbb{R}^{d}$ for some $d \in \mathbb{N}$ and $\mathcal{T} = \mathbb{R}$. Similar to the classical $\mathrm{OU}$ process, $\lambda > 0$ acts as a rate parameter. However, to cope with the new spatial dimension, we no longer integrate a L\'evy process over $(-\infty, t]$; instead, we integrate a homogeneous L\'evy basis $L$ over a set $A_{t}(\mathbf{x}) \subset S$ (with spatial and temporal integrating variables $\xi$ and $s$). The process is well-defined if the integral exists in the sense of the $\mathcal{L}_{0}$ integration theory \cite[]{RR1989}. A summary of this is given in Section 2 of the supplementary material provided in \cite{NV2015}. 
\\
The set $A_{t}(x)$ can be interpreted as the region in space-time that influences the field value at $(\mathbf{x}, t)$. This is in line with the interpretation of an ambit set (see for example \cite{BBV2012}). For an $\mathrm{OU}_{\wedge}$ process, we require that $A_{t}(\mathbf{x}) = A_{0}(0) + (\mathbf{x}, t) \subset S$ for translation invariance. Furthermore:
\begin{equation}
A_{s}(\mathbf{x}) \subset A_{t}(\mathbf{x}) \text{, } \forall \text{ } s < t, \text{ and } A_{t}(\mathbf{x}) \cap (\mathcal{X} \times (t, \infty)) = \emptyset. \label{eqn:ambitset}
\end{equation}
This implies that $A_{t}(\mathbf{x})$ has a temporal component of $(-\infty, t]$ just like the classical case. These conditions on the integrating set also give the $\mathrm{OU}_{\wedge}$ process several characteristic properties of the $\mathrm{OU}$ process: stationarity, Markovianity and an exponentially decaying autocorrelation function (ACF). 
\\
In the literature, there are other definitions of spatio-temporal $\mathrm{OU}$ processes. When the driving L\'evy noise is restricted to be Gaussian and $\mathcal{X} = \mathbb{R}$, a spatio-temporal process can be formed from the product of a temporal $\mathrm{OU}$ process with a spatial one \cite[]{TLB2004}. This is equal to a temporal $\mathrm{OU}$ process when the spatial component is fixed and a spatial $\mathrm{OU}$ process when the temporal component is fixed. Although this model features exponentially decaying temporal and spatial autocorrelation, one limitation is that the spatio-temporal autocorrelation is separable. 
\\
Alternatively, if we discretise space, we can create a spatio-temporal $\mathrm{OU}$ process by considering a multivariate $\mathrm{OU}$ process where each vector component of the process corresponds to a spatial location. Such an approach can be found in \cite{BD2001}. Spatial information is contained in the covariance matrix of the driving Brownian motion. Although this also results in a separable spatio-temporal autocorrelation, the spatial autocorrelation is comparatively flexible to that obtained by the previous method. 
\\
A further step in this direction would be to discretise the time domain. In this case, the $\mathrm{OU}$ process can be represented as an autoregressive (AR) model. A three-stage iterative procedure for the space-time modelling of autoregressive moving averages (ARMA) models can be found in \cite{PD1980}. The spatio-temporal autocorrelation in this case is defined differently and involves spatial neighbours of different ``orders".  
\\
The $\mathrm{OU}_{\wedge}$ process in (\ref{eqn:OUh}) has four key advantages over these alternative models. Firstly, it accommodates non-Gaussian driving noise which may be more appropriate in practice. Secondly, being a model in continuous time and space, it allows us to work with data of varying temporal and spatial scales. Thirdly, it allows for non-separable autocorrelation structures defined in the usual way. Finally, its integration set helps us identify the influence region for a particular time and space location. This could help scientists better understand the spatio-temporal interactions in real-life phenomena. 
\\
A wide class of $\mathrm{OU}_{\wedge}$ processes for $\mathcal{X} = \mathbb{R}$, which we shall refer to as the $g$-class, is given by:
\begin{equation}
Y_{t}(x) = \int_{-\infty}^{t}\int_{x - g(|t-s|)}^{x + g(|t-s|)} \exp(-\lambda(t-s)) L(\mathrm{d}\xi, \mathrm{d}s),
\label{eqn:gclass}
\end{equation} 
where $g$ is a non-negative strictly increasing continuous function on $[0, \infty)$. Now, $A_{t}(x) = \{(\xi, s): s\leq t, x-g(|t-s|)\leq \xi \leq x + g(|t-s|)\}$. The case $g(|t-s|) = c|t-s|$ with $c > 0$ is particularly interesting as it is related to an inhomogeneous stochastic wave equation (see for example, \cite{Dalang2009} and Chapter 5 of \cite{Chow2007}). When $L$ is Gaussian, $Z_{t}(x) = \exp(\lambda t) Y_{t}(x)$ is a mild solution of:
\begin{align*}
\frac{\partial^{2}Z}{\partial t^{2}} = c^{2}\frac{\partial^{2}Z}{\partial x^{2}} + \exp(\lambda t) \dot{L}(x, t), \text{ with }\lim_{t\rightarrow - \infty} Z_{t}(x) = 0, \text{ and } \lim_{t\rightarrow - \infty} \frac{\partial Z}{\partial t} = 0.
\end{align*}
Here, $c$ can be interpreted as the wave speed while the non-linear term $\exp(\lambda t)$ can be seen as the source function which describes how the sources of the waves affect the medium through which they travel. This effect is randomly perturbed by the Gaussian noise $\dot{L}$. This linear choice of $g$ is also the canonical example of an $\mathrm{OU}_{\wedge}$ process. The triangular shape of the integrating set for fixed $x$ and t is thought to be behind the ``$\wedge$" in the process's name.

\paragraph{Outline}
We begin in Section \ref{sec:prelim} by providing the theoretical background of the $\mathrm{OU}_{\wedge}$ process and the definitions of concepts to be used later. In Section \ref{sec:prop}, we review the process's known attributes as mentioned in \cite{BN2004} and derive new properties including spatio-temporal stationarity, temporal and spatial ergodicity, and autocorrelation structures. Sections \ref{sec:simalg} and \ref{sec:siminfer} look at how we can utilise these for simulation and inference respectively. We develop two algorithms for the canonical $\mathrm{OU}_{\wedge}$ process based on discrete convolutions. While one algorithm simulates on rectangular grids, another simulates on diamond-shaped grids. These will be compared and used in simulation studies with two estimation methods. The first is a spatio-temporal extension of the moments-matching inference method in \cite{KAKE2013}, and the second a least-squares approach to involve more lags of the normalised variograms. In Section \ref{sec:app}, we fit the canonical process to radiation anomaly data from the International Research Institute for Climate and Society to illustrate how $\mathrm{OU}_{\wedge}$ processes can be used in practice. We also provide theoretical results for prediction in the Gaussian scenario. The paper concludes with a summary of our key findings and an outlook on future research in Section \ref{sec:outlook}.

\section{Preliminaries}
\label{sec:prelim}

We have seen in (\ref{eqn:OUh}) that an $\mathrm{OU}_{\wedge}$ process is defined in terms of a stochastic integral over a subset of $S$ with respect to a L\'evy basis. To understand this more concretely, we need to set out several notations. Let $\mathcal{S}$ be the Borel $\sigma$-algebra on our space-time $S$ and $\mathcal{B}_{b}(S) = \{E \in \mathcal{S}: \lambda_{d+1}(E) < \infty\}$ where $\lambda_{d+1}$ denotes the Lebesgue measure on $\mathcal{B}(\mathbb{R}^{d+1})$. Throughout this paper, we work in the probability space $(\Omega, \mathcal{F}, P)$. A L\'evy basis is defined as follows \cite[]{BBV2012, Sato2007}: 
\\
\begin{Def}[L\'evy basis] \hfill \\
$L$ is a L\'evy basis on $(S, \mathcal{S})$ if it is an independently scattered and infinitely divisible random measure. This means that it satisfies the following conditions:
\begin{enumerate}
\item $L$ is a set of $\mathbb{R}$-valued random variables $\{L(E): E \in \mathcal{B}_{b}(S)\}$. Let $\{E_{i}: i \in \mathbb{N}\}$ be a sequence of disjoint elements of $\mathcal{B}_{b}(S)$. Then:
\begin{itemize}
\item for $\bigcup_{j = 1}^{\infty} E_{j} \in \mathcal{B}_{b}(S)$, $L(\bigcup_{j = 1}^{\infty} E_{j}) = \sum_{j = 1}^{\infty} L(E_j)$  almost surely;
\item furthermore, $L(E_{i})$ and $L(E_{j})$ are independent for $i\neq j$. 
\end{itemize}
\item For any finite collection $B_{1}, ..., B_{m}$ of elements of $\mathcal{B}_{b}(S)$, the random vector $\mathbf{L} = (L(B_{1}), ..., L(B_{m}))$ is infinitely divisible. That is, for any $n \in \mathbb{N}$, there exists a law $\mu_{n}$ such that the law of $\mathbf{L}$ can be expressed as $\mu = \mu_{n}^{*n}$, the n-fold convolution of $\mu_{n}$ with itself.
\end{enumerate}
\end{Def}
\vspace{4mm}
\begin{Rem}
The realisation of the random measures defined here are not in general measures because they need not have finite variation. The reader is advised to note that there are other definitions of random measures in the literature.
\end{Rem}
\vspace{4mm}
\begin{Def}[Cumulant generating function and cumulants] \label{defn:CGF}\hfill \\
Let $C\{\theta \ddagger Z\} = \log\mathbb{E}\left[\exp\left(i\theta Z\right)\right]$ denote the cumulant generating function (CGF) of a random variable $Z$. This means that $C\{\theta \ddagger Z\}$ is the unique complex-valued continuous function on $\mathbb{R}$ such that $C\{0 \ddagger Z\} = 0$ and $\exp(C\{\theta \ddagger Z\}) = \mathbb{E}[\exp(i\theta Z)]$. More information on the definition of a distinguished logarithm can be found on page 33 of \cite{Sato1999}. The cumulants of $Z$, $\kappa_{l}(Z)$, for $l = 1, 2, ...$, are defined through: $C\{\theta \ddagger Z\} = \sum_{l = 1}^{\infty} \kappa_{l}(Z)(i\theta)^{l}/l!$.
\end{Def}
\vspace{4mm}
\begin{Def}[Homogeneous L\'evy basis and its seed] \label{defn:lseed}\hfill \\
A L\'evy basis $L$ is homogeneous if there exists a random variable $L'$ such that $C\{\theta \ddagger L(E)\} = \int_{E} C\{\theta \ddagger L'\} \mathrm{d}\xi \mathrm{d}s$ for all $E\in\mathcal{B}_{b}(S)$. We refer to $L'$ as the L\'evy seed of $L$. Its CGF, $ C\{\theta \ddagger L'\}$, follows the L\'evy-Khintchine (LK) formula: $i a \theta - \frac{1}{2}b \theta^{2} + \int_{\mathbb{R}} (\exp(i\theta z) - 1 - i\theta z \mathbf{1}_{|z|\leq 1}) \nu(\mathrm{d}z)$ where $a\in\mathbb{R}$, $b\geq 0$ and $\nu$ is a L\'evy measure, i.e. it has no atom at $0$ and satisfies $\int_{\mathbb{R}} \min(1, z^{2}) \nu(\mathrm{d}z) < \infty$. We call $(a, b, \nu)$ the LK triplet of $L'$. 
\end{Def}
The relationship between the CGFs of $L(E)$ and $L'$ was established in Proposition 2.4 of \cite{RR1989}. The link between the CGFs of $L'$ and the resulting $\mathrm{OU}_{\wedge}$ process will be derived in Theorem \ref{thm:cgfform}.
\\
In this paper, we will be using $\mathrm{OU}_{\wedge}$ processes with homogeneous Gaussian bases for illustration. For such a basis, $L(E) \sim N(\mu\lambda_{d+1}(E), \tau^{2}\lambda_{d+1}(E))$ for any $E\in\mathcal{B}_{b}(S)$. Here, $\mu$ and $\tau^{2}$ are the mean and variance of the L\'evy seed respectively. Additional simulation studies have been conducted using homogeneous inverse Gaussian (IG), the normal inverse Gaussian (NIG) and the Gamma bases. Interested readers can refer to the supplementary material provided in \cite{NV2015}.
\\
As previously mentioned, the L\'evy basis $L$ and its seed play important roles in the distribution of the resulting $\mathrm{OU}_{\wedge}$ process. This can be characterised by a spatio-temporal extension of the generalised cumulant functional in \cite{BBV2012}:
\\
\begin{Def}[Generalised cumulant functional and cumulants]\label{defn:cgf} \hfill \\
Let $Y = \{Y_{t}(\mathbf{x})\}_{\mathbf{x}\in\mathbb{R}^{d}, t\in\mathbb{R}}$ denote a stochastic process in space-time ($S = \mathbb{R}^{d} \times \mathbb{R}$), and let $v$ denote any non-random measure such that the integral $v(Y) = \int_{S} Y_{t}(\mathbf{x}) v(\mathrm{d}\mathbf{x}, \mathrm{d}t)$ exists almost surely. The generalised cumulant functional (GCF) of $Y$ with respect to $v$ is given by: $C\{\theta \ddagger v(Y) \} = \log\mathbb{E}\left[\exp\left(i\theta v\left(Y\right)\right)\right]$ where the logarithm is the distinguished logarithm.
\\
If $ v(\mathrm{d}\mathbf{x}, \mathrm{d}t) = \delta_{t}(\mathrm{d}t)\delta_{\mathbf{x}}(\mathrm{d}\mathbf{x})$ where $\delta_{t}$ and $\delta_{\mathbf{x}}$ denote the Dirac measures at $t$ and $\mathbf{x}$ respectively, $C\{\theta \ddagger v(Y)\}$ is CGF of $Y_{t}(\mathbf{x})$. 
\end{Def}
As mentioned in \cite{BBV2012}, there are at least two other interesting cases of the measure $v$:
\begin{itemize}
\item[(i)] if $v(\mathrm{d}\mathbf{x}, \mathrm{d}t) = \theta_{1}\delta_{t_{1}}(\mathrm{d}t)\delta_{\mathbf{x}_{1}}(\mathrm{d}\mathbf{x}) + \dots +  \theta_{n}\delta_{t_{n}}(\mathrm{d}t)\delta_{\mathbf{x}_{n}}(\mathrm{d}\mathbf{x})$, $C\{\theta \ddagger v(Y)\}$ is the joint cumulant generating function (JCGF) of $Y_{t_{1}}(\mathbf{x}_{1}), \dots, Y_{t_{n}}(\mathbf{x}_{n})$; 
\item[(ii)] and if $v(\mathrm{d}\mathbf{x}, \mathrm{d}t) = \mathbf{1}_{I}(\mathbf{x}, t) \mathrm{d}\mathbf{x}\mathrm{d}t$, where $I$ is a region in $S = \mathbb{R}^{d} \times \mathbb{R}$ and $\mathbf{1}_{I}(\mathbf{x}, t) = 1$ if ${(\mathbf{x}, t)\in I}$ and $0$ otherwise, $C\{\theta \ddagger v(Y)\}$ determines the distribution of $\int_{I} Y_{t}(\mathbf{x})\mathrm{d}\mathbf{x}\mathrm{d}t$. 
\end{itemize}
Note that the case (ii) is relevant if we are using $\mathrm{OU}_{\wedge}$ processes to study integrated volatility or intermittency. 
\\
The next few definitions are required to study the theoretical properties of the $\mathrm{OU}_{\wedge}$ process and conduct inference. First, we extend the definition of (strict) temporal stationarity to that of spatio-temporal stationarity.
\\
We say that $Y_{t}(\mathbf{x})$ has temporal stationarity if for every $\mathbf{x} \in \mathbb{R}^{d}$, $\{Y_{t}(\mathbf{x})\}_{t \in \mathbb{R}}$ is a strictly stationary temporal process. That is, for every $n \in \mathbb{N}$ and $\epsilon \in \mathbb{R}$, such that $t_{1}, ..., t_{n} \in \mathbb{R}$, the joint distribution of $Y_{t_{1}}(\mathbf{x}), ..., Y_{t_{n}}(\mathbf{x})$ is the same as that of $Y_{t_{1}+ \epsilon}(\mathbf{x}), ..., Y_{t_{n} + \epsilon}(\mathbf{x})$. Analogously, $Y_{t}(\mathbf{x})$ is said to have spatio-temporal stationarity if for $n \in \mathbb{N}$, $\mathbf{u} \in \mathbb{R}^{d}$ and $\epsilon \in \mathbb{R}$, such that $x_{1}, ..., x_{n} \in \mathbb{R}^{d}$ and $t_{1}, ..., t_{n} \in \mathbb{R}$, the joint distribution of $Y_{t_{1}}(\mathbf{x}_{1}), ..., Y_{t_{n}}(\mathbf{x}_{n})$ is the same as that of $Y_{t_{1} + \epsilon}(\mathbf{x}_{1} + \mathbf{u}), ..., Y_{t_{n} + \epsilon}(\mathbf{x}_{n} + \mathbf{u})$.
\\
For our ergodicity results, we also need the classical definition of ergodicity, here given in the notation of \cite{FS2013}. Let $\{X_{t}\}_{t\in\mathbb{R}}$ be a real-valued strictly stationary process defined on the probability space $((\mathbb{R})^{\mathbb{R}}, \mathcal{F}, P)$ with $\mathcal{F} = \mathcal{B}((\mathbb{R})^{\mathbb{R}})$. We say that $(X_{t})_{t\in\mathbb{R}}$ is ergodic if $T^{-1}\int_{0}^{T}P(A \cap S^{t}B)\mathrm{d}t \stackrel{T\rightarrow\infty}{\rightarrow} P(A)P(B)$, where $A$, $B \in \mathcal{F}$. Here, $(S^{t})_{t\in\mathbb{R}}$ is the induced group of shift transformations on $(\mathbb{R})^{\mathbb{R}}$. That is, $S^{t}x_{s} = x_{s-t}$ for any $(x_{s})_{s\in\mathbb{R}}$ and $t\in\mathbb{R}$. $(X_{t})_{t\in\mathbb{R}}$ is mixing if $P(A\cap S^{t}B) \stackrel{t\rightarrow\infty}{\rightarrow} P(A)P(B)$. 

\section{Theoretical properties}
\label{sec:prop}

In this section, we review the known theoretical properties of $\mathrm{OU}_{\wedge}$ processes, as discussed in \cite{BN2004}, and derive various new results. These are important both from a modelling point of view and for developing statistical inference. The proofs of the Theorems, if not given or referenced, are provided in the Appendix. 

\subsection{Stationarity and Markovianity}

In \cite{BN2004}, it was proved in slightly different notation, that:
\\
\begin{thm} 
\label{thm:stationt}
An $\mathrm{OU}_{\wedge}$ process $\{Y_{t}(\mathbf{x})\}_{t\in\mathcal{T}}$ can be decomposed as:
\begin{align}
Y_{t}(\mathbf{x}) &= \exp{(-\lambda t)} Y_{0}(\mathbf{x}) + U_{t}(\mathbf{x}) + V_{t}(\mathbf{x}), \label{eqn:Ydecomp}\\
\text{where } U_{t}(\mathbf{x}) &= \exp{(-\lambda t)} \int_{A_{t}(\mathbf{x}) \cap (\mathcal{X}\times(-\infty, 0])\setminus A_{0}(\mathbf{x})} \exp{(\lambda s)} L(\mathrm{d}\xi, \mathrm{d}s), \nonumber \\
\text{and } V_{t}(\mathbf{x})&= \exp{(-\lambda t)} \int_{A_{t}(\mathbf{x}) \cap (\mathcal{X}\times(0, \infty))} \exp{(\lambda s)} L(\mathrm{d}\xi, \mathrm{d}s). \nonumber
\end{align}
In representation (\ref{eqn:Ydecomp}), $Y_{0}(\mathbf{x})$, $\{U_{t}(\mathbf{x})\}_{t\geq0}$ and $\{V_{t}(\mathbf{x})\}_{t\geq0}$ are independent. Furthermore, $\{Y_{t}(\mathbf{x})\}_{t\in\mathcal{T}}$ is Markovian and stationary.
\end{thm}

Stationarity is an important feature for inference as it allows us to pair observations based on temporal distances and consider quantities such as the sample autocorrelation. Markovianity is also a highly desirable property as it can be used to speed up simulations, make step-ahead predictions and simplify likelihood computations. 
\\
Besides temporal stationarity, Barndorff-Nielsen and Schmiegel also give an intuitive reasoning for $Y_{t}(\mathbf{x})$'s spatial stationarity. This involves the homogeneity of the L\'evy basis and the ambit set assumptions. However, by using the GCF, one can show that $Y_{t}(\mathbf{x})$ actually has spatio-temporal stationarity. By definition, this includes temporal and spatial stationarity. 
\\
\begin{thm}
\label{thm:GCF}
Let $A = A_{0}(\mathbf{0})$. Assume that for all $\xi \in \mathcal{X} = \mathbb{R}^{d}$ and $s \in \mathcal{T} = \mathbb{R}$, 
\begin{equation*}
h_{A}(\xi,s) = \int_{S} \mathbf{1}_{A}(\xi - \mathbf{x}, s-t) \exp(-\lambda(t-s)) v(\mathrm{d}\mathbf{x}, \mathrm{d}t) < \infty,
\end{equation*}
and that $h_{A}(\xi, s)$ is integrable with respect to the L\'evy basis $L$. Then, the GCF of $Y$ with respect to $v$ can be written as:
\begin{align*}
C\{\theta \ddagger v(Y)\} &= \int_{S} C\{\theta h_{A}(\xi, s) \ddagger L'\} \mathrm{d}\xi\mathrm{d}s \nonumber \\ &= i\theta a \int_{S} h_{A}(\xi, s) \mathrm{d}\xi\mathrm{d}s - \frac{1}{2}\theta^{2}b\int_{S} h^{2}_{A}(\xi, s) \mathrm{d}\xi\mathrm{d}s + \int_{\mathbb{R}} \int_{\mathbb{R}} (\exp(i\theta u z) - 1 - i\theta uz\mathbf{1}_{|z|\leq1}) \nu(\mathrm{d}z) \chi(\mathrm{d}u), 
\end{align*}
where $(a, b, \nu)$ is the LK triplet of $L'$ and $\chi$ is the measure on $\mathbb{R}$ obtained by transforming the Lebesgue measure on $S$ by the mapping $(\xi, s) \rightarrow h_{A}(\xi, s)$. 
\end{thm}

\begin{proof}
This is analogous to the proof for Proposition 5 in \cite{BBV2012} with $h_{A}$ being defined differently to account for space-time and our definition of $Y_{t}(\mathbf{x})$.
\end{proof}

\begin{thm} 
\label{thm:tsstation}
Let $Y_{t}(\mathbf{x})$ be an $\mathrm{OU}_{\wedge}$ process. Then $Y_{t}(\mathbf{x})$ has spatio-temporal stationarity.
\end{thm}

\begin{proof}
Let $v(\mathrm{d}\mathbf{x}, \mathrm{d}t) = \sum_{i = 1}^{n} \theta_{i}\delta_{t_{i}}(\mathrm{d}t)\delta_{\mathbf{x}_{i}}(\mathrm{d}x) $. From Theorem \ref{thm:GCF}, the JCGF of $Y_{t_{1}}(\mathbf{x}_{1}), \cdots, Y_{t_{n}}(\mathbf{x}_{n})$ can be written as:
\begin{align}
C\{\theta \ddagger v(Y)\} &= \int_{S} C\{\theta h_{A}(\xi, s) \ddagger L'\} \mathrm{d}\xi\mathrm{d}s \text{ where } h_{A}(\xi, s) = \sum_{i = 1}^{n}  \mathbf{1}_{A}(\xi - \mathbf{x}_{i}, s - t_{i}) \exp(-\lambda(t_{i}-s)) \theta_{i}.
\end{align}
By using the change of variables $\mathbf{u} = \xi - \mathbf{x}_{1}$ and $\epsilon = s - t_{1}$, we have $C\{\theta \ddagger v(Y)\} = \int_{\mathbb{R}\times\mathbb{R}} C\{\theta h'_{A}(\mathbf{u}, \epsilon) \ddagger L'\} \mathrm{d}u\mathrm{d}\epsilon$, where $h'_{A}(\mathbf{u}, \epsilon) := \sum_{i = 1}^{n}  \mathbf{1}_{A}(\mathbf{u} + \mathbf{x}_{1} - \mathbf{x}_{i}, \epsilon + t_{1} - t_{i}) \exp(-\lambda(t_{i}-t_{1} - \epsilon)) \theta_{i}$. The JCGF does not depend on the space or time coordinates but only their differences. This means that $Y_{t}(x)$ has spatio-temporal stationarity. 
\end{proof}

\subsection{Autocorrelation structures} \label{sec:autoc}

There is a wide range of literature on suitable choices of spatio-temporal autocorrelation and covariance functions (see Section 6.7 of \cite{CW2011} for a review). Hence, we show what type of autocorrelation structures can be obtained in our parsimonious class of $\mathrm{OU}_{\wedge}$ processes when the L\'evy basis $L$ is assumed to have finite second moments.
\\
With spatio-temporal stationarity, we can compute the spatio-temporal autocorrelation of such an $\mathrm{OU}_{\wedge}$ process.  Since $L$ is independently scattered, we find that, in terms of spatial and temporal differences ($d_{\mathbf{x}} \in \mathbb{R}^{d}$ and $d_{t} \in \mathbb{R}$): 
\begin{align}
\Cov[Y_{t}(\mathbf{x}), Y_{t + d_{t}}(\mathbf{x}+d_{\mathbf{x}})] &=  \Var[L'] \exp(-\lambda d_{t})\int_{A_{t}(\mathbf{x}) \cap A_{t+d_{t}}(\mathrm{x} + d_{\mathbf{x}})}\exp(-2\lambda(t-s))\mathrm{d}\xi\mathrm{d}s. \nonumber\\
\Rightarrow \Corr[Y_{t}(\mathbf{x}), Y_{t+d_{t}}(\mathbf{x}+d_{\mathbf{x}})] &= \frac{\exp(-\lambda d_{t})\int_{A_{t}(\mathbf{x}) \cap A_{t+d_{t}}(\mathrm{x} + d_{\mathbf{x}})}\exp(-2\lambda(t-s))\mathrm{d}\xi\mathrm{d}s}{\int_{A_{t}(\mathbf{x})}\exp(-2\lambda(t-s))\mathrm{d}\xi\mathrm{d}s}. \label{eqn:GCorr}
\end{align}
This provides information for both the temporal and the spatial autocorrelation. For the former, we set $d_{\mathbf{x}} = \mathbf{0}$ to get:
\begin{equation*}
\rho^{(T)}(d_{t}):= \Corr[Y_{t}(\mathbf{x}), Y_{t+d_{t}}(\mathbf{x})] = \exp\left(-\lambda |d_{t}|\right). 
\end{equation*}
This is identical to the autocorrelation of a classical $\mathrm{OU}$ process. For the spatial autocorrelation $\rho^{(S)}(d_{\mathbf{x}})$, we set $d_{t} = 0$. The actual form depends on our ambit set $A_{t}(\mathrm{x})$, as is shown in the next example.
\\
\begin{Eg}
\label{eg:Scorr}
Consider an $\mathrm{OU}_{\wedge}$ process of the $g$-class, i.e.: $Y_{t}(x) = \int_{-\infty}^{t}\int_{x - g(|t-s|)}^{x + g(|t-s|)} \exp(-\lambda(t-s)) L(\mathrm{d}\xi, \mathrm{d}s)$, where $g$ is a non-negative strictly increasing continuous function on $[0, \infty)$. Assume without loss of generality that $d_{x}>0$, and let $(x^{*}, t^{*}) = \left(x + \frac{d_{x}}{2}, t - g^{-1}\left(\frac{d_{x}}{2}\right)\right)$ be the intersection point between $A_{t}(x)$ and $A_{t}(x + d_{x})$. Then:
\begin{align}
\Cov[Y_{t}(x), Y_{t}(x+d_{x})] &= \Var[L']  \int_{-\infty}^{t^{*}}\int_{x + d_{x} - g(|t - s|)}^{x +  g(|t-s|)} \exp(-2\lambda(t - s)) \mathrm{d}\xi\mathrm{d}s \nonumber\\
&= \Var[L'] \left( \int_{-\infty}^{t^{*}}2 g(t-s)\exp(-2\lambda(t - s)) \mathrm{d}s  - \int_{-\infty}^{t^{*}}d_{x}\exp(-2\lambda(t - s)) \mathrm{d}s \right) \nonumber\\
&=  \Var[L'] \left( \int_{0}^{\infty}2 g\left(w+ g^{-1}\left(\frac{d_{x}}{2}\right)\right)e^{-2\lambda\left(w+ g^{-1}\left(\frac{d_{x}}{2}\right)\right)} \mathrm{d}w - \int_{0}^{\infty}d_{x}e^{-2\lambda\left(w + g^{-1}\left(\frac{d_{x}}{2}\right)\right)} \mathrm{d}w \right), \label{eqn:covg}
\end{align}
where $w = t^{*} - s$. We give two examples to show how the spatial autocorrelation structure varies according to our choice of the ambit set: 

\begin{itemize}
\item[(i)] If $g(|t-s|) = c|t-s|$, $g^{-1}\left(d_{x}/2\right) = d_{x}/2c$. By solving the integrals in (\ref{eqn:covg}), we find that $\rho^{(S)}(d_{x}) = \exp\left(-\lambda d_{x}/c\right)$. 
\item[(ii)] If $g(|t-s|) = \begin{cases}
  c_{1}|t-s| & \text{if } |t-s| \leq 1, \nonumber \\
  c_{1} + c_{2}(|t-s|-1) & \text{ otherwise,} \nonumber
  \end{cases}$ \hspace{4mm}$g^{-1}\left(d_{x}/2\right) = \begin{cases} d_{x}/2 c_{1} \text{ if } d_{x}\leq 2 c_{1}, \\
1 + \frac{d_{x} - 2 c_{1}}{2 c_{2}} \text{ otherwise. }
\end{cases}$\vspace{2mm}\\ By solving (\ref{eqn:covg}) for $d_{x}\leq 2c_{1}$ and $d_{x}>2c_{1}$ separately, we have:
\begin{equation*}
\rho^{(S)}(d_{x}) = \begin{cases} \frac{\left(c_{1} + (c_{2} - c_{1})e^{-2\lambda\left(1 - d_{x}/2 c_{1}\right)}\right)e^{-\lambda d_{x}/c_{1}}}{c_{1} + (c_{2} - c_{1})e^{-2\lambda}} \text{ for } d_{x}\leq 2 c_{1},\\
\frac{c_{2}e^{-2\lambda\left(1 + \left(d_{x}-2 c_{1}\right)/2 c_{1}\right)}}{c_{1} + (c_{2} - c_{1})e^{-2\lambda}} \text{ otherwise.}
\end{cases}
\end{equation*}
\end{itemize}
\end{Eg}

Figure \ref{fig:rhosplots2} shows the plots of the spatial autocorrelations for cases (i) and (ii) under different values of $\lambda$, $c$, $c_{1}$, $c_{2}$ and $d_{x}$. From Plot (ii), we notice that introducing a change in gradient of $g(|t-s|)$ from $0.5$ to $1$ at $|t-s| = 1$ results in changes in the overall scaling and gradient of the spatial autocorrelation. In particular, since the spatial correlation is a combination of two functions which are joined at $d_{x} = 2 c_{1} = 1$, there is a change of behaviour at that point as denoted by the deviation from the first function (FF). Since we can choose different values of $\lambda$, $c$, $c_{1}$ and $c_{2}$ to fit the curvatures of our empirical spatial autocorrelations, we see that $\mathrm{OU}_{\wedge}$ processes provide a flexible tool for modelling spatial dependencies.
\\
In the canonical case with $g(|t-s|) = c|t-s|$, we can find an explicit form for the spatio-temporal autocorrelation:
\\
\begin{figure}[tbp]
\centering
\caption{Plot of the spatial autocorrelations achieved by different ambit sets of the form $A_{t}(x) = \{(\xi, s): s \leq t, x-g(|t-s|) \leq \xi \leq x+g(|t-s|)\}$ where: (i) $g(|t-s|) = c|t-s|$; (ii) $g(|t-s|) =  c_{1}|t-s|$ if $|t-s| \leq 1$, and $ c_{1} + c_{2}(|t-s|-1)$ otherwise. FF refers to the functional form of $\rho^{(S)}(d_{x})$ for $d_{x}\leq 2c_{1}$.}
\label{fig:rhosplots2}
\includegraphics[width = 5in, height = 2.5in, trim = 0.4in 0.2in 0.4in 0in]{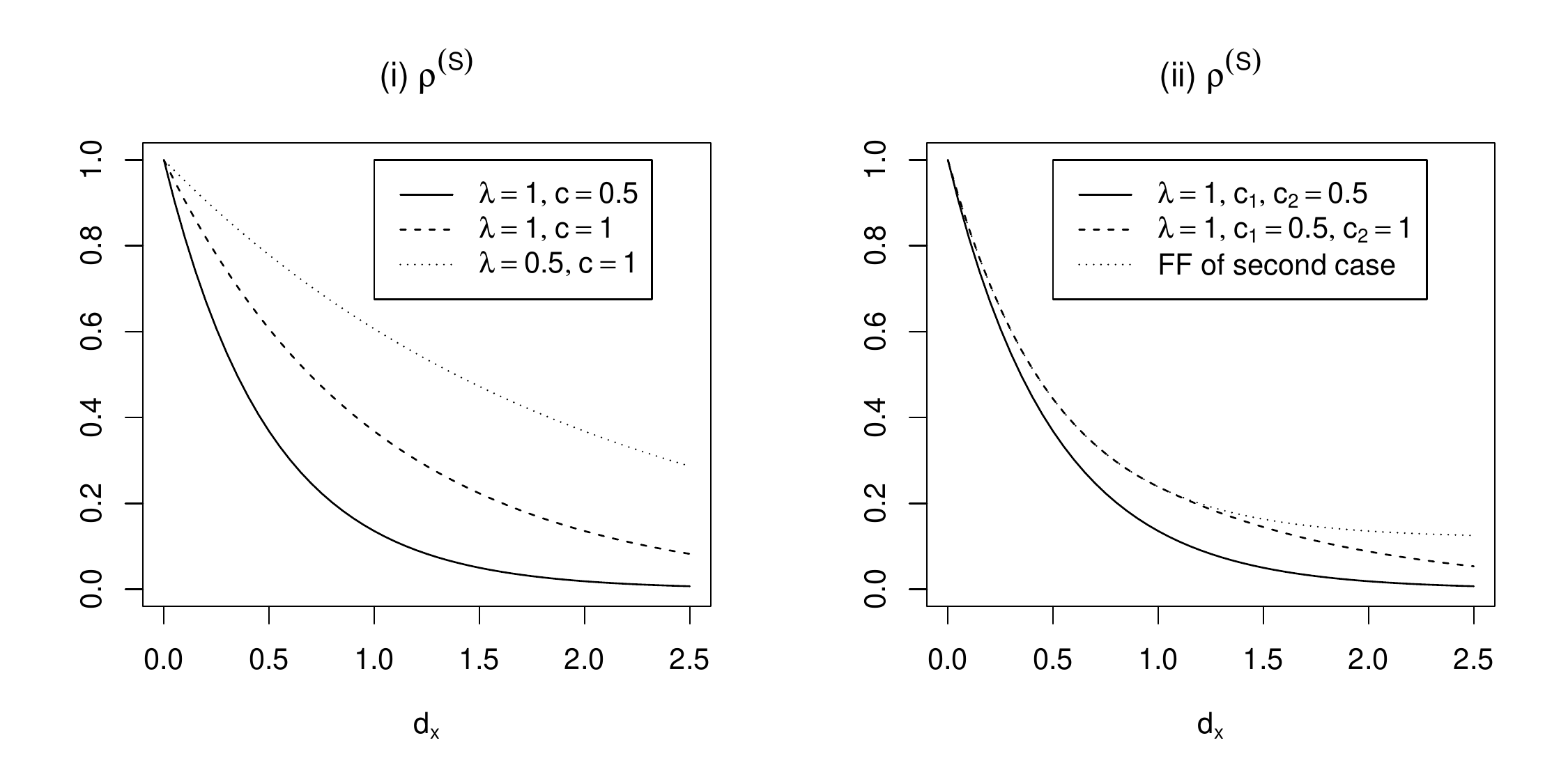}
\end{figure}

\begin{Eg}
\label{eg:STcorr}
Let $Y_{t}(x)$  be an $\mathrm{OU}_{\wedge}$ process described by:
\begin{equation}
Y_{t}(x) = \int_{-\infty}^{t}\int_{x - c|t-s|}^{x + c|t-s|} \exp(-\lambda(t-s)) L(\mathrm{d}\xi, \mathrm{d}s),\label{eqn:TestOUh}
\end{equation}
where $c > 0$. By considering the intersection of $A_{t}(x)$ and $A_{t+d_{t}}(x+d_{x})$ under different relative magnitudes of $d_{x}$ and $c d_{t}$, we find that the spatio-temporal autocorrelation is given by:
\begin{equation}
\Corr[Y_{t}(x), Y_{t + d_{t}}(x+d_{x})]  = \min\left(\exp\left(-\lambda |d_{t}|\right), \exp\left(-\frac{\lambda|d_{x}|}{c}\right)\right). \label{eqn:STCorr}
\end{equation}
This is a non-separable ACF which cannot be obtained through the definitions of the spatio-temporal $\mathrm{OU}$ process in \cite{BD2001} and \cite{TLB2004}. Note that when $d_{x} = 0$, we get $\rho^{(T)}(d_{t}) = \exp(-\lambda |d_{t}|)$ as expected and when $d_{t} = 0$, we get $\rho^{(S)}(d_{x}) = \exp\left(-\lambda |d_{x}|/c\right)$ as calculated before. We have exponential decay in time and in space but with different rate parameters. Interested readers can refer to Section 3.1 of the supplementary material in \cite{NV2015} for the details of this derivation.
\end{Eg}
\vspace{4mm}
\begin{Eg} \label{eg:GJCGF}
The spatio-temporal ACF derived in the previous example can be used to obtain the JCGF of an $\mathrm{OU}_{\wedge}$ process of the form (\ref{eqn:TestOUh}) with a Gaussian basis. The CGF of a Gaussian L\'evy seed $L'$ with mean $\mu$ and variance $\tau^{2}$ is $C\{\theta\ddagger L'\} = i\mu\theta - \tau^{2}\theta^{2}/2$. Let $h_{A}(\xi, s) = \sum_{i = 1}^{n}  \mathbf{1}_{A}(\xi - x_{i}, s - t_{i}) \exp(-\lambda(t_{i}-s)) \theta_{i}$ where $A = A_{0}(0)$, the ambit set at the origin and define $\tilde{\theta}_{i} := \theta\theta_{i}$. With reference to the notation in the proof of Theorem \ref{thm:tsstation}, the JCGF is given by:
\begin{align*}
&\int_{\mathbb{R}\times\mathbb{R}}C\{\theta h_{A}(\xi, s) \ddagger L'\} \mathrm{d}\xi\mathrm{d}s \\
=&  i\mu\sum_{i = 1}^{n} \tilde{\theta}_{i}\int_{A_{t_{i}}(x_{i})} \exp(-\lambda(t_{i}-s))  \mathrm{d}\xi\mathrm{d}s -\sum_{i, j = 1}^{n}\frac{1}{2}\tau^{2}\tilde{\theta}_{i}\tilde{\theta}_{j}  \int_{A_{t_{i}}(x_{i}) \cap A_{t_{j}}(x_{j})} \exp(-\lambda(t_{i} + t_{j}-2s)) \mathrm{d}\xi\mathrm{d}s, 
\\
=& i \frac{2c\mu}{\lambda^{2}}\sum_{i = 1}^{n}\tilde{\theta}_{i}   -\sum_{i, j = 1}^{n}\frac{1}{2}\tau^{2}\tilde{\theta}_{i}\tilde{\theta}_{j}\frac{c}{2\lambda^{2}} \min\left(\exp\left(-\lambda |t_{i} - t_{j}|\right), \exp\left(-\frac{\lambda|x_{i} - x_{j}|}{c}\right)\right) \text{ using (\ref{eqn:GCorr}) and (\ref{eqn:STCorr}).} 
\end{align*}
\end{Eg}

\subsection{Equality in law to the $\mathrm{OU}$ process}

In the previous subsection, we saw that the temporal autocorrelation of an $\mathrm{OU}_{\wedge}$ process is exactly that of an $\mathrm{OU}$ process. In many circumstances, we also find that when its spatial parameter is fixed, the $\mathrm{OU}_{\wedge}$ process is equal in distribution to an $\mathrm{OU}$ process. This is one key reason the $\mathrm{OU}_{\wedge}$ process can be seen as a spatio-temporal $\mathrm{OU}$ process. Before we present the corresponding Theorem, we give a Lemma which is used in its proof:
\\
\begin{Lem}
\label{Lem:Yt1t2}
Let $t_{1} < t_{2} \in \mathbb{R}$ be two arbitrary time points, and let $Y_{t}(\mathbf{x})$ be an $\mathrm{OU}_{\wedge}$ process. Then, for fixed $\mathbf{x}\in\mathbb{R}^{d}$:
\begin{equation}
Y_{t_{2}}(\mathbf{x}) = \exp(-\lambda(t_{2}-t_{1})) Y_{t_{1}}(\mathbf{x}) + \widetilde{U}(t_{1}, t_{2}, \mathbf{x}),
\end{equation}
where  $\widetilde{U}(t_{1}, t_{2}, \mathbf{x}) = \int_{A_{t_{2}}(\mathbf{x})\backslash A_{t_{1}}(\mathbf{x})} \exp(-\lambda(t_{2}-s)) L(\mathrm{d}\xi, \mathrm{d}s)$ is independent of $Y_{t_{1}}(\mathbf{x})$.
\end{Lem}
\begin{proof}
Since $A_{s}(\mathbf{x}) \subset A_{t}(\mathbf{x})$, $\forall s<t$, $A_{t_{2}}(\mathbf{x}) \cap A_{t_{1}}(\mathbf{x}) = A_{t_{1}}(\mathbf{x}) $ and $A_{t_{2}}(\mathbf{x})$ can be partitioned into $A_{t_{1}}(\mathbf{x})$ and $A_{t_{2}}(\mathbf{x}) \backslash A_{t_{1}}(\mathbf{x})$. The result follows directly by using the definition of $Y_{t}(\mathbf{x})$ and the fact that the L\'evy basis is independently scattered.
\end{proof} 

\clearpage

\begin{thm}
\label{thm:OUh-OU}
Let $Y_{t}(\mathbf{x}) = \int_{A_{t}(\mathbf{x})} \exp(-\lambda(t-s)) L(\mathrm{d}\xi, \mathrm{d}s)$ be an $\mathrm{OU}_{\wedge}$ process where the L\'evy seed $L'$ has LK triplet $(a, b, \nu)$, and let $Z_{t} = \int_{-\infty}^{t} \exp(-\lambda(t-s)) \tilde{L}(\mathrm{d}s)$ be an $\mathrm{OU}$ process where the L\'evy seed $\tilde{L}'$ has LK triplet $(\tilde{a}, \tilde{b}, \tilde{\nu})$. Then, for fixed $\mathbf{x}\in\mathbb{R}^{d}$, the process $\{Y_{t}(\mathbf{x})\}_{t\in\mathbb{R}}$ is equal in law to $\{Z_{t}\}_{t\in\mathbb{R}}$ if:   
\begin{align*}
\tilde{a}&= \lambda\left[ \int_{A} e^{\lambda w}  \left[a- \int_{\mathbb{R}}z\mathbf{1}_{1<|z|\leq e^{-\lambda w}} \nu(\mathrm{d}z)\right]\mathrm{d}\mathbf{u}\mathrm{d}w + \int_{-\infty}^{0} e^{\lambda w}  \int_{\mathbb{R}}z\mathbf{1}_{1<|z|\leq e^{-\lambda w}} \tilde{\nu}(\mathrm{d}z)\mathrm{d}w\right],\\
\tilde{b} &= 2\lambda \int_{A} b e^{2\lambda w} \mathrm{d}\mathbf{u}\mathrm{d}w \hspace{2mm} \text{ and } \hspace{2mm} \int_{-\infty}^{0}\tilde{\nu}(e^{-\lambda w}\mathrm{d} y)\mathrm{d}w = \int_{A}\nu(e^{-\lambda w}\mathrm{d} y)\mathrm{d}\mathbf{u}\mathrm{d}w,
\end{align*}
where $A = A_{0}(\mathbf{0})$, $\mathbf{u} = \xi - \mathbf{x}$ and $w = s-t$.
\end{thm}
\vspace{4mm}
\begin{Eg}
\label{eg:gclassOUhOU}
Let $Y_{t}(x)$ be an $\mathrm{OU}_{\wedge}$ process of the $g$-class and suppose that the L\'evy seeds $L'$ and $\tilde{L'}$ have L\'evy densities, i.e. there exists measures $\nu_{d}$ and $\tilde{\nu}_{d}$ such that $\nu(\mathrm{d}y) = \nu_{d}(y)\mathrm{d}y$ and $\tilde{\nu}(\mathrm{d}y) = \tilde{\nu}_{d}(y)\mathrm{d}y$. We can find an explicit formula for $\tilde{\nu}$ by translating the third condition of Theorem \ref{thm:OUh-OU} into $\int_{-\infty}^{0}e^{-\lambda w}\tilde{\nu}_{d}(e^{-\lambda w}y)\mathrm{d}w = \int_{-\infty}^{0}2g(-w) e^{-\lambda w}\nu_{d}(e^{-\lambda w}y)\mathrm{d}w$, $\forall y\in\mathbb{R}$. By substituting $z  = e^{-\lambda w} y$, we have:
\begin{align}
\int_{y}^{\infty} \frac{\tilde{v}_{d}(z)}{\lambda y} \mathrm{d}z &= \int_{y}^{\infty} \frac{2}{\lambda y}g\left(\frac{\log(z/y)}{\lambda}\right) {v}_{d}(z) \mathrm{d}z \Leftrightarrow \int_{y}^{\infty} \tilde{v}_{d}(z) \mathrm{d}z = \int_{y}^{\infty} 2 g\left(\frac{\log(z/y)}{\lambda}\right)  {v}_{d}(z) \mathrm{d}z \label{eqn:LDrelation} \\
\Rightarrow \tilde{\nu}(\mathrm{d}z) &= 2 \left[(\lambda z)^{-1} \left(\int_{z}^{\infty}\left.\frac{\delta g(w)}{\delta w}\right|_{w = \lambda^{-1}\log(z/y)}\nu_{d}(y)\mathrm{d}y\right)\mathrm{d}z + g(0)\nu(\mathrm{d}y)\right], \nonumber
\end{align}
where we have used Leibniz's Integral Rule to differentiate both sides of (\ref{eqn:LDrelation}) with respect to $y$.  
\end{Eg}
\vspace{4mm}
\begin{Eg}
\label{eg:LsubOUhOU}
Let $Y_{t}(x)$ be a canonical $\mathrm{OU}_{\wedge}$ process as described in (\ref{eqn:TestOUh}). If $L$ is a Gaussian basis whose seed has mean $\mu$ and standard deviation $\tau$, for fixed $x$, $Y_{t}(x) \stackrel{d}{=} Z_{t}$ where $Z_{t}$ is an $\mathrm{OU}$ process defined by $Z_{t} = \int_{-\infty}^{t} \exp(-\lambda(t-s)) \widetilde{L}(\mathrm{d}s)$. Here, $\tilde{L}$ is a Gaussian basis whose seed has mean $\tilde{\mu} = 2c\mu/\lambda$ and standard deviation $\tilde{\tau} = \sqrt{c\tau^{2}/\lambda}$.
\end{Eg}
In general, the L\'evy measure of $L$, $\nu$ is related to that of $\tilde{L}$, $\tilde{\nu}$ via a transformation depending on $\lambda$ and the ambit set. This means that the distribution of the L\'evy process $\tilde{L}$ corresponding to $Z_{t}$ may not be of the same family as that of the L\'evy basis $L$ corresponding to $Y_{t}(\mathbf{x})$.

\subsection{Time-changed versions}

Time-changed versions of classical OU processes have been studied by Barndorff-Nielsen and Shephard. These were introduced to eliminate dependence of the cumulants on the parameter $\lambda$ \cite[]{BS2001}. This means that $\lambda$ only appears in the ACFs and can be interpreted as a memory parameter. For the $\mathrm{OU}_{\wedge}$ process, we have an additional spatial dimension. As a result, we have more than one way to introduce a time change. Note that these methods are based on particular classes of $\mathrm{OU}_{\wedge}$ processes. The first time-change method is taken directly from the temporal idea:
\\
\begin{Def}[Time-changed $\mathrm{OU}_{\wedge}$ process of order $q$]
Let $L$, $\lambda$ and $A_{t}(\mathbf{x})$ satisfy the conditions in the definition of an $\mathrm{OU}_{\wedge}$ process. $Y_{t}(\mathbf{x})$ is an time-changed $\mathrm{OU}_{\wedge}$ ($\mathrm{TCOU}_{\wedge}$) process of order $q \in \mathbb{N}$ if:
\begin{equation*}
Y_{t}(\mathbf{x}) = \int_{A_{t}(\mathbf{x})} \exp(-\lambda(t-s)) L(\mathrm{d}\xi, \lambda^{q}\mathrm{d}s), \text{ and the integral is well-defined.}
\end{equation*}
\end{Def}

\begin{thm}
\label{thm:tcform1}
Consider the following $(q+ 1)^{\text{th}}$ order $\mathrm{TCOU}_{\wedge}$ process:
\begin{equation}
Y_{t}(x) = \int_{-\infty}^{t}\int_{x - c|t-s|^{q}}^{x + c|t-s|^{q}} \exp(-\lambda(t-s)) L(\mathrm{d}\xi, \lambda^{q+1} \mathrm{d}s), \label{eqn:TCOUhq}
\end{equation}
where $q \in \mathbb{N}$ is chosen such that the ambit set satisfies the conditions (\ref{eqn:ambitset}). The cumulants of $Y_{t}(x)$ do not depend on $\lambda$.
\end{thm}

Another way to eliminate the $\lambda$ dependence in the cumulants is to define a first-order $\mathrm{TCOU}_{\wedge}$ process with time-changed ambit sets. This means that $A_{t}(x)$ depends on the parameter $\lambda$. 
\\
\begin{thm}
\label{thm:tcform2}
Let $g$ be a non-negative strictly increasing continuous function on $[0, \infty)$. We can construct an $\mathrm{TCOU}_{\wedge}$ process as follows:
\begin{equation}
Y_{t}(x) = \int_{-\infty}^{t}\int_{x - g(\lambda|t-s|)}^{x + g(\lambda|t-s|)} \exp(-\lambda(t-s)) L(\mathrm{d}\xi, \lambda \mathrm{d}s), \label{eqn:TCOUhg}
\end{equation}
where $x , t \in \mathbb{R}$ and $\lambda > 0$. $Y_{t}(x)$ has cumulants which do not depend on $\lambda$.
\end{thm}

The equality in law of the $\mathrm{OU}_{\wedge}$ and the $\mathrm{OU}$ processes at fixed locations can also be preserved after a time-change:
\\
\begin{thm}
\label{thm:tcourelation}
Let $Y^{1}_{t}(\mathbf{x})$ be a $\mathrm{TCOU}_{\wedge}$ process of the type (\ref{eqn:TCOUhq}) where the L\'evy seed $L_{1}'$ has LK triplet $(a_{1}, b_{1}, \nu_{1})$ and let $Y^{2}_{t}(\mathbf{x})$ be a $\mathrm{TCOU}_{\wedge}$ process of the type (\ref{eqn:TCOUhg}) where the L\'evy seed $L_{2}'$ has LK triplet $(a_{2}, b_{2}, \nu_{2})$. Further define  $Z^{1}_{t} =  \int_{-\infty}^{t}\exp(-\lambda (t-s)) \tilde{L}_{1}(\lambda\mathrm{d}s)$ and $Z^{2}_{t} =  \int_{-\infty}^{t}\exp(-\lambda (t-s)) \tilde{L}_{2}(\lambda\mathrm{d}s)$ to be time-changed $\mathrm{OU}$ ($\mathrm{TCOU}$) processes with corresponding L\'evy seed LK triplets $(\tilde{a}_{1}, \tilde{b}_{1}, \tilde{\nu}_{1})$ and  $(\tilde{a}_{2}, \tilde{b}_{2}, \tilde{\nu}_{2})$ respectively. Then, for fixed $\mathbf{x}\in\mathbb{R}^{d}$, the process $\{Y^{1}_{t}(\mathbf{x})\}_{t\in\mathbb{R}}$ is equal in law to $\{Z^{1}_{t}\}_{t\in\mathbb{R}}$ if:   
\begin{align*}
\tilde{a}_{1}&= \int_{0}^{\infty} 2cw^{q} e^{-w}  \left[a_{1}- \int_{\mathbb{R}}z\mathbf{1}_{1<z\leq e^{w}} \nu_{1}(\mathrm{d}z)\right]\mathrm{d}w + \int_{0}^{\infty} e^{- w}  \int_{\mathbb{R}}z\mathbf{1}_{1<z\leq e^{w}} \tilde{\nu}_{1}(\mathrm{d}z)\mathrm{d}w\\
\tilde{b}_{1} &= 4\int_{0}^{\infty} b_{1}cw^{q}  e^{-2w} \mathrm{d}w \hspace{2mm} \text{ and } \hspace{2mm} \int_{0}^{\infty}\tilde{\nu}_{1}(e^{w}\mathrm{d} y)\mathrm{d}w = \int_{0}^{\infty}2cw^{q}\nu_{1}(e^{ w}\mathrm{d} y)\mathrm{d}w,
\end{align*}
where $\mathbf{u} = \mathbf{x} - \xi$ and $w = t-s$. Similarly, for fixed $\mathbf{x}\in\mathbb{R}^{d}$, the process $\{Y^{2}_{t}(\mathbf{x})\}_{t\in\mathbb{R}}$ is equal in law to $\{Z^{2}_{t}\}_{t\in\mathbb{R}}$ if:   
\begin{align*}
\tilde{a}_{2}&= \int_{0}^{\infty} 2g(w) e^{-w}  \left[a_{2}- \int_{\mathbb{R}}z\mathbf{1}_{1<z\leq e^{w}} \nu_{2}(\mathrm{d}z)\right]\mathrm{d}w + \int_{0}^{\infty} e^{- w}  \int_{\mathbb{R}}z\mathbf{1}_{1<z\leq e^{w}} \tilde{\nu}_{2}(\mathrm{d}z)\mathrm{d}w\\
\tilde{b}_{2} &= 4\int_{0}^{\infty} b_{2}g(w)  e^{-2w} \mathrm{d}w \hspace{2mm} \text{ and } \hspace{2mm} \int_{0}^{\infty}\tilde{\nu}_{2}(e^{w}\mathrm{d} y)\mathrm{d}w = \int_{0}^{\infty}2g(w)\nu_{2}(e^{ w}\mathrm{d} y)\mathrm{d}w.
\end{align*}
\end{thm}

\begin{Eg}
When $L_{1}'$ and $L_{2}'$ have L\'evy densities $\nu^{1}_{d}$ and $\nu^{2}_{d}$ respectively, we can find explicit formulae for $\tilde{\nu}_{1}$ and $\tilde{\nu}_{2}$. Similar to Example \ref{eg:gclassOUhOU}, we can write the L\'evy measure conditions in terms of L\'evy densities and use Leibniz's Integral Rule. This gives us: 
\begin{equation*}
\tilde{v}_{1}(\mathrm{d}y) = 2cq y^{-1}\int_{y}^{\infty}\left(\log\left(z/y\right)\right)^{q-1}\nu_{1}(\mathrm{d}z)\mathrm{d}y \text{ and } \tilde{v}_{2}(\mathrm{d}y) =  2\left[y^{-1} \int_{y}^{\infty}\left.\frac{\mathrm{d}g(x)}{\mathrm{d}x}\right|_{x = \log\left(z/y\right)}\nu_{2}(\mathrm{d}z)\mathrm{d}y + g(0)\nu_{2}(\mathrm{d}y)\right].
\end{equation*}
\end{Eg}
As mentioned, working with $\mathrm{TCOU}_{\wedge}$ processes allows us to interpret the parameter $\lambda$ solely as the memory parameter.

\subsection{Ergodicity}
\label{sec:ergodic}
Stationarity and autocorrelations are useful for estimating the parameters of an $\mathrm{OU}_{\wedge}$ process. Another useful property is ergodicity which is typically used to derive a Weak Law of Large Numbers. This in turn is used to establish the consistency of the estimators. We find that the $g$-class of $\mathrm{OU}_{\wedge}$ processes possesses this property in time and space:
\\
\begin{thm}
\label{thm:ergodicity}
Let $Y_{t}(x)$ be defined by (\ref{eqn:gclass}). Then, $\{Y_{t}(x)\}_{t\in\mathbb{R}}$ and $\{Y_{t}(x)\}_{x\in\mathbb{R}}$ are ergodic. 
\end{thm}

\subsection{Probability distributions}
\label{sec:CGF}

Cumulants, together with normalised variograms, are used to conduct inference on $\mathrm{OU}_{\wedge}$ processes  in Section \ref{sec:siminfer}. We present a theorem relating the CGF of an $\mathrm{OU}_{\wedge}$ process $Y_{t}(\mathbf{x})$ to its L\'evy seed $L'$:

\clearpage

\begin{thm}
\label{thm:cgfform}
With $A = A_{0}(\mathbf{0})$, $\mathbf{u} = \xi - \mathbf{x}$ , $w = s-t$ and $(a, b, \nu)$ being the LK triplet of $L'$, the CGF of $Y_{t}(\mathbf{x})$ can be written as:
\begin{align}
C\{\theta \ddagger Y_{t}(\mathbf{x})\} &:= \log\left(\mathbb{E}\left[\exp\left(i\theta Y_{t}(\mathbf{x}\right)\right)]\right) = \int_{A_{t}(\mathbf{x})} C\{\theta e^{-\lambda(t-s)} \ddagger L'\} \mathrm{d}\xi \mathrm{d}s  = \int_{A} C\{\theta e^{\lambda w} \ddagger L'\} \mathrm{d}\mathbf{u} \mathrm{d}w \label{eqn:CCrelation}\\
&= i\theta a_{Y} - \frac{1}{2}\theta^{2} b_{Y} + \int_{\mathbb{R}} (e^{i\theta z} - 1 - i\theta z\mathbf{1}_{|z|\leq1}) \nu_{Y}(\mathrm{d}z), \nonumber
\end{align}
where $a_{Y} =  \int_{A} e^{\lambda w}  \left[a- \int_{\mathbb{R}}z\mathbf{1}_{1<|z|\leq e^{-\lambda w}} \nu(\mathrm{d}z)\right]\mathrm{d}\mathbf{u}\mathrm{d}w$, $b_{Y} = \int_{A} b e^{2\lambda w} \mathrm{d}\mathbf{u}\mathrm{d}w$ and $\nu_{Y}(\mathrm{d}y) = \int_{A}\nu(e^{-\lambda w}\mathrm{d} y)\mathrm{d}\mathbf{u}\mathrm{d}w$.
\end{thm} 
\begin{proof}
We can obtain (\ref{eqn:CCrelation}) by considering a special case of Theorem \ref{thm:GCF}. To find the LK triplet $(a_{Y}, b_{Y}, \nu_{Y})$, we write (\ref{eqn:CCrelation}) in terms of the LK triplet of $L'$ and evaluate each resulting integral. For $a_{Y}$, this involves restricting the integration bound for the inner integral to cases where $1_{|ze^{-\lambda w}|\leq 1} - 1_{|z|\leq 1}$ is non-zero. For $\nu_{Y}$, we use $y = z \exp(\lambda w)$ and change the order of integration so that we integrate over $u$ and $w$ before integrating over $y$. 
\end{proof}
\begin{Eg} \label{eg:GMean}
For the canonical $\mathrm{OU}_{\wedge}$ process given by (\ref{eqn:TestOUh}), suppose that $L$ is a Gaussian basis whose seed has mean $\mu$ and standard deviation $\tau$. Then, we obtain that $Y_{t}(x) \sim N\left(2c\mu/\lambda^{2}, c\tau^{2}/2\lambda^{2}\right)$ since:
\begin{equation}
C\{\theta \ddagger Y_{t}(x)\} = \int_{-\infty}^{t}\int_{x-c|t-s|}^{x+c|t-s|} C\{\theta\exp(-\lambda(t-s)) \ddagger L'\} \mathrm{d}\xi \mathrm{d}s = i\theta\left(\frac{2c\mu}{\lambda^{2}}\right) - \frac{1}{2}\left(\frac{c\tau^{2}}{2\lambda^{2}}\right)\theta^{2}. \label{eqn:GaussianCGF}
\end{equation}
\end{Eg}

As the cumulants of $Y_{t}(\mathbf{x})$ are the coefficients in the power series of its CGF, we can use Theorem \ref{thm:cgfform} to derive expressions for the cumulants of $Y_{t}(\mathbf{x})$, $\kappa_{l}\left(Y_{t}\left(\mathbf{x}\right)\right)$ for $l \in \mathbb{N}$. For the canonical $\mathrm{OU}_{\wedge}$ process, it can be shown that:
\begin{align*}
\kappa_{l}(Y_{t}(x)) = \kappa_{l}(L')\int_{-\infty}^{t} \int_{x-c|t-s|}^{x+c|t-s|} \exp(-l\lambda(t-s)) \mathrm{d}\xi\mathrm{d}s = \kappa_{l}(L')\int_{-\infty}^{t} 2c(t-s) \exp(-l\lambda(t-s)) \mathrm{d}s = \kappa_{l}(L') \frac{2c}{l^{2}\lambda^{2}}.
\end{align*}
Here, $\kappa_{l}(L')$ refers to the $l^{\text{th}}$ cumulant of the L\'evy seed.

\section{Simulation}
\label{sec:simalg}

In this section, we develop two simulation algorithms for the canonical process defined by (\ref{eqn:TestOUh}). Let $Y_{t}(x)$ be a canonical process. By reducing the stochastic integral to a truncated sum, we obtain its discrete convolution (DC) approximation. This can be computed efficiently by the in-built convolution algorithms in software such as R and Mathematica. Based on this idea, we create two algorithms: one for values on the usual rectangular grid and another for values on a so-called diamond grid to mimic the edges of the integration set.  The former can be easily extended to more general $\mathrm{OU}_{\wedge}$ processes. 
\\
By writing $w = t-s$ and $u = x-\xi$, (\ref{eqn:TestOUh}) is equal to:
\begin{align}
Y_{t}(x) &= -\int_{0}^{\infty}\int_{x-cw}^{x+cw}\exp(-\lambda w)L(\mathrm{d}\xi, t - \mathrm{d}w) = -\int_{0}^{\infty}\int_{-\infty}^{\infty}h(u, w) L(x - \mathrm{d}u, t - \mathrm{d}w) \nonumber \\ 
&= \int_{0}^{\infty}\int_{-\infty}^{\infty}h(u, w) L'_{(x, t)}(\mathrm{d}u, \mathrm{d}w), 
\label{eqn:Yhrep} 
\end{align}
where $L'_{(x, t)} := -L(x-\cdot, t - \cdot)$ is the L\'evy basis given by $L'_{x, t}(A) = -L((x, t) - A)$ for $A \in \mathcal{B}_{b}(S)$, and $h(u, w) = \mathbf{1}_{|u| \leq cw}\exp\left(-\lambda w \right)$. Figure \ref{fig:hplot}(i) shows the plot of $h(u, w)$ when $c = 1$. 
\\
Since $h(u, w) \rightarrow 0$ as $u \rightarrow \pm \infty$ and $w \rightarrow \infty$, it seems reasonable to approximate (\ref{eqn:Yhrep}) by the following DC:
\begin{align}
Y_{t}(x) &\approx \sum_{j = 0}^{p}  \sum_{i = -q}^{q} h(u_{i}, w_{j}) W^{(x, t)}_{ij}, \label{eqn:Ydcapprox}
\end{align}

\clearpage

\begin{figure}[tbp]
\centering
\caption{(i) Plot of the kernel $h(u, w) = \mathbf{1}_{|u| \leq cw}\exp\left(-\lambda w \right)$ when $c = 1$ for $ 0 \leq w \leq 4, -4 \leq u \leq 4$; (ii) plot of the rectangular simulation grid set by spatial and temporal coordinates  $\{x_{I}: I = 1, ..., n\}$ and $\{t_{J}: J = 1, ... m\}$ and the extension to accommodate the kernel approximation for the boundaries. The original grid points are denoted by points while the additional ones are denoted by crosses. Here, $n = m = 4$ and $p = q = 2$. The L\'evy noise over the extended simulation grid is represented by the $W$s which are attached to the grid points on their left. For the grid point $(t_{1}, x_{n})$, the kernel is evaluated at the points within the shaded area. The shaded area is an approximation for the ambit set (now in $u, w$ form). The true transformed ambit set is outlined by the bold red line.}
\label{fig:hplot}
\includegraphics[width = 6.2in, height = 3.2in, trim = 0.3in 1in 0.3in 0.1in]{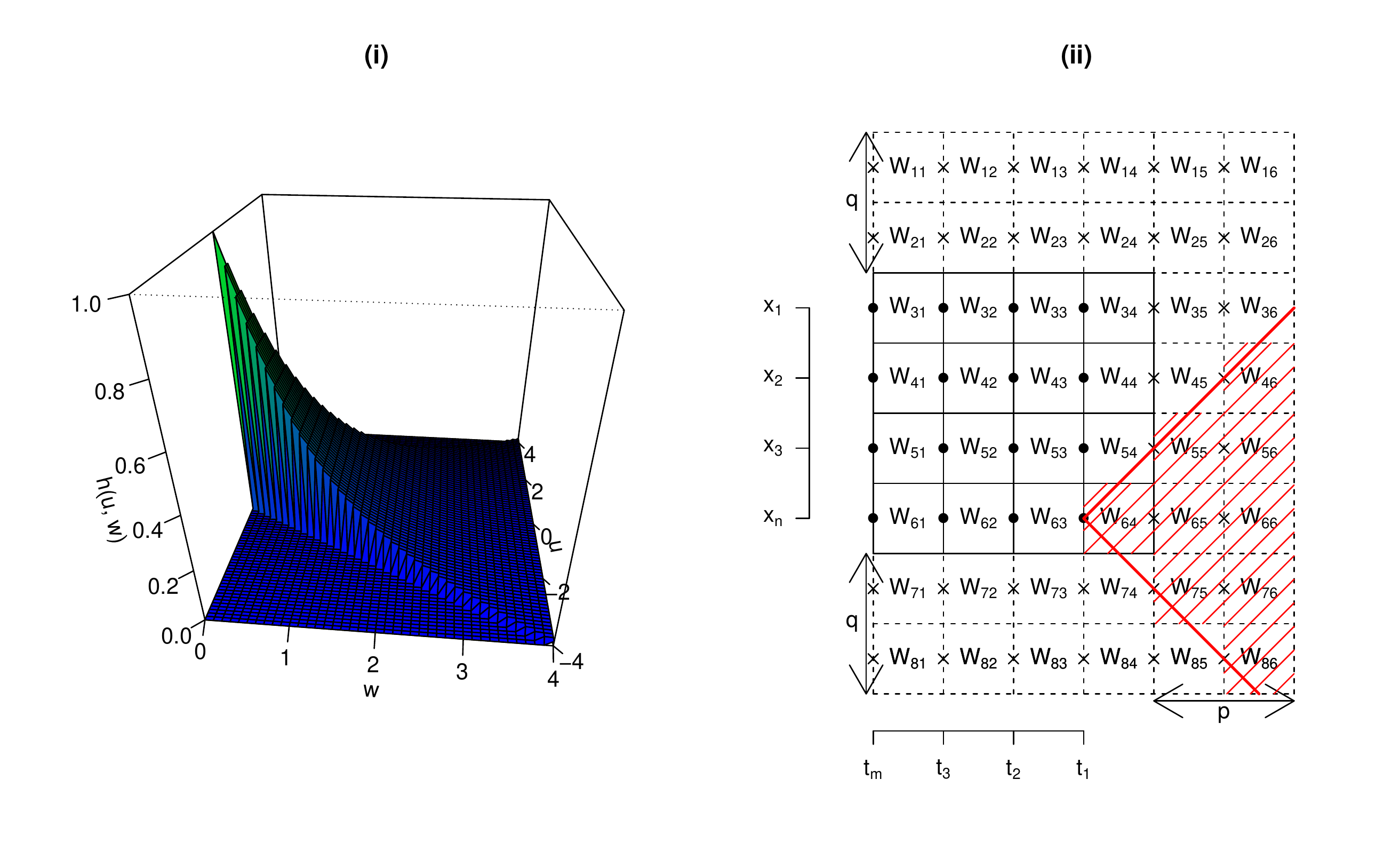}
\end{figure}

\begin{figure}[tbp]
\centering
\caption{(i) Field view of data simulated from the rectangular grid algorithm for $Y_{t}(x)$ when $g(w) = w$, $\lambda = 1$ and $L$ is a Gaussian basis whose seed has mean $0.2$ and standard deviation $0.1$; (ii) a heat map for the same data. The simulation grid is over the space-time region $[0, 10] \times [0, 10]$ with grid size $\triangle_{u} = \triangle_{w} = 0.05$ and kernel truncation parameters $p = q = 300$.}
\label{fig:simeg}
\includegraphics[width = 2.8in, height = 3.0in, trim = 0.4in 0.5in 0.4in 0in]{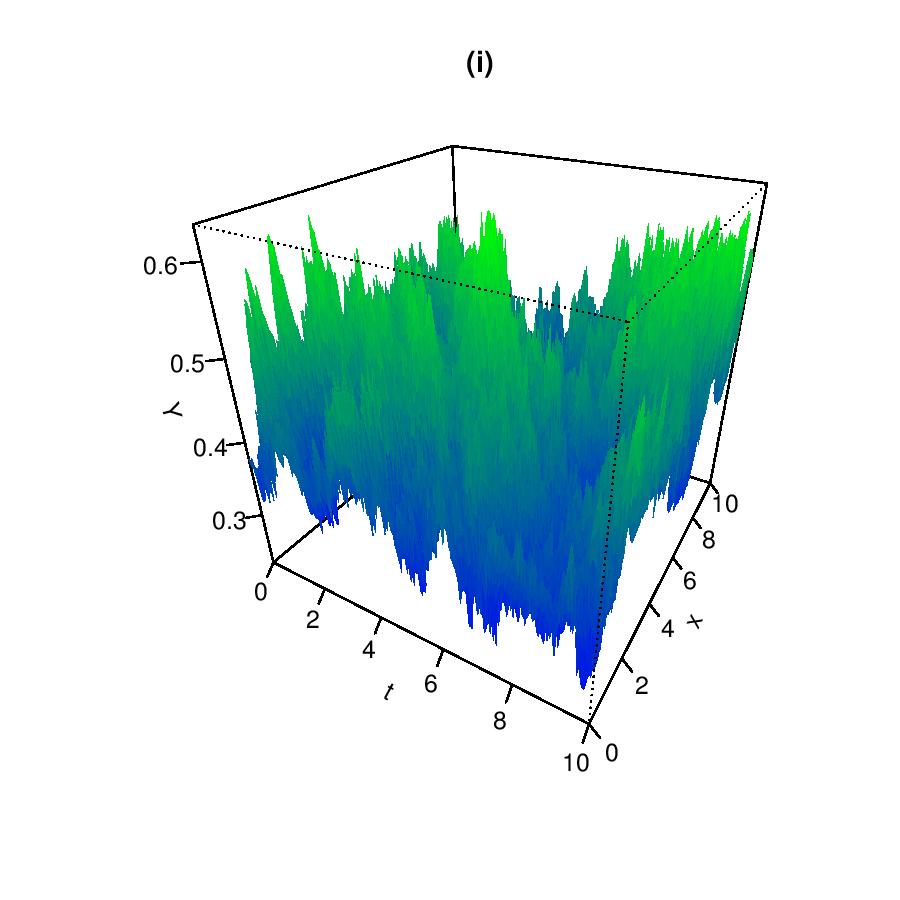}
\includegraphics[width = 2.8in, height = 3.0in, trim = 0.4in 0.5in 0.4in 0in]{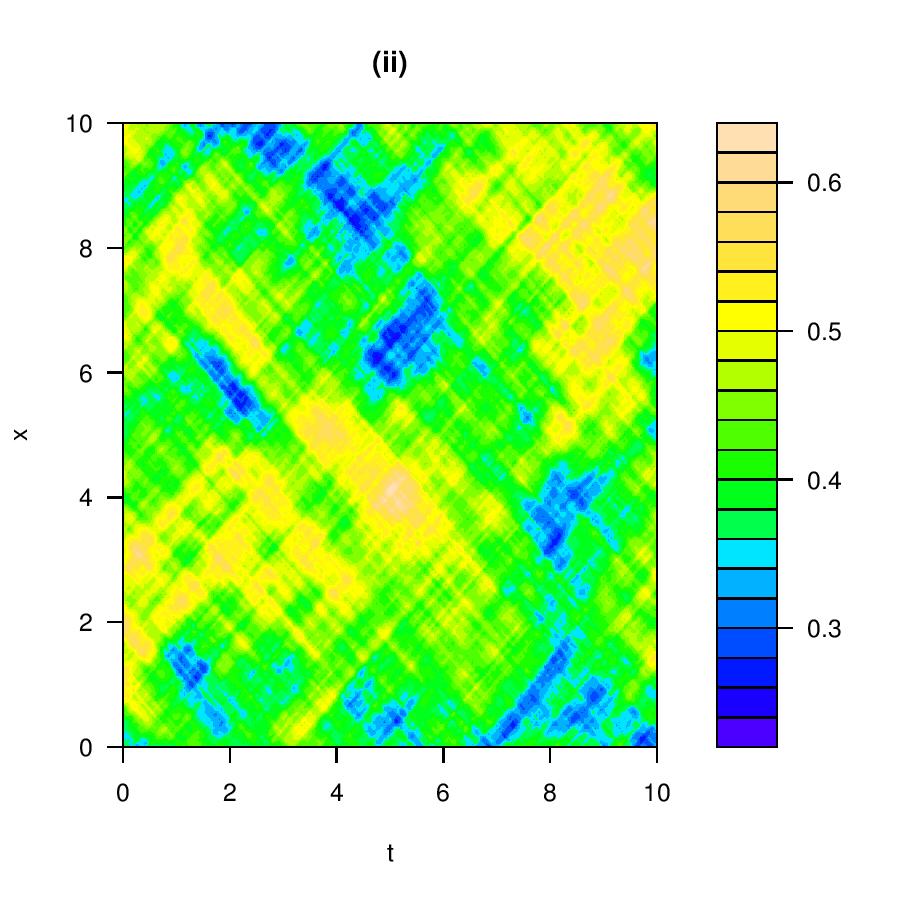}
\end{figure}

\clearpage

where $w_{j} = j\triangle_{w}$ and $u_{i} = i\triangle_{u}$ for some $\triangle_{w}, \triangle_{u} > 0$, and $p$ and $q$ are arbitrarily chosen integers. $\{W^{(x,t)}_{ij}: -q \leq i \leq q, 0 \leq j \leq p\}$ is a set of independent identically distributed random variables which depend on $(x, t)$. These represent the L\'evy noise over each grid segment. For example, for the rectangular grid algorithm, we take $W^{(x, t)}_{ij} \sim L([0, \triangle_{u}] \times [0, \triangle_{w}])$ for all $i$, $j$. The approximation in (\ref{eqn:Ydcapprox}) is expected to improve as $\triangle_{w}, \triangle_{u} \rightarrow 0$ and $p, q \rightarrow \infty$. 

\subsection{Rectangular grid algorithm} \label{sec:RGrid}

Referring to (\ref{eqn:Ydcapprox}), we can use the optimised convolution function for vectors in the R software to create a simulation algorithm for $Y_{t}(x)$. Alternatively, one can use any similar function in other software such as the convolution function for lists in Mathematica. Since such functions are typically written in terms of Discrete Fourier Transforms, it is also possible to code using these directly.
\\
To begin, we choose $p, q, \triangle_{u}$ and $\triangle_{w}$. We create a list of row vectors of values of $h$:
\begin{align*}
\hat{h} = &\{\left[h\left(u_{-q}, w_{0}\right), h\left(u_{-q}, w_{1}\right), \dots ,  h\left(u_{-q}, w_{p}\right)\right], \dots, \left[h\left(u_{0}, w_{0}\right), h\left(u_{0}, w_{1}\right), \dots ,  h\left(u_{0}, w_{p}\right)\right], \\
&\dots, \left[h\left(u_{q}, w_{0}\right), h\left(u_{q}, w_{1}\right), \dots ,  h\left(u_{q}, w_{p}\right)\right]\}.
\end{align*}
Next, we choose our simulation grid by determining the spatial and temporal coordinates $\{x_{I}: I = 1, ..., n\}$ and $\{t_{J}: J = 1, ..., m\}$. These need to be separated by $\triangle_{u}$ and $\triangle_{w}$ respectively. To account for boundary effects, i.e. to be able to use the entire kernel approximation for boundary points, we need to extend the grid by $q$ on both ends of the spatial axis and $p$ on one end of the temporal axis. In Figure \ref{fig:hplot}(ii), we consider the case $n = m = 4$, $p = q = 2$ and $c = 1$. Simulating with such a discretisation scheme means that there are three sources of simulation error: the kernel truncation, the kernel discretisation and the block approximation of the ambit set.
\\
Let $\tilde{n} = n+2q$, $\tilde{m} = m + p$ and $\hat{w}_{kl}$ denote the realisations of $W_{kl}$ for $k = 1, ..., \tilde{n}$ and $l = 1, ..., \tilde{m}$. We create a list of row vectors of these values: $\widehat{W} = \left\{ \left[\hat{w}_{11}, \hat{w}_{12}, \dots ,  \hat{w}_{1\tilde{m}}\right], \dots, \left[\hat{w}_{\tilde{n}1}, \hat{w}_{\tilde{n}2}, \dots ,  \hat{w}_{\tilde{n}\tilde{m}}\right]\right\}$.
\\
With $\hat{h}_{\tilde{i}}$ as the $\tilde{i}^{\text{th}}$ row vector in $\hat{h}$ and $\widehat{W}_{k}$ as the $k^{\text{th}}$ row vector in $\widehat{W}$, we can simulate from $Y_{t}(x)$ using Algorithm \ref{alg:DCrect}.
\\
Step \algref{alg:DCrect}{line:conrect} is equivalent to calculating $\mathbf{y} = [\hat{Y}_{t_{m}}(x_{I}), \dots, \hat{Y}_{t_{1}}(x_{I})]$ with $\hat{Y}_{t_{J}}(x_{I}) = \sum_{j = 0}^{p}\sum_{i = -q}^{q} h(u_{i}, w_{j}) \hat{w}_{(I + \tilde{i}-1)\tilde{j}}$, where $\tilde{i} = i + q + 1$, and $\tilde{j} = J+j$ for $J = 1, \dots, m$.
\\
Figure \ref{fig:simeg} shows some data simulated over the space-time region $[0, 10] \times [0, 10]$ for $Y_{t}(x)$ when $c = 1$, $\lambda = 1$ and $L$ is a Gaussian basis whose seed has mean $0.2$ and standard deviation $0.1$. The step sizes were $\triangle_{u} = \triangle_{w} = 0.05$ and $p = q = 300$.

\subsection{Diamond grid algorithm} \label{sec:sscom}

Simulating the process on a square or rectangular grid requires us to approximate the triangular ambit set by blocks. This will affect the temporal and spatial autocorrelations which are dependent on the approximated intersection areas (see Section \ref{sec:autoc}). We want to better approximate the shape of our ambit set and their intersections by simulating on a grid that follows the shape of the ambit set. We also want to be in the DC framework for computational efficiency. To this end, we consider a diamond grid (DG) built on the original rectangular grid (RG). 
\\
We start with a motivating example. Suppose that we have the equi-spaced square grid $\{(x_{i}, t_{j}): i = 1, \dots, n, j = 1, \dots, m\}$ with $n = m = 5$. We want to simulate values for $Y_{t}(x)$ when $c = 1$. Figure \ref{fig:dgridkernel}(i) shows how one may construct a DG on this square grid. The new grid consists of diamonds with height twice the spatial spacing $\triangle_{x}$ and width twice the temporal spacing. As before, we need to extend the grid by $q$ and $p$ units to account for boundary effects.
\\
With $c = 1$, we require that the left corner of the diamond has an upper slope of $1$ and a lower slope of $-1$. Using diamonds to approximate the ambit set allows us to capture the ambit shape more accurately. However,  one consequence is that we cannot obtain $Y$ values for the points within the diamonds. We can only obtain $Y$ values for points at the edges of the diamond. That is, our simulated data only consists of field values at $(x_{i}, t_{j})$ where $i = 1, \dots, 5$, and $j = 1, 3, 5 \text{ if } i \text{ is odd and } j = 2, 4 \text{ if } i \text{ is even}$. 

\begin{algorithm}[tbp]
\caption{Discrete convolution on a rectangular grid}\label{alg:DCrect}
\begin{algorithmic}[1]
\State $Y \gets rep(NA, m)$ \Comment{We will append rows to $Y$ as $x_{I}$ increases to form our data matrix.}
\For{$I = 1, \dots, n$} 
\State	$y \gets rep(0, m)$ \Comment{For each $x_{I}$, we create a storage vector.}
\For {$\tilde{i} = 1, \dots, 2q+1$}
\State $j \gets I + \tilde{i} - 1$ 
\State $y \gets y + convolve(\hat{W}_{j}, \hat{h}_{\tilde{i}}, type = 'f')$ \Comment{We add the convolutions of $\hat{h}_{\tilde{i}}$ with $\hat{W}_{I + \tilde{i}  - 1}$ to $\mathbf{y}$.} \label{line:conrect}
\EndFor
\State	$Y \gets rbind(Y, rev(yvect))$ \Comment{Append $\mathbf{y}$ in reversed form so the time index increases from the left.}
\EndFor
\State $Y \gets Y[-1, ]$ \Comment{Remove the first dummy row of $Y$.}
\end{algorithmic}
\end{algorithm}

The choice of the grid shape affects all parts of our simulation. To account for the distribution of the L\'evy noise within the diamonds, we do not simulate values for $W$ from independent $L([0, \triangle_{x}] \times [0, \triangle_{t}])$ random variables. Instead, we scatter the $W$s uniformly about the DG. If the $W$s lie within diamonds, we simulate values for it from independent L\'evy random variables whose distribution is proportional to the area of the diamond. If the $W$s lie at the edges of the diamonds (denoted by the circles), we set their values to $0$. 
\\
Next, we construct the truncated kernel matrix with the missing $Y$ values in mind. For $p = q = 2$, we define: 
\begin{align*}
\widehat{H} = \begin{pmatrix}
h(u_{-2}, w_{0}) & 0 & h(u_{-2, w_{2}}) \\
0 & h(u_{-1}, w_{1}) & 0 \\
h(u_{0}, w_{0}) & 0 & h(u_{0}, w_{2}) \\
0 & h(u_{1}, w_{1}) & 0 \\
h(u_{2}, w_{0}) & 0 & h(u_{2}, w_{2})
\end{pmatrix},
\end{align*}
where $h(u, w) = \mathbf{1}_{|u| \leq w} \exp(-\lambda w)$. Figure \ref{fig:dgridkernel}(ii) shows the action of the kernel matrix on the grid of $W$s for the simulation of $Y_{t_{m}}(x_{3})$. The blue shaded area denotes the approximated ambit set $\hat{A}_{t_{m}}(x_{3})$ and the blue cross denotes its starting point. To correctly weigh the L\'evy noise in the shaded area by the kernel values, we overlay the area with the red dotted lines whose intersections denote the values in the truncated kernel matrix. The value for $Y_{t_{m}}(x_{3})$ is then found by multiplying the overlaid kernel values with their respective $W$ values. 
\\
Note that by using the DC algorithm, we will automatically generate a value for the points within the diamonds as they lie on the square grid. However, the alternation with the $0$s in the kernel matrix cancels out any values generated for these points. To see this, consider $Y_{t_{4}}(x_{3})$ and shift the red dotted lines to the right by $\triangle_{t}$. Now, the non-zero kernel values are mutliplied with the zero-valued $W$s and the kernel zeros are multiplied with the non-zero $W$s.
\\
Now that we have an intuition for the DG algorithm, we can construct it. Upon further examination, we find that in order for steps to be easily extended to other cases, we require $p$ and $q$ to be even integers, $n$ and $m$ to be odd integers and $\triangle_{x} = c\triangle_{t}$. Under these conditions, we can create the algorithm for general $c$. 
\\
The pseudo code in Algorithm \ref{alg:ktdia} constructs the $(2q+1) \times (p+1)$ truncated kernel matrix: 
\begin{align*}
\widehat{H}_{ij} = \begin{cases} \mathbf{1}_{|u_{i}|\leq c w_{j}}\exp(-\lambda w_{j}) & \text{if } i + j \text{ is even,} \\
0 & \text{otherwise,}\\
\end{cases}
\end{align*}
where $w_{j} = j\triangle_{t}$ for $j = 0, \dots, p$ and $u_{i} = i\triangle_{x}$ for $i = -q, \dots, q$. This is later split to create a list of the rows of $\widehat{H}$, $\hat{h}$:
\begin{align*}
\hat{h} = \{ &\left[h\left(u_{-q}, w_{0}\right), h\left(u_{-q}, w_{1}\right), \dots ,  h\left(u_{-q}, w_{p}\right)\right], \dots, \left[h\left(u_{0}, w_{0}\right), h\left(u_{0}, w_{1}\right), \dots ,  h\left(u_{0}, w_{p}\right)\right], \\
&\dots, \left[h\left(u_{q}, w_{0}\right), h\left(u_{q}, w_{1}\right), \dots ,  h\left(u_{q}, w_{p}\right)\right]\}.
\end{align*}
With our truncated kernel matrix, we proceed to simulate from the L\'evy random variables:
\begin{align*}
\widehat{W}_{kl} \stackrel{d}{=} \begin{cases} L([0, 2c\triangle_{t}]\times[0, \triangle_{t}]) & \text{if } k + l \text{ is even,} \\
0 & \text{otherwise,}\\
\end{cases}
\end{align*}

\clearpage

\begin{figure}[tbp]
\centering
\caption{(i) Plot of the diamond grid for the case $c = 1$ with $n = m = 5$ and $p = q = 2$ (the $W$s denote the underlying L\'evy random variables); (ii) the simulation of $Y_{t_{m}}(x_{3})$ for the case $c = 1$ with $n = m = 5$ and $p = q = 2$. The shaded blue region denotes the approximated ambit set and the dotted red lines denote the truncated kernel matrix. The intersection of the lines show the kernel value to be multiplied by the $W$ at that location.} \label{fig:dgridkernel}
\includegraphics[width = 5.2in, height = 3.4in, trim = 0.4in 0.4in 0.4in 0in]{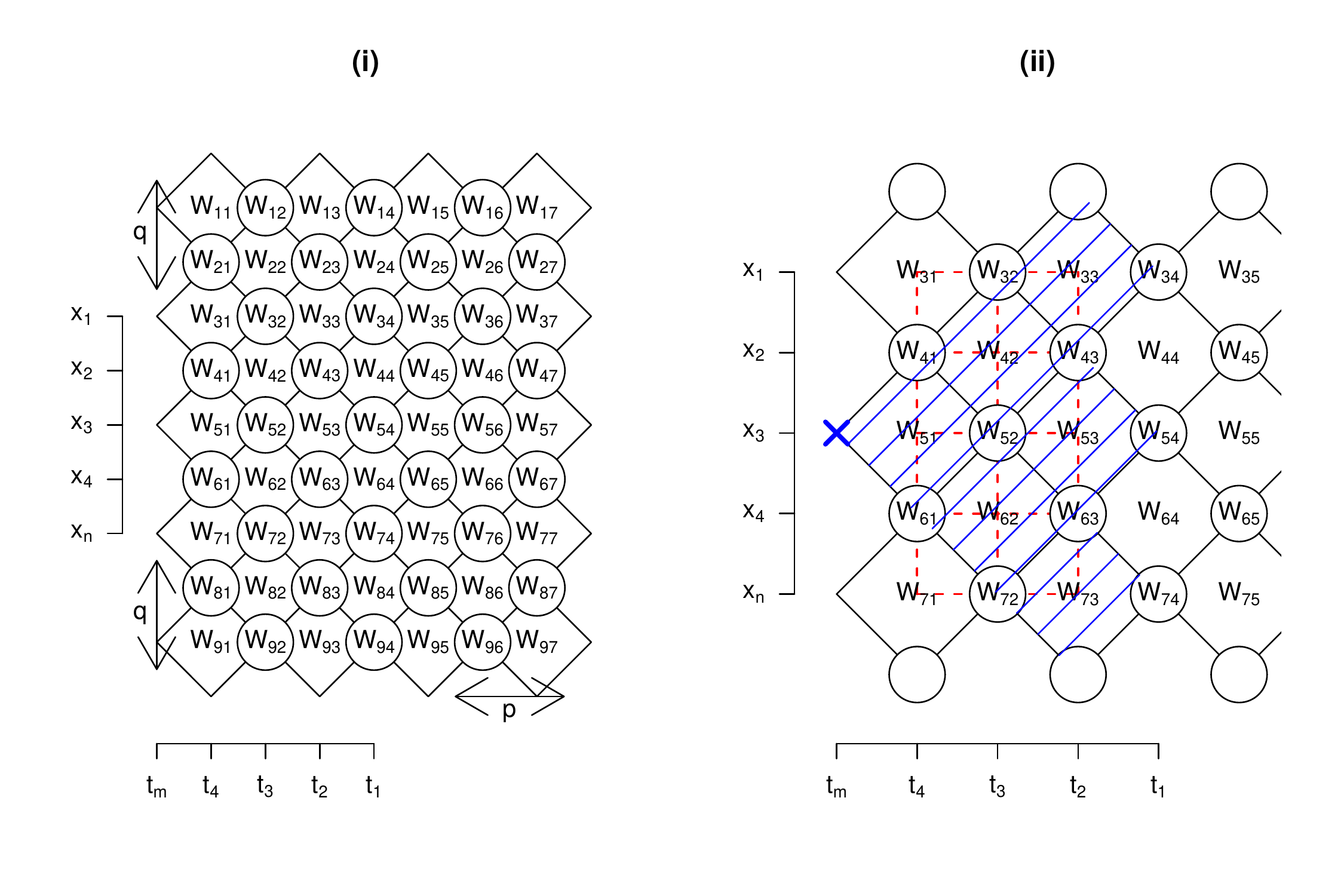}
\end{figure} 

\begin{algorithm}[tbp]
\caption{Obtaining the truncated kernel matrix for the diamond grid algorithm}
\label{alg:ktdia}
\begin{algorithmic}[1]
\Require $\triangle_{t} > 0, \triangle_{x} = c\triangle_{t}, p, q \bmod 2 = 0, n, m \bmod 2 = 1$ 
\State $\mathbf{x} \gets seq(0, n*\triangle_{x}, by = \triangle_{x})$
\State $\mathbf{t} \gets seq(0, m*\triangle_{t}, by = \triangle_{t})$ \Comment{Construct the underlying rectangular grid.}
\State $\mathbf{w} \gets \triangle_{t}*seq(0, p, by = 1)$
\State $\mathbf{u} \gets \triangle_{x}*seq(-q, q, by = 1)$ \Comment{Set the coordinates of the truncated kernel matrix.}
\State $\widehat{H} \gets matrix(0, nrow = 2*q+1, ncol = p+1)$ 
\For{$i  = 1, \dots, 2*q+1$} \Comment{We fill the matrix $\hat{H}$.}
\For {$j = 1, \dots, p+1$} 
\State $k \gets i + j$ 
\If{$k \bmod 2 = 0$}
\State $\widehat{H}_{ij} \gets  h(\mathbf{u}_{i}, \mathbf{w}_{j})$
\EndIf
\EndFor
\EndFor
\State $\hat{h} \gets split(\widehat{H}, row(\widehat{H}))$ \Comment{We split $\widehat{H}$ into a list of its rows.}
\end{algorithmic}
\end{algorithm}

\begin{figure}[tbp]
\centering
\caption{The finer diamond grid embedded within the rectangular grid which is represented by the red dots. Here, $c = 2$, $m = 10$ and $n = 5$.}
\label{fig:finerdia}
\includegraphics[width = 2.8in, height = 2.4in, trim = 0.4in 0.5in 0.4in 0.8in]{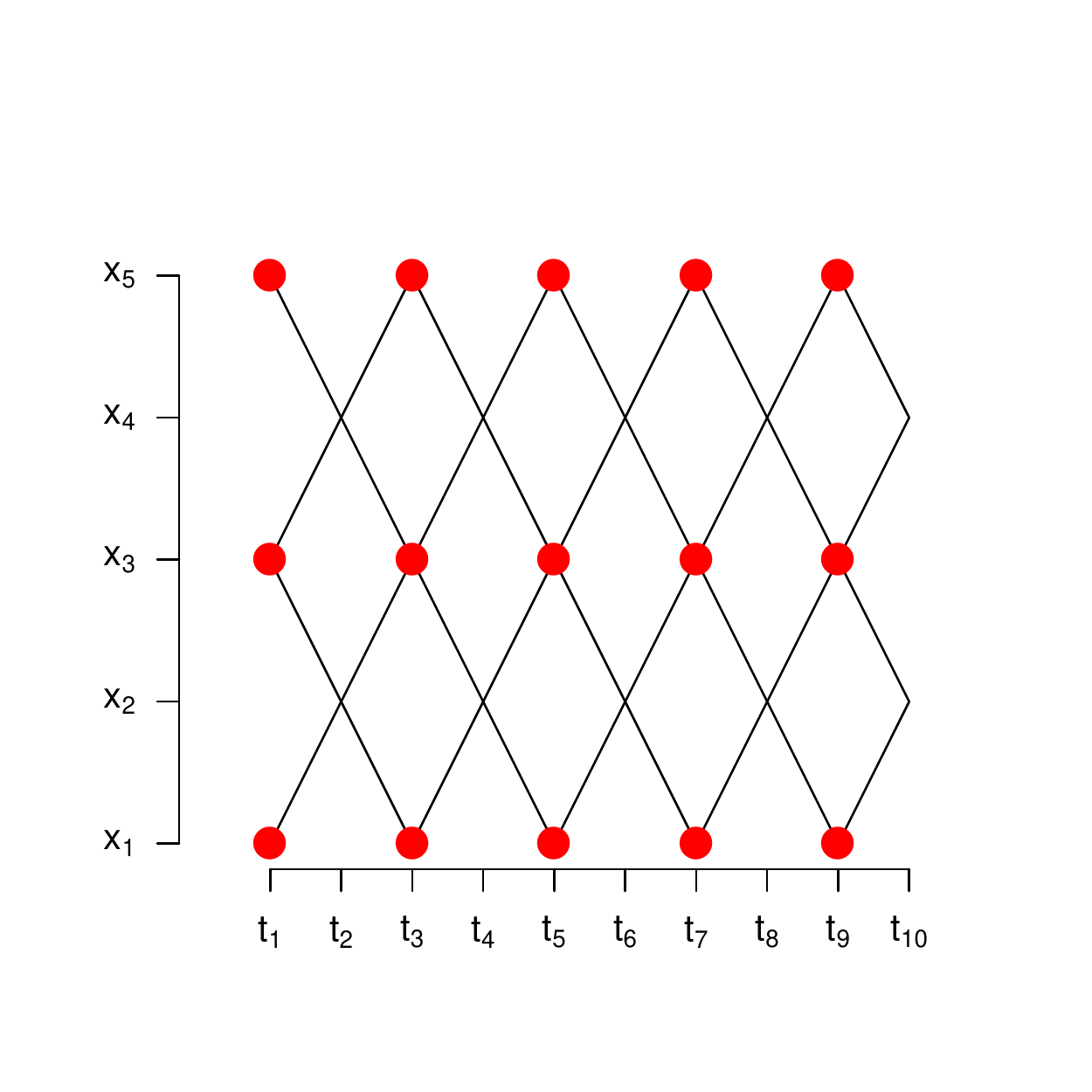}
\end{figure} 

\clearpage

for $k = 1, ..., \tilde{n}$ and $l = 1, ..., \tilde{m}$. Denote the realisations by $\hat{w}_{kl}$. We create a list of row vectors of these values:
\begin{align*}
\widehat{W} = \{ \left[\hat{w}_{11}, \hat{w}_{12}, \dots ,  \hat{w}_{1\tilde{m}}\right], \dots, \left[\hat{w}_{\tilde{n}1}, \hat{w}_{\tilde{n}2}, \dots ,  \hat{w}_{\tilde{n}\tilde{m}}\right]\}.
\end{align*}
All that is required now is to perform the DC of $\hat{h}$ and $\widehat{W}$ by row. The steps are identical to that in the RG algorithm (Algorithm \ref{alg:DCrect}). Finally, with the data matrix $Y$, we label $Y_{IJ}$ where $I + J$ is odd, as missing values.
\\
Instead of having missing values at every alternate space-time point, we can also obtain field values on a RG by simulating on a finer DG but later discarding data on alternate rows and columns. Figure \ref{fig:finerdia} shows an example for $c = 2$. 
\\
\begin{Rem}
\label{re:asim}
Using the DC approach to simulation incorporates a moving window and an efficient in-built algorithm. Alternatives include a definition-based algorithm with double summations and a cut-off point or an algorithm based on the Markov property. However, these involve quadratically or linearly increasing number of operations.
\end{Rem}

\subsection{Mean squared errors}

As they give discrete approximations to our continuous field, both the RG and DG algorithms involve simulation error.
\\
\begin{thm}
\label{thm:simerror}
Let $Y_{t}(x)$ be a $g$-class $\mathrm{OU}_{\wedge}$ process, and $Z_{t}(x)$ be its DC approximation provided by the RG algorithm. Writing $L'$ for the L\'evy seed of $L$ and $A_{t}(x)$ for the relevant ambit set (in the definition of $Y$), we have that:
\begin{align}
&\mathbb{E}\left[|Y_{t}(x) - Z_{t}(x)|^{2}\right] \nonumber \\
=& \left(\int_{\mathbb{R}\times\mathbb{R}}\left[k(x, t, \xi, s) - h(x, t, \xi, s)\right]\mathbb{E}\left[L'\right]\mathrm{d}\xi\mathrm{d}s\right)^{2} + \int_{\mathbb{R}\times\mathbb{R}}\left(k(x, t, \xi, s)-h(x, t, \xi, s)\right)^{2}\Var\left[L'\right]\mathrm{d}\xi\mathrm{d}s , \label{eqn:MSE} 
\end{align}
where $k(x, t, \xi, s) = \mathbf{1}_{A_{t}(x)}(\xi, s)\exp\left(-\lambda\left(t-s\right)\right)$ and:
\begin{equation*}
h(x, t, \xi, s) = \sum_{j = 0}^{p}\sum_{i = -q}^{q} \mathbf{1}_{|i\triangle_{x}|\leq g(j\triangle_{t})}\mathbf{1}_{\left(t - (j+1)\triangle_{t}, t - j\triangle_{t}\right]}(s) \mathbf{1}_{\left[x + i\triangle_{x}- \frac{\triangle_{x}}{2}, x + i\triangle_{x} + \frac{\triangle_{x}}{2}\right)}(\xi) \exp\left(-\lambda j \triangle_{t}\right).
\end{equation*} 
\end{thm}
We note that the proof for Theorem \ref{thm:simerror} (in the Appendix) works for any approximation of an $\mathrm{OU}_{\wedge}$ process. We only need to be able to formulate $h(\mathrm{x}, t, \xi, s)$ appropriately. Thus, in principle, one can find the mean squared errors (MSEs) for RG algorithms simulating general $\mathrm{OU}_{\wedge}$ processes and DG algorithms simulating canonical $\mathrm{OU}_{\wedge}$ processes. We present the MSE for the latter case:
\\
\begin{thm}
\label{thm:simerrorDG}
Let $Y_{t}(x)$ be a canonical $\mathrm{OU}_{\wedge}$ process with shape parameter $c$, and $Z_{t}(x)$ be its DC approximation provided by the DG algorithm. We have that (\ref{eqn:MSE}) applies with:
\begin{align*}
 h(x, t, \xi, s) &= \sum_{j = 0}^{p}\sum_{i = 0}^{j}\exp\left(-\lambda j \triangle_{t}\right)\left[ \mathbf{1}_{\left(t - (j+1)\triangle_{t}, t - j\triangle_{t}\right]}(s) \mathbf{1}_{\left[x + c(2i-j)\triangle_{t} -c(t-j\triangle_{t}-s), x + c(2i-j)\triangle_{t} + c(t-j\triangle_{t}-s)\right)}(\xi) \right. \\
&\left. + \mathbf{1}_{\left(t - (j+2)\triangle_{t}, t - (j+1)\triangle_{t}\right]}(s)\mathbf{1}_{\left[x + c(2i-j-1)\triangle_{t} + c(t-j\triangle_{t}-s), x + c(2i-j+1)\triangle_{t} -c(t-j\triangle_{t}-s)\right)}(\xi) \right].
\end{align*}
\end{thm}

The MSEs in Theorems \ref{thm:simerror} and \ref{thm:simerrorDG} are affected by two different asymptotics: $\triangle_{t}, \triangle_{x} \rightarrow 0$ and $R_{t} = p\triangle_{t}, R_{x} = q\triangle_{x} \rightarrow \infty$. We illustrate this in the following example, whose computational details can be found in Section 3.2 of the supplementary material provided in \cite{NV2015}: 

\clearpage

\begin{figure}[htbp]
\centering
\caption{(i) MSEs of the canonical $\mathrm{OU}_{\wedge}$ simulations, as a function of the grid size for fixed kernel truncation bound $R = 15$ (in black: the RG algorithm with rectangle width $\triangle$; in red: the DG algorithm with diamond width $2\triangle$; and in blue: the DG algorithm with diamond width $\triangle$); (ii) MSEs as a function of $R$ for fixed grid size $0.05$; (iii) MSEs as a function of the grid size for $R = 0.05/\triangle$;(iv) the corresponding three simulation grids and the definition of "grid size" in each case (in black: the RG; in dotted red: the DG; and in blue: the finer DG). Here, $\lambda = c = 1$, and the mean and standard deviation of the L\'evy seed are $0.2$ and $0.1$ respectively. The coloured dotted lines refer to the respective non-zero asymptotic limits when $\triangle \rightarrow 0$ in Plot (i) and $R \rightarrow \infty$ in Plot (ii). Since the algorithms share the same limit for the first case, only one limit is shown.}
\label{fig:MSECompare}
\includegraphics[width = 2.7in, height = 3.2in, trim = 0in 0in 0in 0in]{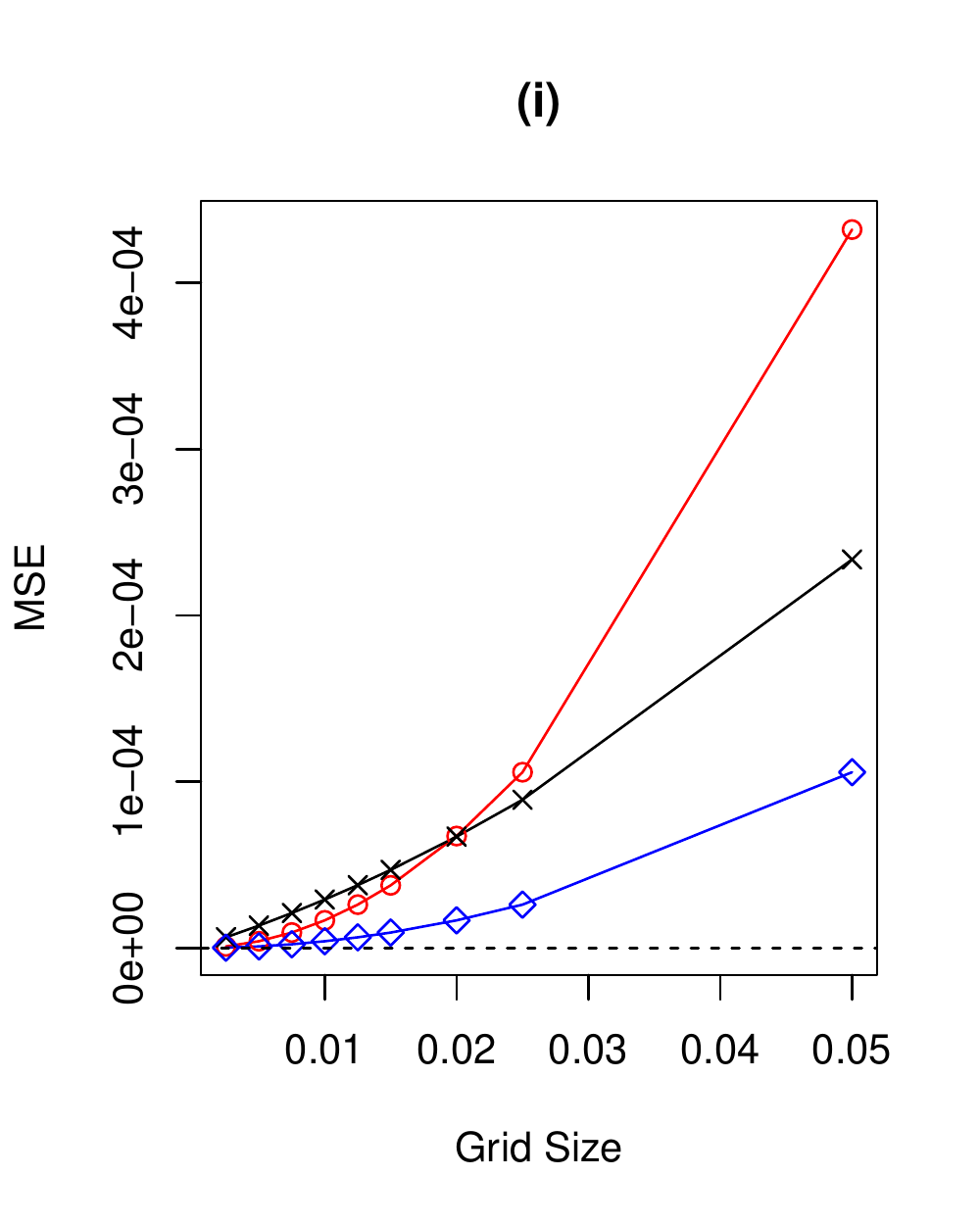}
\includegraphics[width = 2.7in, height = 3.2in, trim = 0in 0in 0in 0in]{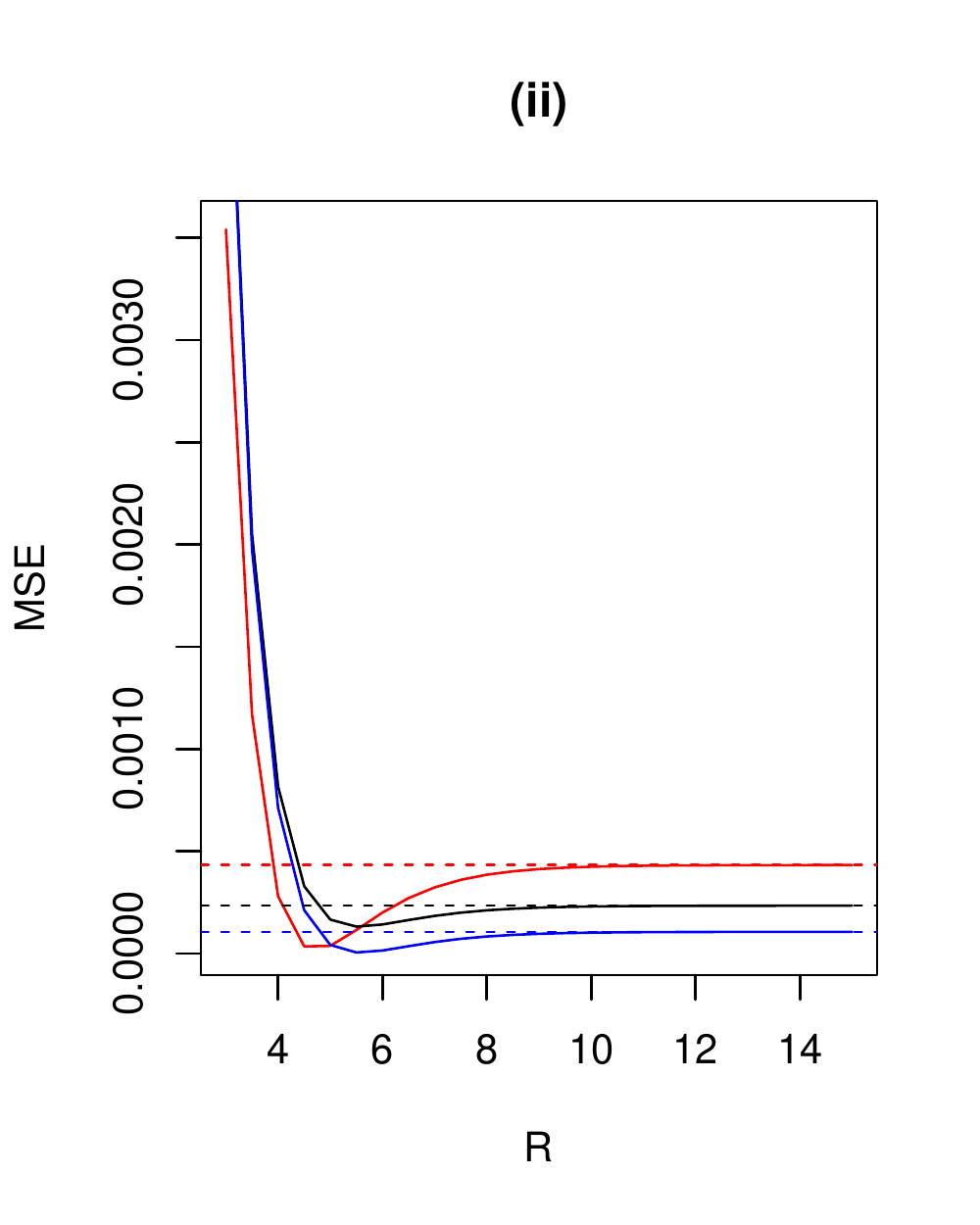}
\includegraphics[width = 2.7in, height = 3.2in, trim = 0in 0in 0in 0in]{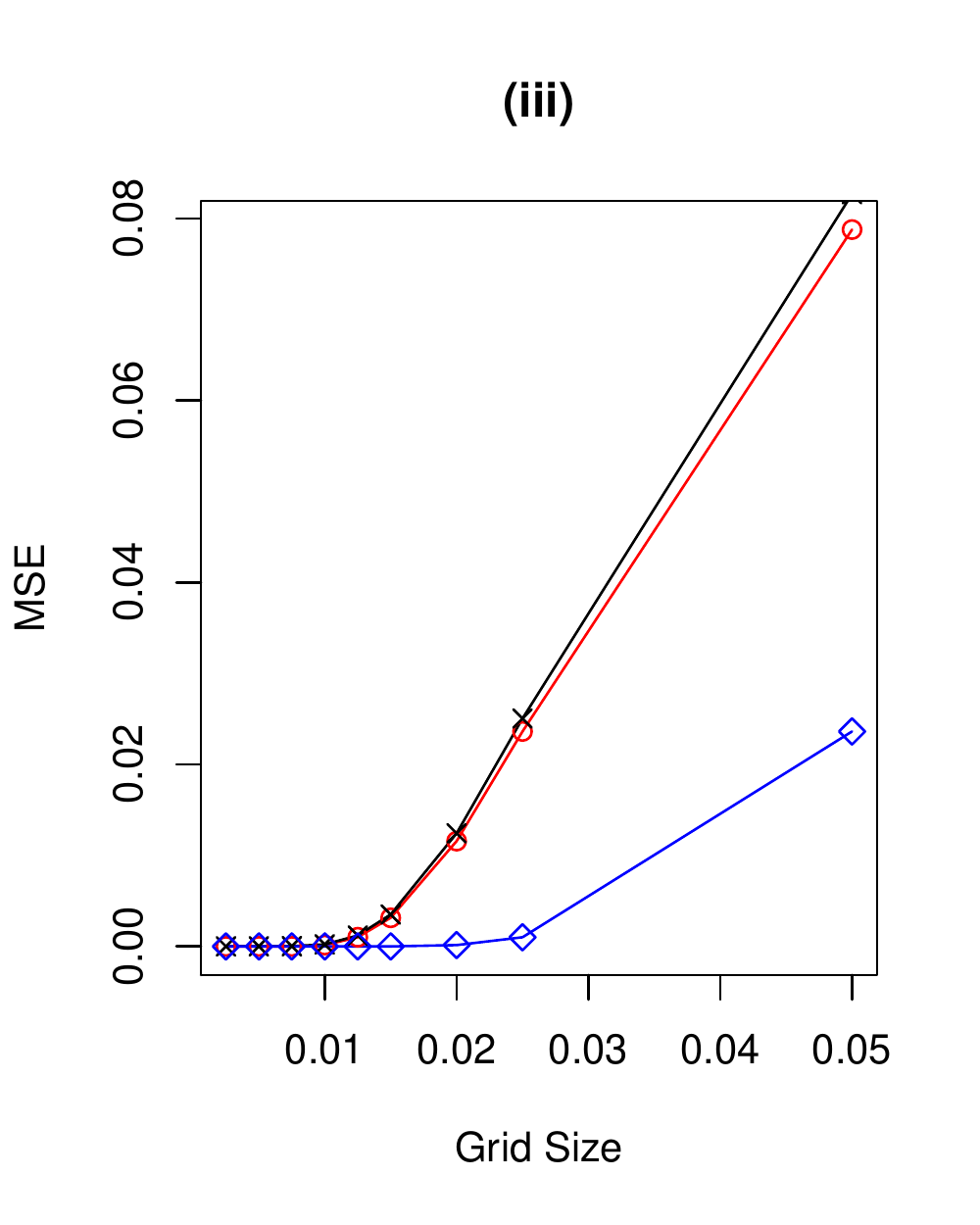}
\includegraphics[width = 2.7in, height = 3.2in, trim = 0in 0in 0in 0in]{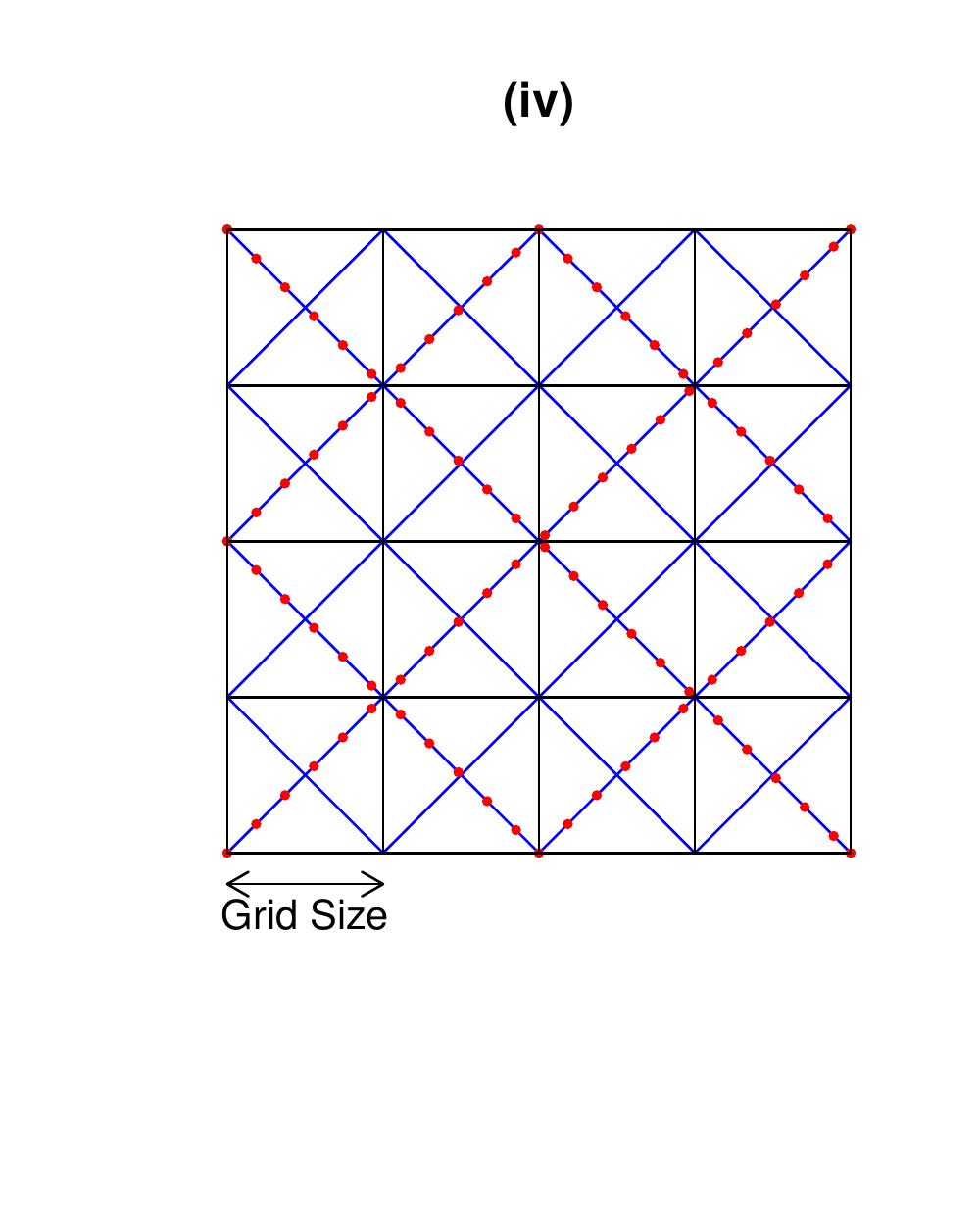}
\end{figure} 

\begin{Eg} \label{eg:MSE}
Let $Y_{t}(x)$ be the canonical process with $c = 1$, and let the mean and variance of the L\'evy seed be $\mu$ and $\tau^2$ respectively. Fix $p = q$, $\triangle_{t} = \triangle_{x} = \triangle$ and $R = p\triangle$. By using series identities and expansions, we find that for the RG algorithm: 
\begin{align*}
& \mathbb{E}\left[\left|Y_{t}(x)-Z_{t}(x)\right|^{2}\right] = \left[\frac{4\mu^{2}}{\lambda^{4}}(1+\lambda R)^{2}+\frac{\tau^{2}}{2\lambda^{2}}(1+2\lambda R)\right]e^{-2R\lambda} + O(\triangle) \text{ for fixed } R, \\
&\text{and } \mathbb{E}\left[\left|Y_{t}(x)-Z_{t}(x)\right|^{2}\right] \stackrel{R\rightarrow \infty}{\longrightarrow}\frac{4\mu^{2}}{\lambda^{4}} \left(1- \left[ \frac{\lambda^{2}\triangle^{2}}{\left(1 - e^{-\lambda \triangle}\right)^{2}}\right]\left[\frac{1 + e^{-\lambda\triangle} }{2}\right]\right)^{2} \\
&+\tau^{2} \left[\frac{1}{2\lambda^{2}}\left(1 - \left(\frac{2\left(1 - e^{-\lambda\triangle}\right)}{\lambda\triangle}-1\right)\frac{4\lambda^{2}\triangle^{2} e^{-2\lambda \triangle}}{\left(1 - e^{-2\lambda \triangle}\right)^{2}} \right) + \left( \frac{e^{-\lambda\triangle}}{\lambda^{2}}+ \frac{\triangle}{2\lambda} - \frac{\left(1 - e^{-\frac{\lambda\triangle}{2}}\right)}{\lambda^{2}\left(\frac{\lambda\triangle}{2}\right)}\right)\frac{2\lambda\triangle}{1 - e^{-2\lambda\triangle}}\right].
\end{align*}
\\
In a similar way, we find that the DG MSE and its limits are:
\begin{align*}
\mathbb{E}\left[\left|Y_{t}(x)-Z_{t}(x)\right|^{2}\right] &= \left[\frac{4\mu^{2}}{\lambda^{4}}(1+R)^{2}\lambda^{2}+\frac{\tau^{2}}{2\lambda^{2}}(1 + 2\lambda R)\right]e^{-2R\lambda} + O(\triangle) \text{ for fixed } R, \text{ and}\\
\mathbb{E}\left[\left|Y_{t}(x)-Z_{t}(x)\right|^{2}\right] &\stackrel{R\rightarrow \infty}{\longrightarrow} \frac{4\mu^{2}}{\lambda^{4}} \left(1- \left[ \frac{\lambda^{2}\triangle^{2}}{\left(1 - e^{-\lambda \triangle}\right)^{2}}\right]\right)^{2} +\frac{\tau^{2}}{2\lambda^{2}} \left[1 - \left\{\frac{2}{\lambda^{2}\triangle^{2}}\left(1 - e^{-\lambda\triangle}\right)^{2}-1\right\}\left[\frac{4\lambda^{2}\triangle^{2}}{\left(1 - e^{-2\lambda \triangle}\right)^{2}}\right]\right].
\end{align*}
When $\lambda = 1$, the MSEs of the RG and DG algorithms converge to the same limit with the same order $\triangle$ when $\triangle\rightarrow 0$ (here, the Big O notation has been used). Due to the infinite series in $\triangle$ and $R$, it is difficult to compare the MSEs analytically. Instead, we compute them numerically and plot the results for two scenarios: decreasing grid size with $R = 15$ and increasing $R$ with grid size equal to $0.05$. For both cases, we set $\mu = 0.2$ and $\tau = 0.1$. 
\\
From Figure \ref{fig:MSECompare}(i), we see that the MSEs of both the RG and the DG algorithms decrease to their common limit when the grid size decreases. So far, $\triangle$ for the DG algorithm refers to the grid size of the underlying RG. If one defines the DG grid size as $\triangle$, the diamonds would have width $2\triangle$ (see the red dotted lines in Figure \ref{fig:MSECompare}(iv)). In this case, the MSE of the DG algorithm is larger than that of the RG algorithm for grid sizes more than $0.02$, but less than that of the RG algorithm for grid sizes less than $0.02$. It seems that, when the grid size decreases beyond $0.02$, the ambit set approximation error of the RG decreases more slowly than the kernel discretisation error of this DG. If instead, one defines the DG grid size to be the distance to the nearest spatial or temporal neighbour, the DG with the same grid size as the RG in Figure \ref{fig:MSECompare}(iv) would be the one in blue. In this case, the MSE of the DG algorithm is consistently lower than that of the RG algorithm as can be seen from Figure \ref{fig:MSECompare}(i).
\\
Figure \ref{fig:MSECompare}(ii) shows the behaviour of the MSEs when $R$ increases from $3$ to $15$. Instead of a monotonic decrease as $R$ increases, we observe that the MSEs for both the RG and the DG algorithms drop sharply before slowly rising to their asymptotic limits (denoted by the dotted lines in their corresponding colours). This illustrates the interaction between kernel truncation and kernel discretisation errors. When we fix the grid size but increase $R$, kernel discretisation error increases as more terms are involved but kernel truncation error decreases. Both changes occur at decreasing rates due to the exponential decay of our kernel function. Eventually, these changes become insignificant and MSEs level out when $R$ is large enough (around $10$ here). 
\\
Based on our computations, the RG and DG MSEs converge to zero only if $\triangle\rightarrow 0$ and $R\rightarrow \infty$. If $p =\left[1/\triangle^{2}\right]$ so that $R\propto 1/\triangle$ and $ R \rightarrow \infty$ as $\triangle\rightarrow 0$, the MSEs in both cases converge to zero at the order $\triangle$. (Equivalently, we can fix $\triangle\propto 1/R$ and let $R\rightarrow \infty$.) This can be seen by replacing the $R$s in the MSE expressions by $K/\triangle$ for some constant $K$. 
\\
In Figure \ref{fig:MSECompare}(iii), we illustrate the RG and DG MSEs for $R = 0.05/\triangle$, $\lambda = 1$, $\mu = 0.2$ and $\tau = 0.1$. In this plot of the MSEs against grid size, the MSEs of the DG algorithms are consistently lower than that of the RG algorithm. This arose from our choice of $0.05$ for the proportionality constant between $R$ and $\triangle$. If instead we choose this to be $1$, the $R$ values for $\triangle < 0.05$ would be greater than $20$. Thus, there will be an insignificant change in MSE due to the increasing $R$. Instead, the MSE behaviour will be dominated by that caused by the change in grid size and the resulting plot will look very similar to Figure \ref{fig:MSECompare}(i). 
\end{Eg}

\section{Inference}
\label{sec:siminfer}

In this section, we focus on inference techniques for the canonical $\mathrm{OU}_{\wedge}$ process $Y_{t}(x)$ defined by (\ref{eqn:TestOUh}). 

\subsection{Moment-matching inference}
\label{sec:mmest}

We consider a spatio-temporal extension of the moments-matching (MM) method in \cite{KAKE2013}. This involves matching the theoretical forms and the empirical values for the normalised spatial and temporal variograms, as well as the cumulants of $Y_{t}(x)$. The normalised spatial and temporal variograms are defined as:
\begin{align}
\bar{\gamma}^{(S)}(d_{x})  &:=  \frac{\mathbb{E}\left[\left(Y_{t}\left(x\right)-Y_{t}\left(x-d_{x}\right)\right)^{2}\right]}{\Var\left[Y_{t}\left(x\right)\right]} = 2(1-\rho^{(S)}(d_{x})) = 2\left(1-\exp\left(-\frac{\lambda d_{x}}{c}\right)\right), \label{eqn:Svario}\\
\text{and } \bar{\gamma}^{(T)}(d_{t})  &:= \frac{\mathbb{E}\left[\left(Y_{t}\left(x\right)-Y_{t-d_{t}}\left(x\right)\right)^{2}\right]}{\Var\left[Y_{t}\left(x\right)\right]}=  2(1-\rho^{(T)}(d_{t})) =  2(1-\exp(-\lambda d_{t})). \label{eqn:Tvario}
\end{align}
Suppose we have data $\{Y_{\mathbf{v}}: \mathbf{v} \in \mathcal{V}\}$ where $\mathcal{V}$ is the set of space-time observation sites. Let $N(d_{x}, d_{t})$ be the set containing all the pairs of indices of sites with spatial distance $d_{x}> 0$ and temporal distance $d_{t}> 0$, and let $\hat{\kappa}_{2}(Y_{\mathbf{v}})$ denote the empirical variance in our observations. We can estimate $\bar{\gamma}^{(S)}(d_{x})$ and $\bar{\gamma}^{(T)}(d_{t})$ by:
\begin{align}
\hat{\bar{\gamma}}^{(S)}(d_{x}) &= \frac{1}{|N(d_{x}, 0)|} \sum_{(i, j) \in N(d_{x}, 0)} \frac{(Y_{\mathbf{v}_{i}} - Y_{\mathbf{v}_{j}})^{2}}{\hat{\kappa}_{2}(Y_{\mathbf{v}})}, \text{ and }\hat{\bar{\gamma}}^{(T)}(d_{t}) = \frac{1}{|N(0, d_{t})|} \sum_{(i, j) \in N(0, d_{t})} \frac{(Y_{\mathbf{v}_{i}} - Y_{\mathbf{v}_{j}})^{2}}{\hat{\kappa}_{2}(Y_{\mathbf{v}})}. \label{eqn:ENVxt}
\end{align}
By matching the empirical and the theoretical forms, we can estimate $\lambda$ and $c$ by $\hat{\lambda} = -d_{t}^{-1}\log\left(1 - \frac{\hat{\bar{\gamma}}^{(T)}(d_{t})}{2}\right)$, and $\hat{c}  =  -\hat{\lambda} d_{x}/\log\left(1 - \hat{\bar{\gamma}}^{(S)}(d_{x})/2\right)$. 
\\
Having found $\hat{\lambda}$ and $\hat{c}$, we use the cumulants of $Y_{t}(x)$ to estimate the parameters of the L\'evy basis $L$. For $l = 1, 2, 3,$ and $4$, we can estimate $\kappa_{l}\left(Y_{t}\left(x\right)\right)$ non-parametrically by (see for example, page 391 of \cite{KS1976}):
\begin{align*}
&\hat{\kappa}_{1}(Y_{t}(x)) = \frac{1}{D}S_{1}, \text{ }\hat{\kappa}_{2}(Y_{t}(x)) = \frac{1}{D(D-1)}(DS_{2} - S_{1}^{2}),\text{ } \hat{\kappa}_{3}(Y_{t}(x)) = \frac{1}{D(D-1)(D-2)}(D^{2}S_{3} - 3DS_{2}S_{1} + 2S_{1}^{3}), \\
&\text{and } \hat{\kappa}_{4}(Y_{t}(x)) = \frac{1}{D(D-1)(D-2)(D-3)}[(D^{3}+D^{2})S_{4} - 4(D^{2}+D)S_{3}S_{1} - 3(D^{2}-D)S_{2}^{2} + 12DS_{2}S_{1}^{2} - 6S_{1}^{4}], 
\end{align*}
where $D$ denotes our sample size and $S_{k} = \sum_{i = 1}^{D} Y^{k}_{\mathbf{v}_{i}}$. Now, we can match the estimates with the theoretical representations and solve for the basis parameters. We give an example for the Gaussian basis:
\\
\begin{Eg}
Let $L$ be a Gaussian basis whose seed has mean $\mu$ and variance $\tau^{2}$. With $\lambda = \hat{\lambda}$ and $c = \hat{c}$, we obtain:
\begin{align*}
\hat{\kappa}_{1}(Y_{t}(x)) = \hat{\mu} \frac{2\hat{c}}{\hat{\lambda}^{2}} \Rightarrow \hat{\mu} = \frac{\hat{\kappa}_{1}(Y_{t}(x))\hat{\lambda}^{2}}{2\hat{c}} \text{ and } \hat{\kappa}_{2}(Y_{t}(x)) &= \hat{\tau}^{2} \frac{\hat{c}}{2\hat{\lambda}^{2}} \Rightarrow  \hat{\tau} = \sqrt{\frac{2\hat{\kappa}_{2}(Y_{t}(x))\hat{\lambda}^{2}}{\hat{c}}}.
\end{align*}
\end{Eg}

\subsection{A least-squares adaptation} \label{sec:LSadapt}

The MM method introduced in the previous subsection involved using the normalised variograms at single lags to find $\hat{\lambda}$ and $\hat{c}$. Instead, we can compute (\ref{eqn:Svario}) and (\ref{eqn:Tvario}) at several lags and use least-squares (LS) to fit the values to the theoretical curves. By fitting the empirical temporal variogram to its theoretical form in (\ref{eqn:Tvario}), we can find $\hat{\lambda}$. With this, we can find $\hat{c}$ by fitting the empirical spatial variogram to the expression in (\ref{eqn:Svario}). We can then proceed as in the MM method to find the L\'evy basis parameters. This approach is useful for the cases where we have more than two parameters in the normalised variograms such as case (ii) in Example \ref{eg:Scorr}.
\\
The LS methodology can be considered as a variant of the Generalised Method of Moments (GMM). These two methods are similar in that they minimise error functions involving ``moment'' conditions. However, there is a key difference. In the LS method, we used normalised variograms to separate the effects of the L\'evy basis parameters and the autocorrelation parameters ($\lambda$ and $c$). We first estimated $\lambda$ and $c$ before plugging our estimates into the expressions for the cumulants of our process to estimate the basis parameters. In a GMM framework, however, we cannot use normalised variograms or cumulants as we have to work with ergodic averages of functions computed from our data. This means that we cannot separate the effects of the basis and autocorrelation parameters. Instead, we have to combine all the ``moment" conditions involving variograms and moments of our data into one optimisation criterion and find all the estimates simultaneously. High computational effort is required to navigate through the high dimensional surface of this optimisation criterion. More work is required to see how a GMM approach can be applied to $\mathrm{OU}_{\wedge}$ processes.

\begin{figure}[tbp]
\centering
\caption{Theoretical and empirical curves of: (i) the normalised temporal variogram; (ii) and the normalised spatial variogram. The black curves denote the theoretical curves, while the red and blue curves denote the moments-matching fit and the least-squares fit using $15$ lags respectively. Each lag corresponds to $0.05$ units.}
\label{fig:NVLS}
\includegraphics[width = 4in, height = 2.3in, trim = 0.1in 0.5in 0.1in 0.1in]{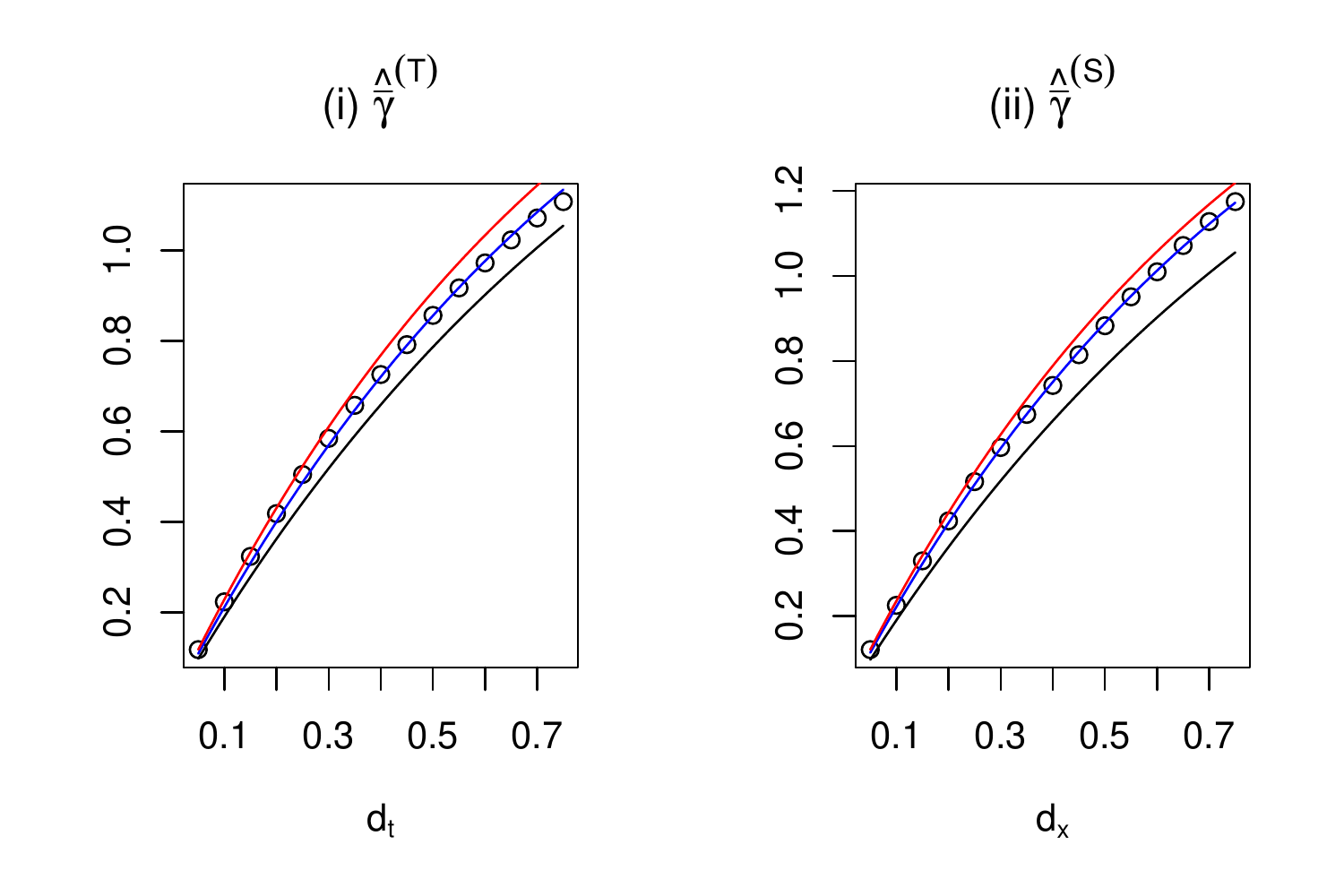}
\end{figure}

\subsection{Properties of our estimators}

\begin{thm} \label{thm:consist}
Suppose that we have lattice data, $\{Y_{\mathbf{v}_{i}} :\mathbf{v}_{i} = (x_{j}, t_{k}), k = 1, \dots, m, j = 1, \dots, n\}$, for a canonical $\mathrm{OU}_{\wedge}$ process. Then, the MM and LS estimators for the model parameters are consistent under increasing domain asymptotics. 
\end{thm}

The proof of this asymptotic property can be found in the Appendix. To investigate properties of our estimators under finite samples, we can conduct simulation experiments. For illustrative purposes, we assume a Gaussian basis and fix $\lambda = c = 1$. We use the RG algorithm with $\triangle_{t} = \triangle_{x} = 0.05$ and $p = q = 300$ to generate $500$ data sets covering the $[0, 10] \times [0, 10]$ space-time region. Similarly, we use the DG algorithm (with diamond width $2\triangle_{t}$) to generate $500$ data sets. For both cases, each dataset took about $27$ seconds to create on a PC with characteristics: Intel$^{\circledR}$ Core\texttrademark i7-3770 CPU Processor @ 3.40GHz; 8GB of RAM; Windows 8.1 64-bit . 
\\
We obtain the MM and LS estimates of the $500$ RG and DG datasets. For the LS procedure, we consider $15$ time and space lags for the normalised variograms, each lag being $0.05$ units. For each lag, these were calculated using (\ref{eqn:ENVxt}). Figure \ref{fig:NVLS} shows the LS fit of the theoretical forms, (\ref{eqn:Svario}) and (\ref{eqn:Tvario}), to the empirical values in blue for one RG dataset. From the fitted normalised temporal variogram, we obtain $\hat{\lambda} = 1.037$; from the fitted normalised spatial variogram, we obtain $\hat{\lambda}/\hat{c} = 1.090$ so that $\hat{c} = 0.951$. 
\\
The results for the $500$ RG datasets are shown in the first two columns of Figure \ref{fig:BPGau1}. The red lines in the plots denote the true values of the parameters. We notice that while the MM estimates of $c$ in Plot (f) always lie below the true value $1$, the LS estimates of $c$ in Plot (e) are centred around $1$. This can be explained as follows: unlike the MM method which only fits the normalised variogram at the first lag, the LS method fits $15$ lags. This has the effect of pulling the estimated variograms closer to the true curves in most cases such as that in Figure \ref{fig:NVLS}, where the blue curve lies closer to the black than the red.
\\
This increase in accuracy, however, comes at the expense of higher standard deviations. While the MM estimates for $c$ in Plot (f) have a range of about $0.03$, the LS estimates have a range of about $0.5$. Such an increase in standard deviation is also seen for the DG data sets whose results are shown in the last two columns of Figure \ref{fig:BPGau1}. While the MM estimates for $c$ in Plot (h) have a range of about $0.07$, the LS estimates in Plot (g) have a range of about $0.9$. As a result of using the estimated normalised variogram at multiple lags, there are 
more sources of variation in the LS method. In addition, the estimated normalised variograms at higher lags incur larger standard deviations due to the smaller proportion of data used for their calculation.
\\
We have seen that despite their inherent errors, our DC algorithms are useful for revealing properties of our MM and LS estimators. In the next subsection, we use our inference methods to tell us more about our simulation algorithms.
 
\begin{figure}[tbp]
\centering
\caption{Box plots of the parameter estimates from the LS and MM methods under the RG and DG algorithms in the case $c = 1$ and $L$ being a Gaussian basis. The true values are denoted by the red lines.}
\label{fig:BPGau1}
\includegraphics[width = 4.5in, height = 7in, trim = 0.4in 0.5in 0.4in 0in]{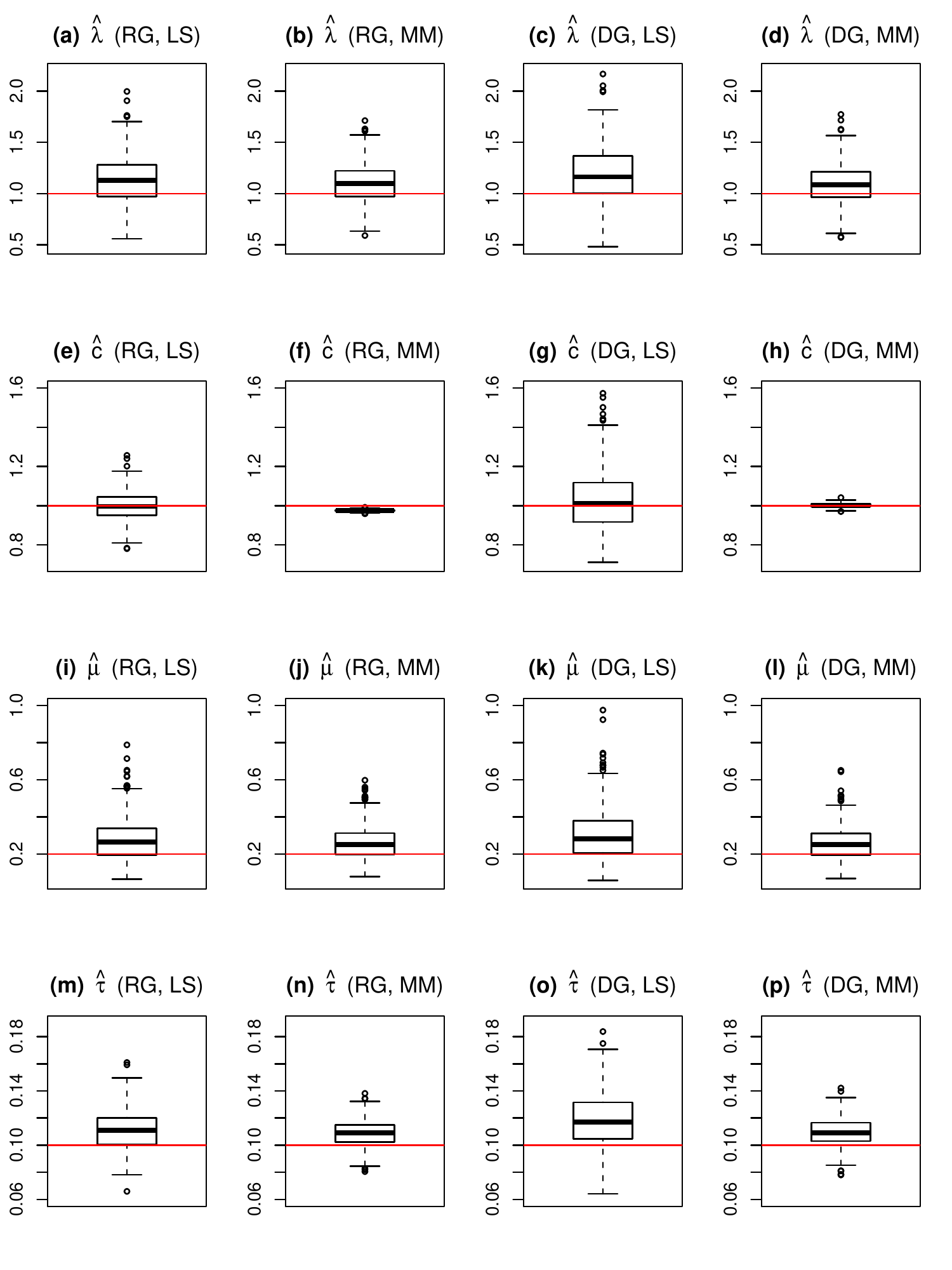}
\end{figure}

\subsection{Investigating ambit set approximations in our algorithms}

For fixed $\triangle_{t}$, the DG algorithm (with diamond width $2\triangle$) exhibits larger kernel discretisation error than the RG algorithm. However, by construction, it approximates the ambit sets of the canonical $\mathrm{OU}_{\wedge}$ processes better. This results in a better representation of the autocorrelations and variograms of the process. We use the MM results from the $500$ RG and DG data sets to illustrate this.
\\
From the box plots of the MM estimates of the RG data in the second column of Figure \ref{fig:BPGau1}, we see that reasonable results are obtained for all parameters except $c$. There seems to be a consistent underestimation of $c$. In comparison, reasonable results are obtained for all parameters including $c$ in the last column of Figure \ref{fig:BPGau1} where the MM method is applied to the DG data. As $\hat{c}$ is obtained directly from the estimated spatial normalised variogram, this implies that mimicking the edges of the triangular ambit sets of the canonical $\mathrm{OU}_{\wedge}$ process through DGs instead of approximating it by blocks through RGs, allows us to better capture the spatio-temporal dependencies of our process. 
\\
The better ambit set approximation of the DG algorithm, as well as the bias-variance trade-off of the LS method mentioned in the previous subsection, were also observed in cases where $c$ was set to $2$ or $0.5$ and $L$ to IG, NIG or Gamma bases. The associated box plots can be found in Section 4 of the supplementary material provided in \cite{NV2015}. 

%
%
%
%

\section{Application to radiation anomaly data} \label{sec:app}

Space-time ARMAs have been used for satellite ozone data and suggested for other global coverage data including those on cloudiness (see for example, \cite{NS1992}). Since $\mathrm{OU}_{\wedge}$ processes can be viewed as continuous extensions of space-time ARMAs and cloud cover can be inferred from outgoing longwave radiation (OLR), we apply $\mathrm{OU}_{\wedge}$ processes to the radiation anomaly data provided by the \cite{website:NOAA}.

\subsection{Fitting the canonical model} \label{sec:empill}

The spatial component $X$ of our data is longitude in degrees east ($^{\circ}$E) with $\mathcal{X} = \{1.25, 3.75, \dots, 356.25, 358.75\}$. The temporal component $T$ consists of sets of five days starting from 1-5 January 2014 to 18-22 September 2014 so that $\mathcal{T} = \{1, 2, \dots, 52, 53\}$. In total, we have $144 \times 53 = 7632$ observations of five-day average OLR anomalies averaged over the latitudes $5$ degrees south to $5$ degrees north, measured in watts per square metre. These anomalies are calculated with respect to the base period 1979-1995 and capture the changes in precipitation patterns over the time. Specifically, negative anomalies imply increased cloudiness and an enhanced likelihood of precipitation, and positive anomalies imply decreased cloudiness and a lower chance of precipitation. As this relationship between OLR anomalies and precipitation is most reliable in tropical regions, our data only covers the latitudes near the equator. A heat map of our data can be found in Figures \ref{fig:GFittedHeat}(i) and \ref{fig:NIGFittedHeat}(i). Since no deterministic trend or seasonal component was found, we proceed to fit the canonical model to the data. 
\\
Autocorrelation structures are a key feature in space-time ARMAs and $\mathrm{OU}_{\wedge}$ processes. Figures \ref{fig:LWRadTSt} and \ref{fig:LWRadTSs} show the time series and ACF plots of our data at three spatial points, and the spatial series and ACF plots at three temporal points respectively. In the ACF plots, the red curves represent the fitted ACF from the MM estimation while the blue curves represent the fitted ACF from the LS adaptation using $15$ lags. The fits seem rather similar for the temporal series in Figure \ref{fig:LWRadTSt}. Indeed, the estimates obtained for the rate parameter $\lambda$ from the MM and LS methods are quite close, $1.081$ and $1.033$ respectively. In comparison, the LS fit in Figure \ref{fig:LWRadTSs} seems to be much better than that of the MM fit. By considering more lags for the normalised spatial variogram, we have $\hat{c} = 14.819$ instead of $31.504$, a substantial difference. The LS approach may be preferred in this setting as we are aiming to fit the empirical autocorrelations well. These values of $\lambda$ and $c$ are important as they tell us how much weight to give to the information on neighbouring locations and which neighbouring space-time locations to consider.

\begin{figure}[tbp]
\centering
\caption{(i) Empirical densities of the radiation anomaly data (in black) and the simulated data sets combined (in red for the Gaussian basis and blue for the NIG basis); (ii) log densities of the same data and simulated data sets combined.}
\label{fig:DenCompare}
\includegraphics[width = 5in, height = 3.2in, trim = 0.2in 0.5in 0.2in 0in]{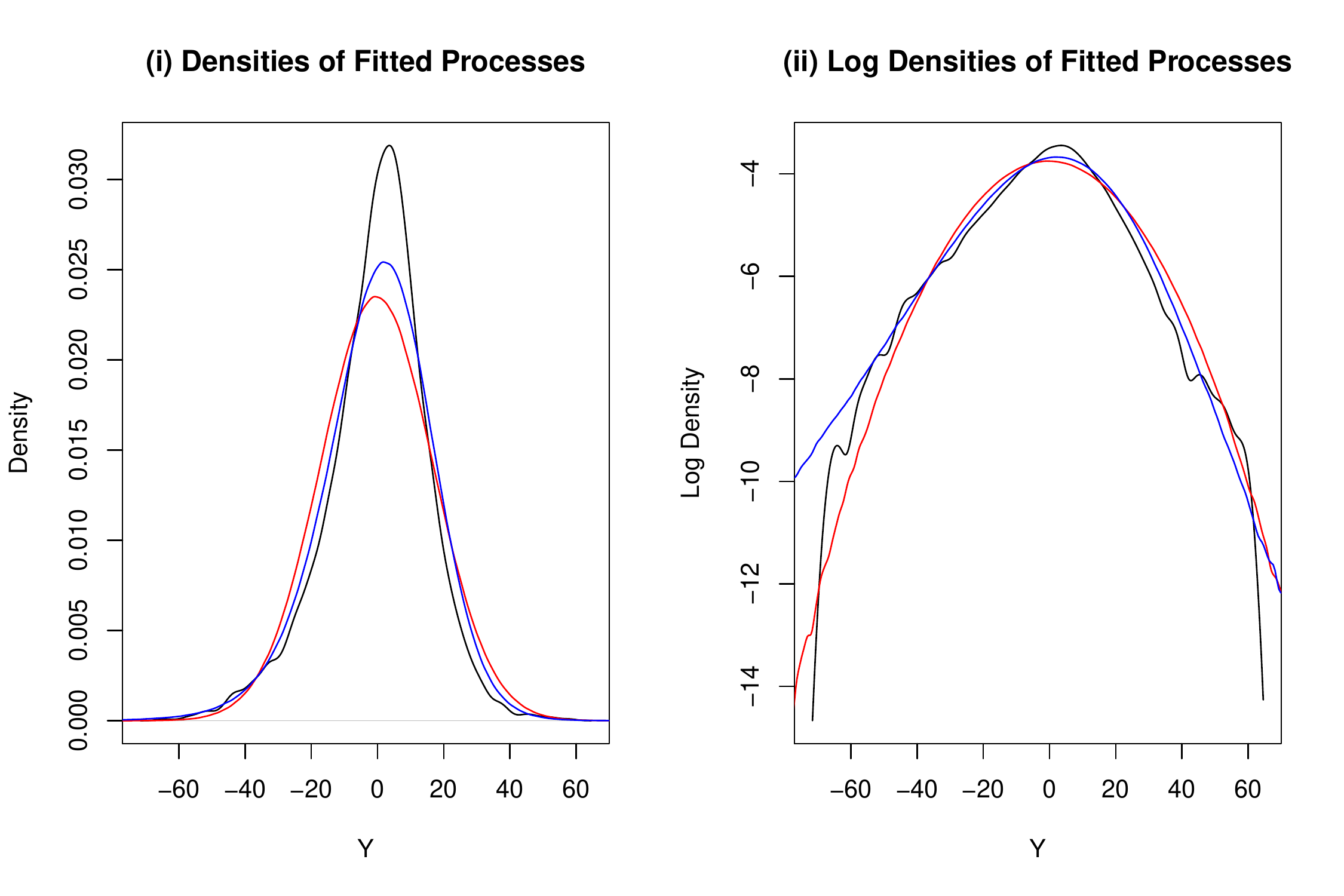}
\end{figure}

\begin{figure}[here]
\centering
\caption{Heat plot of: (i) the radiation anomaly data; (ii) the Gaussian simulated data with random seed 1.} \label{fig:GFittedHeat}
\includegraphics[width = 3in, height = 2.8in, trim = 0.4in 0.5in 0in 0in]{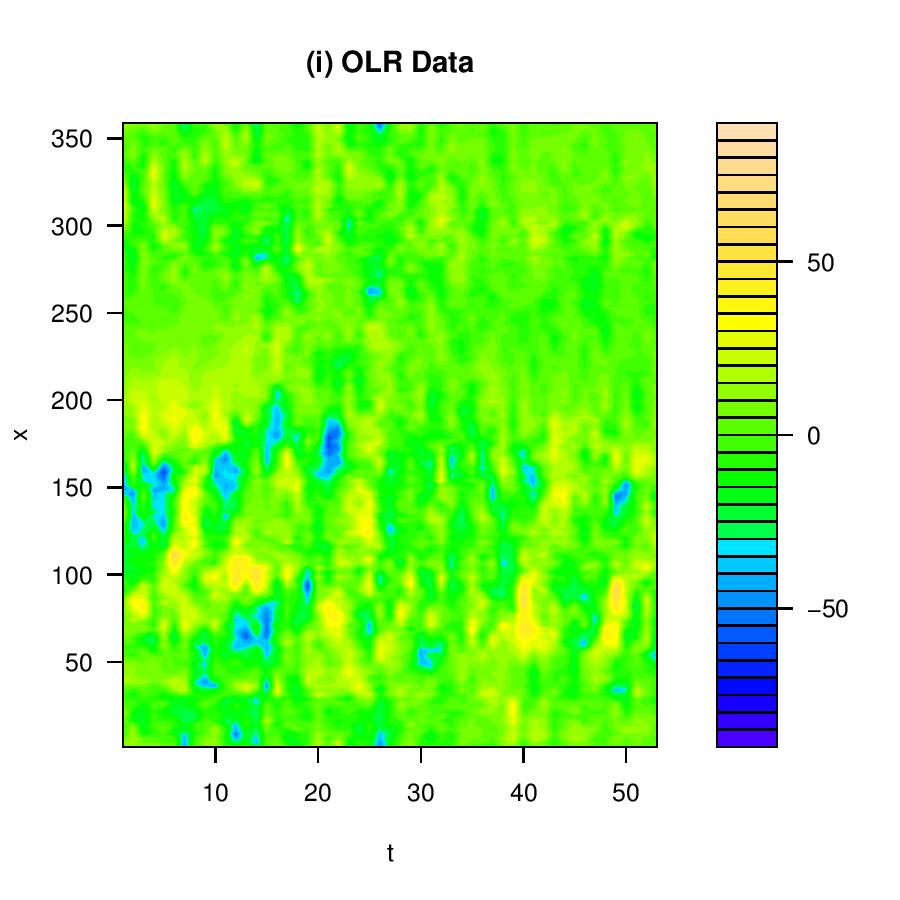}
\includegraphics[width = 3.1in, height = 2.8in, trim = 0.3in 0.5in 0in 0in]{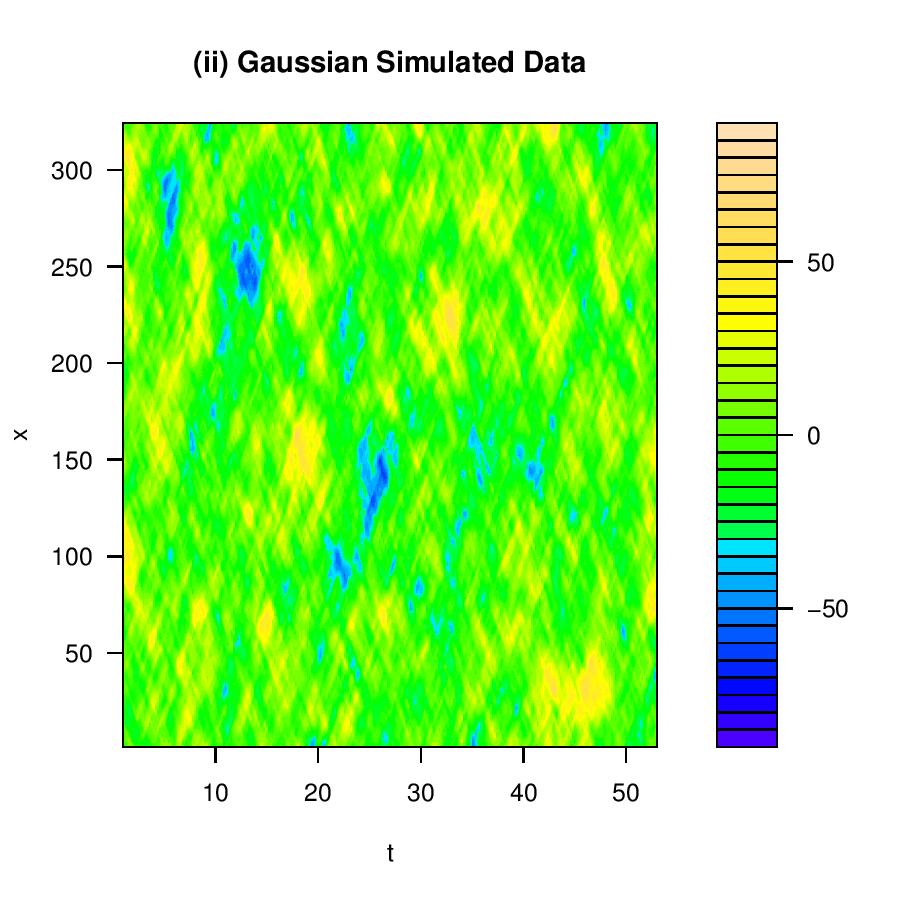}
\end{figure}

\begin{figure}[here]
\centering
\caption{Heat plot of: (i) the radiation anomaly data; (ii) the NIG simulated data with random seed 1.} \label{fig:NIGFittedHeat}
\includegraphics[width = 3in, height = 2.8in, trim = 0.4in 0.5in 0in 0in]{LWRadDataHeat.jpg}
\includegraphics[width = 3.1in, height = 2.8in, trim = 0.3in 0.5in 0in 0in]{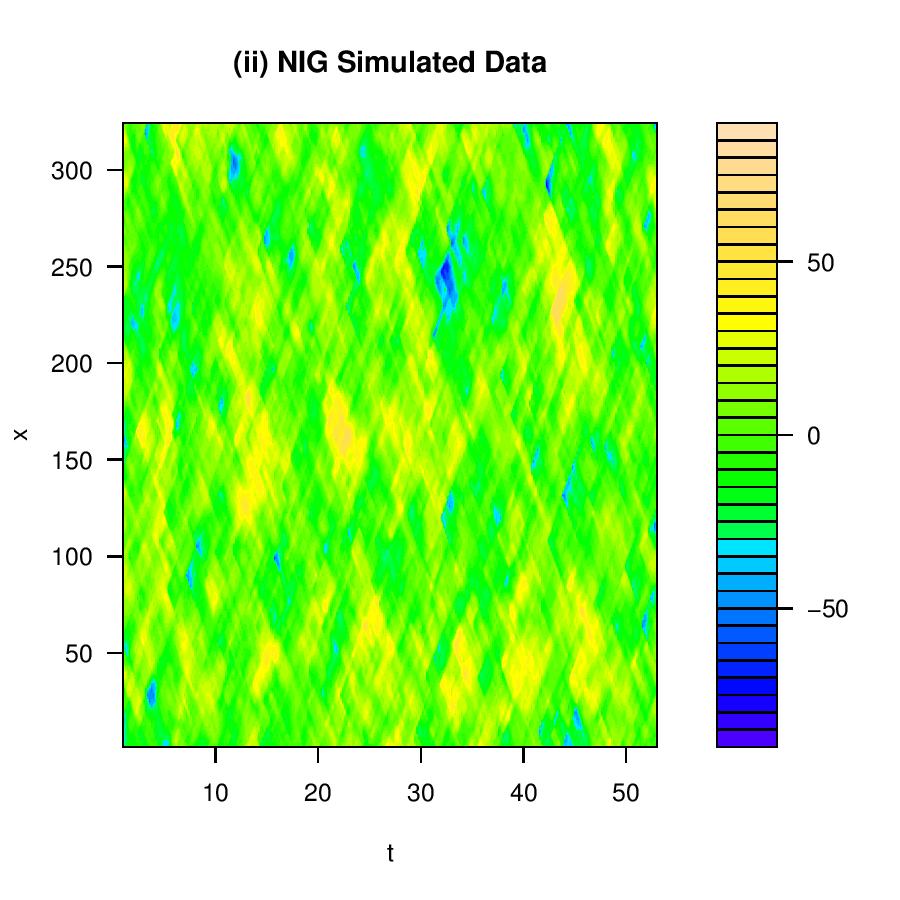}
\end{figure}

\begin{figure}[tbp]
\centering
\caption{Properties in time for fixed spatial locations of the radiation anomaly data. Time series and ACF plots for: (a)-(b) $Y_{t}(1.25)$; (c)-(d) $Y_{t}(118.75)$; and (e)-(f) $Y_{t}(238.75)$. From our definition of $T$, the time lags are the index differences between the pentads considered. In the ACF plots, the black vertical bars denote the empirical autocorrelation values, while the red and blue curves represent the ACFs arising from the moments-matching estimation and the least-squares approach respectively.}
\label{fig:LWRadTSt}
\includegraphics[width = 5.2in, height = 8.8in, trim = 0.2in 0.2in 0.2in 0in]{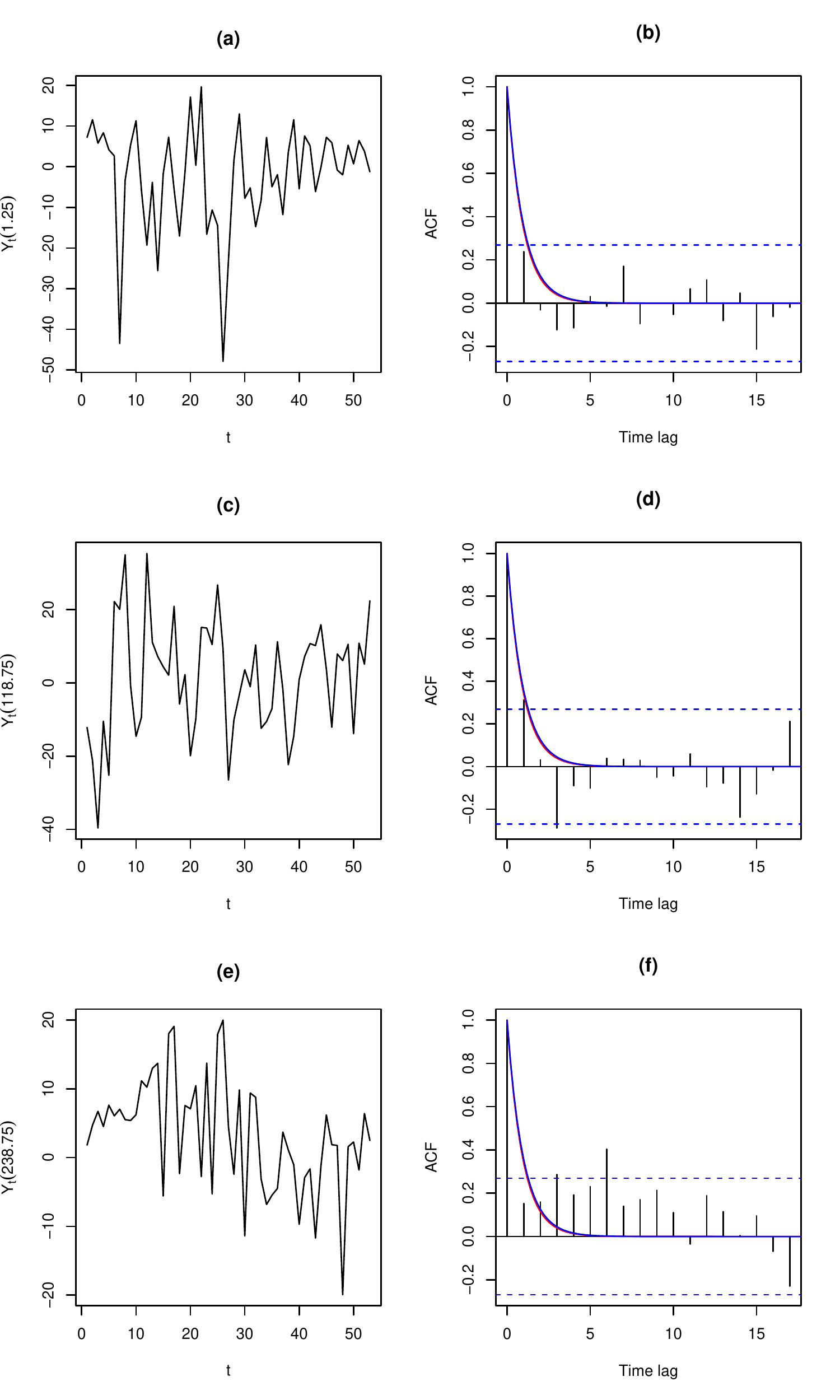}
\end{figure}

\begin{figure}[tbp]
\centering
\caption{Properties in space for fixed temporal locations of the radiation anomaly data. Spatial series and ACF plots for: (a)-(b) $Y_{1}(x)$; (c)-(d) $Y_{26}(x)$; (e)-(f) $Y_{53}(x)$. Each spatial lag is of $2.5^{\circ}$E.  In the ACF plots, the black vertical bars denote the empirical autocorrelation values, while the red and blue curves represent the ACFs arising from the moments-matching estimation and the least-squares approach respectively.}
\label{fig:LWRadTSs}
\includegraphics[width = 5.2in, height = 8.8in, trim = 0.2in 0.2in 0.2in 0in]{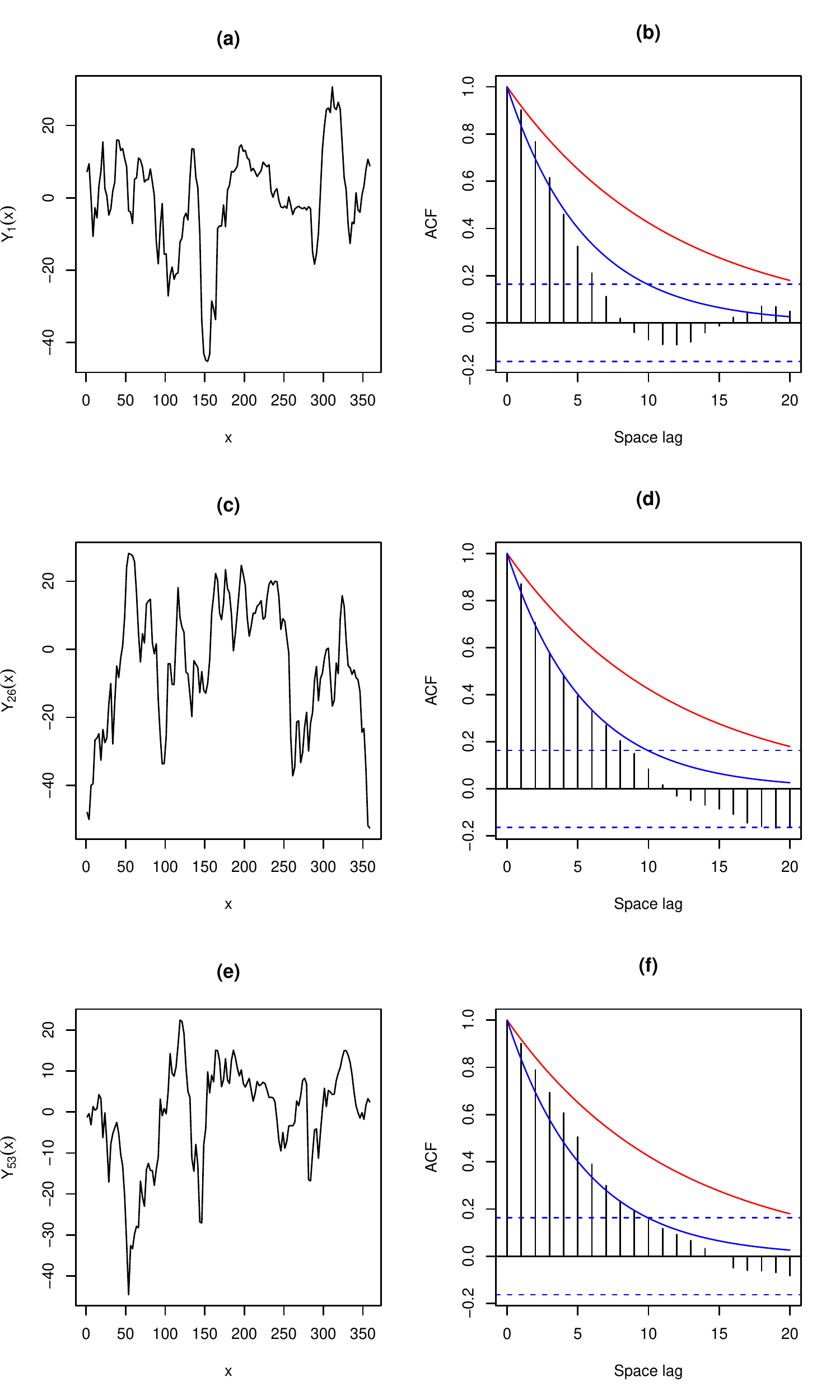}
\end{figure}

\clearpage

We continue with the LS estimates for $\lambda$ and $c$, and calculate the empirical cumulants of the L\'evy seed, $\hat{\kappa}_{l}(L')$. Since we have both positive and negative anomalies, we consider the Gaussian and NIG models. For the Gaussian basis, we have $\hat{\mu} = -0.00598$ and $\hat{\tau} = 5.827$; for the NIG basis, we have $\hat{\alpha} = 0.0765$, $\hat{\beta} = -0.0260$, $\hat{\delta} = 2.161$ and $\hat{\mu} = 0.775$. 
\\
In order to assess the goodness-of-fit of the marginal distribution, we simulate from the fitted models using the DG algorithm. We set $\triangle_{t} = 0.1$, $\triangle_{x} = \hat{c}\triangle_{t}$, $n = 219$, $m = 521$, $p = 300$ and $q = 30$. Figure \ref{fig:DenCompare} shows the densities and the logarithms of the densities computed from the data (in black), $100$ simulations of the fitted Gaussian model (in red) and $100$ simulations of the fitted NIG model (in blue). From Figure \ref{fig:DenCompare}, the model with the NIG basis seems to do a better job at capturing the peak and tail behaviour of the process. Note that comparing the extreme tail behaviour is not very informative as our OLR data does not contain many observations with extreme values.
\\
Using the DG algorithm, we can simulate from our fitted models to compare the data generated to the original. The heat plots of the simulated data from one run are placed alongside the heat plots of the OLR data in Figures \ref{fig:GFittedHeat} and \ref{fig:NIGFittedHeat}. Note that alternate coordinates of the simulated data were used to create these plots so as to avoid the missing values. Although we cannot judge the appropriateness of the fitted models based on individual runs, insight can be obtained from the comparison. As diagonal lines seem to outline diamond-like structures in the simulated data, one can say that the fitted canonical $\mathrm{OU}_{\wedge}$ processes give linear approximations to the spatio-temporal dependencies observed in the OLR phenomena. The simulated data also exhibit clusters of high and low values, a key feature of the original data. The heat plots from nine other simulation runs can be found in Section 5 of the supplementary material, \cite{NV2015}.

\subsection{Prediction under the Gaussian assumption}

In the case of the Gaussian canonical $\mathrm{OU}_{\wedge}$ process, predictive distributions at new space-time locations can be computed:
\\
\begin{thm}
Let $Y_{t}(x)$ be a Gaussian canonical $\mathrm{OU}_{\wedge}$ process with rate parameter $\lambda$, shape parameter $c$, and whose L\'evy seed $L'$ has mean and standard deviation $\mu$ and $\tau$. Suppose that we have $n$ observations at different space-time locations $\{Y_{i} = Y_{t_{i}}(x_{i}): i = 1, \dots, n\}$, and we want to predict $Y_{0} = Y_{t_{0}}(x_{0})$, the field value at a new space-time location. \\
Let $\mathbf{Y} = \left(Y_{1}, \dots, Y_{n}\right)$ and $v = (x, t)$. We define $r_{Y}(v_{0}, \mathbf{v})$ to be the $1\times m$ correlation vector between $Y(v_{0})$ and $\mathbf{Y}$ and let $R_{Y}$ to be the $m\times m$ correlation matrix for $\mathbf{Y}$ such that the covariance matrix of $(Y_{0}, \mathbf{Y})$:
\begin{equation*}
\Sigma = \frac{c\tau^{2}}{2\lambda^{2}}\begin{pmatrix} 1 & r_{Y}(v_{0}, \mathbf{v}) \\
r_{Y}(v_{0}, \mathbf{v})^{T} & R_{Y} \end{pmatrix}.
\end{equation*}
With $\mathbf{1}_{n}$ denoting the $n$-dimensional column vector of ones, the predictive distribution of $Y_{0}|\mathbf{Y}$ is $N(\mu^{*}, \Sigma^{*})$ where:
\begin{equation*}
\mu^{*} = \mathbb{E}\left[Y_{0}|\mathbf{Y}\right] = \frac{2c\mu}{\lambda^{2}} + r_{Y}(v_{0}, \mathbf{v})R_{Y}^{-1}\left(\mathbf{Y} - \frac{2c\mu}{\lambda^{2}}\mathbf{1}_{n}\right) \text{, and } \Sigma^{*} = \Var\left[Y_{0}|\mathbf{Y}\right] = \frac{c\tau^{2}}{2\lambda^{2}}\left(1 - r_{Y}(v_{0}, \mathbf{v})R_{Y}^{-1}r_{Y}(v_{0}, \mathbf{v})^{T}\right).
\end{equation*}
$\mu^{*}$ and $\Sigma^{*}$ can be seen as the best predictor and its prediction error variance respectively. 
\end{thm}

\begin{proof}
 From Example \ref{eg:GJCGF}, we have an explicit expression for the JCGF of $\left(Y_{0}, \mathbf{Y}\right)$: 
\begin{equation*}
C\{\theta \ddagger \left(Y_{0}, \mathbf{Y}\right)\}= i \frac{2c\mu}{\lambda^{2}}\sum_{i = 0}^{n}\theta_{i}   -\sum_{i, j = 0}^{n}\frac{1}{2}\tau^{2}\theta_{i}\theta_{j}\frac{c}{2\lambda^{2}} \min\left(\exp\left(-\lambda |t_{i} - t_{j}|\right), \exp\left(-\frac{\lambda|x_{i} - x_{j}|}{c}\right)\right).
\end{equation*}
This means that $\left(Y_{0}, \mathbf{Y}\right) \sim N(\tilde{\mu}, \Sigma)$ where $\tilde{\mu} = 2c\mu\lambda^{-2}\mathbf{1}_{n+1}$ and $\Sigma$ is a positive-semidefinite and symmetric covariance matrix with $\Sigma_{ij} = c\tau^{2}\left(2\lambda^{2}\right)^{-1}\min\left(\exp\left(-\lambda |t_{i} - t_{j}|\right), \exp\left(-\lambda|x_{i} - x_{j}|/c\right)\right)$. By using conditioning arguments on a multivariate normal distribution, we obtain the desired result. 
\end{proof}

In the preceding analysis, the $\mathrm{OU}_{\wedge}$ process was used directly as a spatio-temporal model. Alternatively, in line with its origins, it can be used as a model for spatio-temporal stochastic volatility in other modelling frameworks. For example, in the setting of a dynamic Cox process, it can be used, in place of the spatio-temporal $\mathrm{OU}$ process in \cite{BD2001}, to model the intensity process. Slightly modified versions of $\mathrm{OU}_{\wedge}$ processes have also been used in multiplicative models for turbulent casades \cite[]{BN2004, SBE2005, HS2013}. These exponential spatio-temporal OU models have been shown to reproduce the scaling behaviour of energy dissipation.

\section{Conclusions and Outlook}
\label{sec:outlook}

The spatio-temporal Ornstein-Uhlenbeck ($\mathrm{OU}$) process which we introduced and referred to as the $\mathrm{OU}_{\wedge}$ process, is an example of how we can extend models in time to space-time and yet retain desirable features of the original model. In Section \ref{sec:prop}, we have shown that the properties of stationarity, exponentially decaying ACFs and ergodicity transfer from the original temporal $\mathrm{OU}$ process to the $\mathrm{OU}_{\wedge}$ process. $\mathrm{OU}$ and $\mathrm{OU}_{\wedge}$ processes can be further linked through their equality in law when the spatial parameters are fixed. Time-changing, which has been done in the study of $\mathrm{OU}$ processes to remove the distributional dependence on the rate parameter $\lambda$, can also be applied to $\mathrm{OU}_{\wedge}$ processes. The additional spatial dimension of the $\mathrm{OU}_{\wedge}$ process, however, makes it more flexible than the $\mathrm{OU}$ process. Thus, there is more than one way to establish a time change. The greater flexibility of $\mathrm{OU}_{\wedge}$ processes is also seen in our ability to obtain different spatial autocorrelations by choosing different ambit sets.
\\
In the second part of our study, we tackled the problems of simulation and inference. We created two simulation algorithms for the canonical $\mathrm{OU}_{\wedge}$ process: one generating values on a rectangular grid (RG) and the other generating values on a diamond grid (DG). Concurrently, using the property of stationarity and explicit forms for the autocorrelations, we developed a moments-matching (MM) estimation method for $\mathrm{OU}_{\wedge}$ processes. A least-squares (LS) extension where the normalised variograms were evaluated at more than one lag to estimate the rate and shape parameters ($\lambda$ and $c$) was also considered. 
\\
In our analysis, we showed that the mean squared errors (MSEs) of our simulation algorithms decreased with grid size as expected. However, the errors only converge to zero if the kernel truncation range simultaneously increases to infinity. It was also seen that the DG algorithm consistently incurs less MSE than the RG algorithm if we define the grid size to be the distance to the nearest spatial or temporal neighbour. The superiority of the DG algorithm in simulating the canonical process is also seen through its ability to reflect the spatio-temporal dependencies more accurately. Reasonable results for $c$, which determines the spatial autocorrelation structure, were obtained from both inference methods for the data generated on a DG. On the other hand, there was consistent underestimation of $c$ when we applied the MM method to data generated on the RG. In our study, we chose the DG algorithm to mimic the edges of the triangular ambit sets. It may be trickier to find an appropriate grid shape for more general $\mathrm{OU}_{\wedge}$ processes. 
\\
The MM and LS inference methods, which we developed, have been proven to result in consistent parameter estimates. In our simulation study, it was also shown that the LS adaptation exhibits a bias-variance trade-off. This is obvious in the case of the RG data as the $c$ estimates obtained were centred about their true values but with higher standard deviations.  
\\
After evaluating our simulation and inference techniques, we looked at how $\mathrm{OU}_{\wedge}$ processes can be used in practice. As an empirical illustration, we applied our estimation methods to radiation anomaly data from the International Research Institute for Climate and Society. Good fits of the ACFs were obtained using the LS approach. This suggests that despite the bias-variance trade-off, the LS adaptation of the MM method may be preferred over the original in an empirical context. On another note, the good fits of the ACFs highlighted the relevance of our class of $\mathrm{OU}_{\wedge}$ processes for empirical work and the use of such processes as modelling tools. Explicit formulae for the best predictor and the prediction error variance in the case of a Gaussian canonical $\mathrm{OU}_{\wedge}$ process are also provided.
\\
In this paper, we studied inference methods based on moments of the $\mathrm{OU}_{\wedge}$ process. One topic for future research would be to design other inference techniques. In Example \ref{eg:GJCGF}, we worked out the explicit form of the joint cumulant generating function for points of our canonical $\mathrm{OU}_{\wedge}$ process when the L\'evy basis is Gaussian. This can be used to identify a multivariate normal distribution and its likelihood. Thus, we can conduct maximum likelihood estimation. If no further considerations are made, the optimisation procedure would involve inverting high dimensional matrices and would be computationally expensive. It would be interesting to develop ways to bypass these operations through, for example, pair-wise likelihood or composite likelihoods. Pseudo-likelihood methods can also be considered for non-Gaussian bases.
\\
Another way forward from our research would be to study the theory, simulation and inference techniques for $\mathrm{supOU}_{\wedge}$ processes. These processes are obtained by randomising the memory parameter in the $\mathrm{OU}_{\wedge}$ process, and would allow us to model long memory instead of exponentially decaying autocorrelation.  

\paragraph{Supplementary Material}

Supplementary material for this paper can be found in \cite{NV2015}. This contains a summary of the integration theory in \cite{RR1989}, the details behind Examples \ref{eg:STcorr} and \ref{eg:MSE}, as well as additional results from our simulation and empirical studies. 

\section*{Acknowledgements}

We would like to thank the editor, the associate editor and two anonymous referees for their constructive feedback. Much gratitude also goes to Ole E. Barndorff-Nielsen for his helpful comments and Emil Hedevang for helpful discussion regarding discrete convolution codes. M. Nguyen would like to express her appreciation to Imperial College for her PhD scholarship which supported this research. A.E.D. Veraart acknowledges financial support by a Marie Curie FP7 Integration Grant within the 7th European Union Framework Programme.

\section*{Appendix: Proofs of Theorems}

\begin{proof}[Proof of Theorem \ref{thm:OUh-OU}]
We first show that the given conditions make $Y_{t}(\mathbf{x}) \stackrel{d}{=} Z_{t}$ for fixed $t$ and $\mathbf{x}$. From Theorem \ref{thm:cgfform}, we know that the LK triplet of $Y_{t}(x)$ is $(a_{Y}, b_{Y}, \nu_{Y})$ such that $a_{Y} =  \int_{A} e^{\lambda w}  \left[a- \int_{\mathbb{R}}z\mathbf{1}_{1<|z|\leq e^{-\lambda w}} \nu(\mathrm{d}z)\right]\mathrm{d}\mathbf{u}\mathrm{d}w$, $b_{Y} =  \int_{A} b e^{2\lambda w} \mathrm{d}\mathbf{u}\mathrm{d}w$ and $\nu_{Y}(\mathrm{d}y) = \int_{A}\nu(e^{-\lambda w}\mathrm{d} y)\mathrm{d}\mathbf{u}\mathrm{d}w$. Similarly, one can show that the LK triplet of $Z_{t}$ is $(a_{Z}, b_{Z}, \nu_{Z})$ where $a_{Z} =  \tilde{a}\lambda^{-1} - \int_{-\infty}^{0} e^{\lambda w} \int_{\mathbb{R}}z\mathbf{1}_{1<|z|\leq e^{-\lambda w}} \tilde{\nu}(\mathrm{d}z)\mathrm{d}w$, $b_{Z} =  \tilde{b}(2\lambda)^{-1}$ and $\nu_{Z}(\mathrm{d}y) = \int_{-\infty}^{0}\tilde{\nu}(e^{-\lambda w}\mathrm{d}y)\mathrm{d}w$. The given conditions on $\tilde{a}$, $\tilde{b}$ and $\tilde{\nu}$ imply that $a_{Y} = a_{Z}$, $b_{Y} = b_{Z}$ and $\nu_{Y} = \nu_{Z}$. Thus, $Y_{t}(\mathbf{x}) \stackrel{d}{=} Z_{t}$.
\\
Next, we prove that the chosen $\tilde{\nu}$ is a valid L\'evy measure, i.e. $\int_{\mathbb{R}}\min(1, z^{2}) \tilde{\nu}(\mathrm{d}z) < \infty$. Since $Y_{t}(\mathbf{x})$ is a well-defined $\mathrm{OU}_{\wedge}$ process:
\begin{align*}
\infty &> 2\lambda\int_{\mathbb{R}} \min(1, z^{2})\nu_{Y}(\mathrm{d}z) = 2\lambda\int_{-\infty}^{0}\int_{\mathbb{R}} \min(1, z^{2}) \tilde{\nu}(e^{-\lambda w}\mathrm{d}z) \mathrm{d}w \\
&> 2\lambda\int_{-\infty}^{0}\int_{\mathbb{R}}  \min(\exp(2\lambda w), z^{2}) \tilde{\nu}(e^{-\lambda w}\mathrm{d}z) \mathrm{d}w = \int_{\mathbb{R}} \min(1,y^{2}) \tilde{\nu}(\mathrm{d}y),
\end{align*}
where $y = \exp(-\lambda w)z$ and the first equality holds via Fubini's theorem.
\\
We now show that $Y_{t}(\mathbf{x}) \stackrel{d}{=} Z_{t}$ as a process in $t$ for fixed $\mathbf{x}$. Let $n\in\mathbb{N}$ and let $t_{1} < \dots < t_{n}$ denote arbitrary time points. Since both processes $\{Y_{t}(\mathbf{x})\}_{t\in\mathbb{R}}$ and $\{Z_{t}\}_{t\in\mathbb{R}}$ are Markov, their finite dimensional distributions can be written as:
\begin{align}
f\left(Y_{t_{1}}(\mathbf{x}), \dots, Y_{t_{n}}(\mathbf{x})\right) &= \prod\limits_{i = 2}^{n} f\left(Y_{t_{i}}(\mathbf{x})|Y_{t_{i-1}}(\mathbf{x})\right) \times f\left(Y_{t_{1}}(\mathbf{x})\right), \label{eqn:jointY}\\
\text{and }f\left(Z_{t_{1}}, \dots, Z_{t_{n}}\right) &= \prod\limits_{i = 2}^{n} f\left(Z_{t_{i}}|Z_{t_{i-1}}\right) \times f\left(Z_{t_{1}}\right). \label{eqn:jointZ}
\end{align}
Here, $f\left(Y_{t_{i}}|Y_{t_{i}}\right)$ denotes the conditional distribution of $Y_{t_{i}}$ given $Y_{t_{i-1}}$ and $f\left(Y_{t_{1}}\right)$ denotes the marginal distribution of $Y_{t_{1}}$. Analogous notation has been used for $Z_{t}$. We already know that $Y_{t_{1}}(\mathbf{x}) \stackrel{d}{=} Z_{t_{1}}$. So, to show that the functions (\ref{eqn:jointY}) and (\ref{eqn:jointZ}) are the same, we need to prove that the conditional distributions $Y_{t_{i}}|Y_{t_{i-1}}$ and  $Z_{t_{i}}|Z_{t_{i-1}}$ are the same for arbitrary $i \in \{2, \dots, n\}$.
\\
To motivate this, we work with $Y_{t_{2}}(\mathbf{x})$ and $Y_{t_{1}}(\mathbf{x})$. From Lemma \ref{Lem:Yt1t2}, $Y_{t_{2}}(\mathbf{x}) = e^{-\lambda(t_{2}-t_{1})} Y_{t_{1}}(\mathbf{x}) + \widetilde{U}(t_{1}, t_{2}, \mathbf{x})$ where  $\widetilde{U}(t_{1}, t_{2}, \mathbf{x}) = \int_{A_{t_{2}}(\mathbf{x})\backslash A_{t_{1}}(\mathbf{x})} e^{-\lambda(t_{2}-s)} L(\mathrm{d}\xi, \mathrm{d}s)$ is independent of $Y_{t_{1}}(\mathbf{x})$. Similarly, we find that $Z_{t_{2}}= e^{-\lambda(t_{2}-t_{1})} Z_{t_{1}}(\mathbf{x}) + \widetilde{V}(t_{1}, t_{2})$ where  $\widetilde{V}(t_{1}, t_{2}) = \int_{t_{1}}^{t_{2}} e^{-\lambda(t_{2}-s)} \tilde{L}(\mathrm{d}s)$ is independent of $Z_{t_{1}}$. Since $Y_{t_{2}}(\mathbf{x}) \stackrel{d}{=} Z_{t_{2}}$ and $Y_{t_{1}}(\mathbf{x}) \stackrel{d}{=} Z_{t_{1}}$, it follows that $\widetilde{U}(t_{1}, t_{2}, \mathbf{x}) \stackrel{d}{=} \widetilde{V}(t_{1}, t_{2})$ and $Y_{t_{2}}|Y_{t_{1}} \stackrel{d}{=} Z_{t_{2}}|Z_{t_{1}}$. The proof is completed by repeating the same argument for the remaining conditional distributions. 
\end{proof}

\begin{proof}[Proof of Theorem \ref{thm:tcform1}]
From Theorem \ref{thm:cgfform}, we have:
\begin{align}
C\{\theta \ddagger Y_{t}(x)\} &= \int_{-\infty}^{t}\int_{x - c|t-s|^{q}}^{x + c|t-s|^{q}} C\{\theta\exp(-\lambda(t-s)) \ddagger L'\} \lambda^{q+1} \mathrm{d}\xi \mathrm{d}s \nonumber \\
&= \int_{0}^{\infty} \frac{2cw^{q}}{\lambda^{q}} C\{\theta\exp(-w) \ddagger L'\} \lambda^{q+1} \left(\frac{1}{\lambda}\right) \mathrm{d}w \text{, by using } w = \lambda(t-s), \nonumber \\
&= \int_{0}^{\infty} 2cw^{q} C\{\theta\exp(-w) \ddagger L'\} \mathrm{d}w, \text{ which does not depend on } \lambda. \label{eqn:TTCOUcc}
\end{align}
\end{proof}

\begin{proof}[Proof of Theorem \ref{thm:tcform2}]
We have:
\begin{align}
C\{\theta \ddagger Y_{t}(x)\} &= \int_{-\infty}^{t}\int_{x - g(\lambda|t-s|)}^{x + g(\lambda|t-s|)} C\{\theta\exp(-\lambda(t-s)) \ddagger L'\} \lambda \mathrm{d}\xi \mathrm{d}s \nonumber \\
&= \int_{0}^{\infty} 2g(w) C\{\theta\exp(-w) \ddagger L'\} \lambda \left(\frac{1}{\lambda}\right) \mathrm{d}w \text{, by using } w = \lambda(t-s), \nonumber \\
&= \int_{0}^{\infty} 2g(w) C\{\theta\exp(-w) \ddagger L'\} \mathrm{d}w, \text{ which does not depend on } \lambda. \label{eqn:TTCOUcc2}
\end{align}
\end{proof}

\begin{proof}[Proof of Theorem \ref{thm:tcourelation}]
This proof is similar to that of Theorem \ref{thm:OUh-OU}. From the proofs of Theorems \ref{thm:tcform1} and \ref{thm:tcform2}, $C\{\theta \ddagger Y^{1}_{t}(x)\} = \int_{0}^{\infty}2cw^{q}C\{\theta\exp\left(-w\right) \ddagger L_{1}'\}\mathrm{d}w$ and $C\{\theta \ddagger Y^{2}_{t}(x)\} = \int_{0}^{\infty}2g(w)C\{\theta\exp\left(-w\right) \ddagger L_{2}'\}\mathrm{d}w$ respectively. The conditions for $(\tilde{a}_{1}, \tilde{b}_{1}, \tilde{\nu}_{1})$ and $(\tilde{a}_{2}, \tilde{b}_{2}, \tilde{\nu}_{2})$ are found by equating the LK triplets of $Y^{1}_{t}(x)$ and $Z^{1}_{t}$, and those of $Y^{2}_{t}(x)$ and $Z^{2}_{t}$. In both cases, the arguments in the proof of Theorem \ref{thm:OUh-OU} can be used to establish the validity of the chosen L\'evy measure and extend the equality in law for fixed $t$ to equality in law as processes. 
\end{proof} 

\begin{proof}[Proof of Theorem \ref{thm:ergodicity}]
A one-dimensional adaptation of Theorem 3.5 in \cite{FS2013} says that the mixed moving average $X_{t} := \int_{\mathbb{R}}\int_{S} f(\zeta, t-s) L(\mathrm{d}s, \mathrm{d}\zeta)$ is mixing. Here, $L$ is a real-valued L\'evy basis on $S \times \mathbb{R}$, $\zeta \in S$ is a varying parameter and $f:S\times\mathbb{R} \rightarrow \mathbb{R}$ is a measurable function. Using $f = f_{x}$ for fixed $x$:
\begin{equation*}
Y_{t}(x) = \int_{\mathbb{R}}\int_{\mathbb{R}} \mathbf{1}_{A_{t}(x)}(\xi, s)\exp(-\lambda(t-s)) L(\mathrm{d}s, \mathrm{d}\xi) = \int_{\mathbb{R}}\int_{S} f_{x}(\zeta, t-s) L(\mathrm{d}s, \mathrm{d}\zeta), 
\end{equation*}
where $S = \mathbb{R}, \text{ } \zeta = x - \xi, \text{ and } f_{x}(\zeta, t-s) = \mathbf{1}_{t-s>0}\mathbf{1}_{-g(t-s)\leq \zeta \leq g(t-s)}\exp(-\lambda(t-s))$. So, $\{Y_{t}(x)\}_{t \in \mathbb{R}}$ is mixing. As mixing implies ergodicity, the process is ergodic \cite[]{FS2013}.
\\
Since $x$ can act in place of $t$, we can also use $f = f_{t}$ for fixed $t$ to obtain:
\begin{align*}
Y_{t}(x) &= \int_{\mathbb{R}}\int_{\mathbb{R}} \mathbf{1}_{A_{t}(x)}(\xi, s)\exp(-\lambda(t-s)) L(\mathrm{d}s, \mathrm{d}\xi) = \int_{\mathbb{R}}\int_{S} f_{t}(\zeta, x-\xi) L(\mathrm{d}\xi, \mathrm{d}\zeta),
\end{align*} 
is mixing. Now, $S = \mathbb{R}$, $\zeta = t - s$, and $f_{t}(\zeta, x - \xi) = \mathbf{1}_{\zeta>0}\mathbf{1}_{-g(\zeta)\leq x-\xi \leq g(\zeta)}\exp(-\lambda \zeta)$.
\end{proof}

\begin{proof}[Proof of Theorems \ref{thm:simerror} and \ref{thm:simerrorDG}]
Using the DC algorithm, simulating from $Z$ is the same as simulating from \\
$\int_{\mathbb{R}\times\mathbb{R}} h(x, t, \xi, s) L(\mathrm{d}\xi, \mathrm{d}s)$ since $L$ is a homogeneous L\'evy basis. We write $f = k - h$. From the \cite{RR1989} theory (via the CGF), if the stochastic integral of $f$ over $\mathbb{R}\times\mathbb{R}$ is well-defined, then: \begin{equation*}
\mathbb{E}\left[\left(\int_{\mathbb{R}\times\mathbb{R}}f\mathrm{d}L\right)^{2}\right] = \left(\int_{\mathbb{R}\times\mathbb{R}}f\mathbb{E}\left[L'\right]\mathrm{d}\xi\mathrm{d}s\right)^{2} + \int_{\mathbb{R}\times\mathbb{R}}f^{2}\Var\left[L'\right]\mathrm{d}\xi\mathrm{d}s.
\end{equation*}
\end{proof}

\begin{proof}[Proof of Theorem \ref{thm:consist}]
From the ergodicity of $\{Y_{t}(x)\}_{t\in\mathbb{R}}$ and the stationarity of the process, $\hat{\kappa}_{1}(Y_{t}(x)) = D^{-1}\sum_{i = 1}^{D}Y_{\mathbf{v}_{i}}\stackrel{p}{\rightarrow} \mathbb{E}\left[Y_{t}(x)\right]= \kappa_{1}(Y_{t}(x))$ as $D\rightarrow \infty$. Similarly, we can show that all the empirical cumulants are consistent.
\\
Now, let $\{Z_{t}\}_{t\in\mathbb{R}} = g(\{Y_{t}(x), Y_{t-d_{t}}(x), \dots\})$ where $g:\mathbb{R}^{\infty}\rightarrow \mathbb{R}$ is a measurable function. It can be shown that $\{Z_{t}\}_{t\in\mathbb{R}}$ is an ergodic process (see for example, Theorem 36.4 on page 495 of \cite{Billing1995}). Since $\hat{\gamma}_{x}(d_{t}) = (m-1)^{-1}\sum_{k = 1}^{m-1}\left(Y_{t_{k + 1}}(x) - Y_{t_{k}}(x)\right)^{2}$ is a measurable function of $\{Y_{t_{m}}(x), Y_{t_{m-1}}(x), \dots\}$, and $Y_{t}(x)$ is stationary, $\hat{\gamma}_{x}(d_{t}) \stackrel{p}{\rightarrow} \gamma^{(T)}(d_{t})$ as $m\rightarrow\infty$. Thus, $\hat{\gamma}^{(T)}(d_{t}) = n^{-1}\sum_{j = 1}^{n} \hat{\gamma}_{x_{j}}(d_{t}) \stackrel{p}{\rightarrow} \gamma^{(T)}(d_{t})$. Likewise, $\hat{\gamma}^{(S)}(d_{x}) \stackrel{p}{\rightarrow} \gamma^{(S)}(d_{x})$ as $n \rightarrow \infty$. Under increasing domain asymptotics where the number of spatial and temporal pairs increase, $\hat{\gamma}^{(T)}(d_{t}) \stackrel{p}{\rightarrow} \gamma^{(T)}(d_{t})$ and $\hat{\gamma}^{(S)}(d_{x}) \stackrel{p}{\rightarrow} \gamma^{(S)}(d_{x})$.
\\
Using Slutsky's theorem and the consistency of $\hat{\kappa}_{2}$, the empirical normalised variograms are consistent. We deduce that the MM estimators of $\lambda$, $c$, and the distributional parameters are consistent via the Continuous Mapping Theorem.
\\
Once we have the consistency of our empirical normalised variograms, we use an edited version of Theorem 3.1 on page 70 of \cite{LLC2002} to deduce that all our LS estimators are consistent. Specifically, we replace the empirical and theoretical variograms in the Theorem by their normalised forms. Since Condition (C.1) obviously holds while Conditions (C.2) and (C.3) hold from the exponential structure of our normalised variograms and the choice of the identity matrix of the weight matrix, all the arguments in the proof apply.
\end{proof}

\bibliographystyle{agsm} 
\bibliography{refs}

\end{document}


\maketitle



\section{Introduction}
\label{sec:Intro}

This document provides the readers of \cite{WP2015} with additional information on the theoretical background of the spatio-temporal Ornstein-Uhlenbeck process, computational details for Examples 2 and 8, as well as additional simulation and empirical results. 
\\
In \cite{WP2015}, we introduced the $\mathrm{OU}_{\wedge}$ process which can be interpreted as a spatio-temporal Ornstein-Uhlenbeck process. This is defined as a random field, $Y_{t}(\mathbf{x})$, in space-time $S = \mathcal{X} \times \mathcal{T}$ such that: 
\begin{equation}
Y_{t}(\mathbf{x}) = \int_{A_{t}(\mathbf{x})} \exp(-\lambda(t-s)) L(\mathrm{d}\xi, \mathrm{d}s), \label{eqn:OUh}
\end{equation}
where $\lambda > 0$, $L$ is a homogeneous L\'evy basis and $A_{t}(\mathbf{x}) = A_{0}(0) + (\mathbf{x}, t)$ is a subset of $S$ satisfying:
\begin{equation*}
A_{s}(\mathbf{x}) \subset A_{t}(\mathbf{x}) \text{, } \forall \text{ } s < t, \text{ and } A_{t}(\mathbf{x}) \cap \big(\mathcal{X} \times (t, \infty)\big) = \emptyset.
\end{equation*}
We stated that the process is well-defined if the integral exists in the sense of \cite{RR1989}. In Section \ref{sec:RRtheory}, we summarise the key concepts of this integration theory and present the conditions for $\mathrm{OU}_{\wedge}$ integrals to exist. Additionally, we prove that the processes we worked with in our simulations, i.e. the canonical $\mathrm{OU}_{\wedge}$ processes given by:
\begin{equation}
Y_{t}(x) = \int_{-\infty}^{t}\int_{x - c|t-s|}^{x + c|t-s|} \exp(-\lambda(t-s)) L(\mathrm{d}\xi, \mathrm{d}s),\label{eqn:TestOUh}
\end{equation}
where $c > 0$, and $L$ is either a homogeneous Gaussian, inverse Gaussian (IG), normal inverse Gaussian (NIG) or Gamma basis, are well-defined.
\\
In Section \ref{sec:proofs}, we provide the full details behind Examples 2 and 8 in \cite{WP2015}. These address the spatio-temporal autocorrelation structure of the canonical $\mathrm{OU}_{\wedge}$ process and the mean squared errors (MSEs) of the discrete convolution simulation algorithms respectively.
\\
In Section 5.3 and 5.4 of \cite{WP2015}, we presented the inference results of the data simulated from the canonical $\mathrm{OU}_{\wedge}$ process with $L$ Gaussian and $c = 1$. We then mentioned that similar features were observed when we varied the L\'evy basis and the value of $c$. In Section \ref{sec:addsim} of this document, we display the rest of the results from applying the moments-matching (MM) estimation method and its least-squares (LS) adaptation to data simulated using the rectangular grid (RG) and diamond grid (DG) algorithms.
\\
In Section 6 of \cite{WP2015}, we conducted an empirical study of radiation anomaly data \cite[]{website:NOAA}. Using the LS inference method, we fitted two canonical $\mathrm{OU}_{\wedge}$ processes to the data: one with a Gaussian basis and the other with a NIG basis. We also displayed heat plots of data generated from these fitted models. More plots from different data runs are given in Section \ref{sec:addemp} of this document. 

\section{Stochastic integration theory and the existence of the $\mathrm{OU}_{\wedge}$ processes} \label{sec:RRtheory}

The $\mathcal{L}_{0}$ theory in \cite{RR1989} provides us with the definition and existence conditions for stochastic integrals involving deterministic integrands and general L\'evy bases. Since we only deal with homogeneous L\'evy bases in $\mathrm{OU}_{\wedge}$ processes, we formulate our summary of the theory in this setting. Let $\mathcal{S}$ be the Borel $\sigma$-algebra on our space-time $S = \mathbb{R}^{d} \times \mathbb{R}$ for $d\in\mathbb{N}$ and $\mathcal{B}_{b}(S) = \{E \in \mathcal{S}: \lambda_{d+1}(E) < \infty\}$ where $\lambda_{d+1}$ denotes the Lebesgue measure on $\mathcal{B}(\mathbb{R}^{d+1})$. We work in the probability space $(\Omega, \mathcal{F}, P)$. As is commonly done, we start with the integral of a simple function:
 \\
\begin{Def}[Stochastic integral of a simple function] \hfill \\
Let $\{y_{j}: j = 1, ..., m\}$ be a set of real numbers, and $\{E_{j}: j = 1,..., m\}$ be a collection of disjoint sets of $\mathcal{B}_{b}(S)$. Then, $f(\mathbf{x}, t) = \sum_{j = 1}^{m} y_{j}\mathbf{1}_{E_{j}}$ is called a simple function on $S$. Here, $\mathbf{1}_{E_{j}} = 1$ if $(\mathbf{x}, t) \in E_{j}$ and $0$ otherwise.
\\
We define the stochastic integral of $f$ over $A \in \mathcal{S}$ by $\int_{A}f L(\mathrm{d}\xi,\mathrm{d}s) = \sum_{j = 1}^{m} y_{j} L(A \cap E_{j})$.
\end{Def}

Now, we extend the definition to measurable functions:
\\
\begin{Def}[$L$-integrable functions, integrable functions and their stochastic integrals] \hfill \\
A measurable function $f:(S, \mathcal{S}) \rightarrow (\mathbb{R}, \mathcal{B}(\mathbb{R}))$ is said to be $L$-integrable if there exists a sequence $\{f_{n}\}$ of simple functions such that:
\begin{itemize}
\item[(i)] $f_{n} \rightarrow f$ almost everywhere with respect to $L$.
\item[(ii)] for every $A \in \mathcal{S}$, the sequence $\{\int_{A}f_{n}L(\mathrm{d}\xi, \mathrm{d}s)\}$ converges in probability as $n \rightarrow \infty$. 
\end{itemize}

If $f$ is $L$-integrable, we write:
\begin{equation}
\int_{A}f L(\mathrm{d}\xi, \mathrm{d}s) = P-\lim_{n\rightarrow\infty} \int_{A} f_{n} L(\mathrm{d}\xi, \mathrm{d}s).
\label{eqn:stochintegral}
\end{equation}
The construction is well-defined as the limit does not depend on $\{f_{n}\}$. 
\\
If we are just interested in the integral of $f$ over one set, $A \in \mathcal{S}$ say, we can simplify this idea. We call $f$ integrable over $A$ if (i) holds, and the sequence $\int_{A} f_{n}L(\mathrm{d}\xi, \mathrm{d}s)$ converges in probability as $n\rightarrow \infty$. Similar to before, we define the required stochastic integral via (\ref{eqn:stochintegral}).
\end{Def}

We have seen that the stochastic integral is defined as a limit in probability. Theorem 2.7 in \cite{RR1989} provides us with explicit conditions for the existence of this limit:
\\
\begin{thm}
\label{thm:intconditions}
Let $L$ be a homogeneous L\'evy basis on $(S, \mathcal{S})$ and let the L\'evy-Khintchine triplet of its seed be $(a, b, \nu)$. The measurable function $f:(S, \mathcal{S}) \rightarrow (\mathbb{R}, \mathcal{B}(\mathbb{R}))$ is $L$-integrable if and only if:
\begin{gather*}
\int_{S} |U(f(\xi, s))|\mathrm{d}\xi\mathrm{d}s < \infty, \int_{S} b|f(\xi, s)|^{2} \mathrm{d}\xi\mathrm{d}s < \infty, \text{ and } \int_{S} V_{0}(f(\xi, s)) \mathrm{d}\xi\mathrm{d}s <\infty, \\
\text{where } U(u) = u a + \int_{\mathbb{R}} \big(\rho(zu) - u\rho(z)\big) \nu(\mathrm{d}z), \text{ }\rho(z) = z\mathbf{1}_{|z|\leq 1}, \text{ and } V_{0}(u) = \int_{\mathbb{R}} \min(1, |zu|^{2})\nu(\mathrm{d}z).
\end{gather*}
By virtue of the proof to the theorem in \cite{RR1989}, we can check that a stochastic integral over a chosen set $A \in \mathcal{S}$ exists by evaluating the same conditions over $A$ instead of $S$. 
\end{thm}

Since the L\'evy bases we work with in our simulation studies (i.e. the homogeneous Gaussian, IG, Gamma and NIG bases) have finite second moments, the $\mathcal{L}_{2}$ integration theory in \cite{Walsh1986a} tells us that our $\mathrm{OU}_{\wedge}$ processes are well-defined. As $\mathcal{L}_{2}$ convergence implies convergence in probability, Rajput and Rosinski's integration criteria should be all met. We check that this is true for the canonical processes. Let $A = A_{t}(\mathrm{x}) = \{(\xi, s): s<t, x-c|t-s|<\xi<x+c|t-s|\}$. Then, the conditions we require are:
\begin{gather}
\int_{A} |ua + \int_{\mathbb{R}}\left(zu 1_{|zu|\leq 1} - zu 1_{|z|\leq 1} \right)\nu(\mathrm{d}z)|\mathrm{d}\xi\mathrm{d}s < \infty ,  \label{eqn:c1}\\
\int_{A} b\exp(-2\lambda(t-s))\mathrm{d}\xi\mathrm{d}s < \infty, \label{eqn:c2}\\
\text{ and } \int_{A} \int_{\mathbb{R}}\min(1, |zu|^{2})\nu(\mathrm{d}z)\mathrm{d}\xi\mathrm{d}s < \infty, \label{eqn:c3}
\end{gather}
where $u = \exp(-\lambda(t-s))$.
\\
Condition (\ref{eqn:c2}) is satisfied for all homogeneous L\'evy bases:
\begin{equation*} 
\int_{A}b\exp(-2\lambda(t-s))\mathrm{d}\xi\mathrm{d}s =  \int^{t}_{-\infty} 2c(t-s)be^{-2\lambda(t-s)}\mathrm{d}s = \int^{\infty}_{0} 2cwbe^{-2\lambda w}\mathrm{d}w = \frac{cb}{2\lambda^{2}} < \infty,
\end{equation*}
by using $w = t-s$.
\\
To check the other two conditions, we will use the following inequalities: 
\begin{align}
\min(1, |zu|^{2} ) &\leq |zu|^{2}, \label{eqn:fieq}\\
\text{and } zu1_{|zu|\leq1} - zu1_{|z|\leq1} &\leq (u^{2} + |u|)z^{2}. \label{eqn:sieq}
\end{align}
The first inequality is obviously true and the second can be proved by considering different relative magnitudes of $|zu|$ and $|z|$. 
\\
Now, for Condition (\ref{eqn:c1}):
\begin{align*}
&\int_{A} \left|ua + \int_{\mathbb{R}}\left(zu 1_{|zu|\leq 1} - zu 1_{|z|\leq 1} \right)\nu(\mathrm{d}z)\right|\mathrm{d}\xi\mathrm{d}s \\ 
&\leq \int_{A} \left(|ua| + \left|\int_{\mathbb{R}}\left(zu 1_{|zu|\leq 1} - zu 1_{|z|\leq 1} \right)\nu(\mathrm{d}z)\right|\right)\mathrm{d}\xi\mathrm{d}s \text{ by the Triangle Inequality,} \\
&\leq \int_{A} \left(|ua| + \int_{\mathbb{R}}(u^{2} + |u|)z^{2}\nu(\mathrm{d}z)\right)\mathrm{d}\xi\mathrm{d}s \text{ by (\ref{eqn:sieq}),} \\
&= |a|\int_{A} u \mathrm{d}\xi\mathrm{d}s + \left(\int_{\mathbb{R}}z^{2}\nu(\mathrm{d}z) \right)\int_{A}u^{2}\mathrm{d}\xi\mathrm{d}s + \left(\int_{\mathbb{R}}z^{2}\nu(\mathrm{d}z) \right)\int_{A}|u|\mathrm{d}\xi\mathrm{d}s.
\end{align*} 
Since $\int_{\mathbb{R}}z^{2}\nu(\mathrm{d}z) < \infty$ for our L\'evy bases with finite second moments, and the integrand function $u$ is both absolutely integrable and square-integrable, the expression is finite. So, Condition (\ref{eqn:c1}) holds. 
\\
By using similar arguments and (\ref{eqn:fieq}), we also show that Condition (\ref{eqn:c3}) is satisfied in our scenario:
\begin{align*}
\int_{A} \int_{\mathbb{R}}\min(1, |zu|^{2})\nu(\mathrm{d}z)\mathrm{d}\xi\mathrm{d}s \leq \int_{A} \int_{\mathbb{R}}|zu|^{2}\nu(\mathrm{d}z)\mathrm{d}\xi\mathrm{d}s \leq \left(\int_{\mathbb{R}}z^{2}\nu(\mathrm{d}z) \right)\int_{A}u^{2}\mathrm{d}\xi\mathrm{d}s < \infty.
\end{align*}
Thus, the canonical processes which we used in our simulation studies are well-defined.

\section{Computational details} \label{sec:proofs}

In this section, we give the full proofs of the results in Examples 2 and 8 of the main paper (MP), \cite{WP2015}. 

\subsection{Example 2 (MP)}

\begin{proof}
\label{eg:STcorr}
Let $Y_{t}(x)$  be the canonical $\mathrm{OU}_{\wedge}$ process described by (\ref{eqn:TestOUh}). Then with $(x^{*}, t^{*})$ being the intersection point between $A_{t}(x)$ and $A_{t+d_{t}}(x + d_{x})$: 
\begin{align*}
\Cov[Y_{t}(x), Y_{t + d_{t}}(x+d_{x})]  &= \Var[L']  \exp(-\lambda d_{t})\int_{-\infty}^{t^{*}}\int_{x^{*}-c|s - t^{*}|}^{x^{*} + c|s-t^{*}|} \exp(-2\lambda(t - s)) \mathrm{d}\xi\mathrm{d}s\\
&= \Var[L']  \exp(-\lambda d_{t})\int_{-\infty}^{t^{*}}2c|s - t^{*}|\exp(-2\lambda(t - s))\mathrm{d}s.
\end{align*}
Without loss of generality, assume that $d_{t} \geq 0$, there are now three cases to consider: 

\begin{itemize}
\item[(i)] If  $d_{x} > c d_{t}$, the intersection occurs on the lower bound of $A_{t+d_{t}}(x+d_{x})$ and the upper bound of $A_{t}(x)$. So:
\begin{align*}
x^{*} = x + c(t - t^{*}) &= x + d_{x} - c(t + d_{t} -t^{*}) \Rightarrow t^{*} = \frac{1}{2c}(2ct + cd_{t} - d_{x}) = t + \frac{d_{t}}{2} - \frac{d_{x}}{2c}. \\
\Rightarrow \Cov[Y_{t}(x), Y_{t + d_{t}}(x+d_{x})]  &= \Var[L']  \exp(-\lambda d_{t})\int_{-\infty}^{t^{*}}2c(t^{*}-s)\exp\left(-2\lambda\left(t^{*} - \frac{d_{t}}{2} + \frac{d_{x}}{2c} - s\right)\right)\mathrm{d}s \\
&= \Var[L']  \exp(-\frac{\lambda d_{x}}{c})\int_{0}^{\infty}2cw\exp(-2\lambda w)\mathrm{d}w \text{ where } w = t^{*} - s.\\
\Rightarrow \Corr[Y_{t}(x), Y_{t + d_{t}}(x+d_{x})]  &= \exp\left(-\frac{\lambda d_{x}}{c}\right).
\end{align*}
\item[(ii)] If $-c d_{t} \leq d_{x} \leq c d_{t}$, $A_{t}(x) \cap A_{t+d_{t}}(x+d_{t}) = A_{t}(x)$. So:
\begin{align*}
\Cov[Y_{t}(x), Y_{t + d_{t}}(x+d_{x})] &= \Var[L']  \exp(-\lambda d_{t})\int_{-\infty}^{t}2c(t-s)\exp(-2\lambda(t - s))\mathrm{d}s \\
&= \Var[L']  \exp(-\lambda d_{t})\int_{0}^{\infty}2cw\exp(-2\lambda w)\mathrm{d}w \text{ where } w = t^{*} - s. \\
\Rightarrow \Corr[Y_{t}(x), Y_{t + d_{t}}(x+d_{x})]  &= \exp(-\lambda d_{t}).
\end{align*}
\item[(iii)] If $d_{x} < -c d_{t}$, the intersection occurs on the upper bound of $A_{t+d_{t}}(x+d_{x})$ and the lower bound of $A_{t}(x)$. So:
\begin{align*}
x^{*} = x - c(t - t^{*}) &= x + d_{x} + c(t + d_{t} -t^{*}) \Rightarrow t^{*} = \frac{1}{2c}(2ct + cd_{t} + d_{x}) = t + \frac{d_{t}}{2} + \frac{d_{x}}{2c}. \\
\Rightarrow \Cov[Y_{t}(x), Y_{t + d_{t}}(x+d_{x})] &= \Var[L']  \exp(-\lambda d_{t})\int_{-\infty}^{t^{*}}2c(t^{*}-s)\exp\left(-2\lambda\left(t^{*} - \frac{d_{t}}{2} - \frac{d_{x}}{2c} - s\right)\right)\mathrm{d}s \\
&= \Var[L']  \exp\left(\frac{\lambda d_{x}}{c}\right)\int_{0}^{\infty}2cw\exp(-2\lambda w)\mathrm{d}w \text{ where } w = t^{*} - s. \\
\Rightarrow \Corr[Y_{t}(x), Y_{t + d_{t}}(x+d_{x})]  &= \exp\left(\frac{\lambda d_{x}}{c}\right).
\end{align*}
\end{itemize}
In conclusion, for $d_{t} \geq 0$, we have:
\begin{align*}
\Corr[Y_{t}(x), Y_{t + d_{t}}(x+d_{x})]  &:= \rho^{(S, T)}(|d_{x}|, d_{t}) = \begin{cases}
  \exp(-\lambda d_{t}) & \text{if } |d_{x}| \leq c d_{t}, \\
   \exp\left(-\frac{\lambda |d_{x}|}{c}\right) & \text{ otherwise,} 
  \end{cases} \nonumber\\
&= \min\left(\exp\left(-\lambda d_{t}\right), \exp\left(-\frac{\lambda|d_{x}|}{c}\right)\right). \label{eqn:STCorr}
\end{align*}
\end{proof}

\subsection{Example 8 (MP)}

\begin{proof}
First, we compute the MSE of the RG algorithm. The first term on the right-hand side of (18, MP) is equal to:
\begin{align*}
&\left(\frac{2\mu}{\lambda^{2}} -\mu\int_{\mathbb{R}\times\mathbb{R}} \sum_{j = 0}^{R/\triangle}\sum_{i = -R/\triangle}^{R/\triangle} \mathbf{1}_{|i\triangle|\leq j\triangle}\mathbf{1}_{\left(t - (j+1)\triangle, t - j\triangle\right]}(s) \mathbf{1}_{\left[x + i\triangle- \frac{\triangle}{2}, x + i\triangle + \frac{\triangle}{2}\right)}(\xi) \exp\left(-\lambda j \triangle\right)\mathrm{d}\xi\mathrm{d}s  \right)^{2} \\
&= \left(\frac{2\mu}{\lambda^{2}} -\mu\sum_{j = 0}^{R/\triangle}(2j+1)\triangle^{2} \exp\left(-\lambda j \triangle\right) \right)^{2} \\
&= \left(\frac{2\mu}{\lambda^{2}} - 2\mu\left[ \frac{\triangle^{2}}{\left(1 - e^{-\lambda \triangle}\right)^{2}}\right]  \left[ \frac{1 + e^{-\lambda\triangle} - (\frac{2R}{\triangle} + 3)e^{-(R+\triangle)\lambda}+ \left(\frac{2R}{\triangle}+1\right) e^{-(R+2\triangle)\lambda}}{2}\right]  \right)^{2}.
\end{align*}
When $R$ is fixed, $\triangle^{2}/\left(1 - e^{-\lambda \triangle}\right)^{2} = 1/\lambda^{2} + O(\triangle)$. In addition:
\begin{align*}
&\frac{1 + e^{-\lambda\triangle} - (\frac{2R}{\triangle} + 3)e^{-(R+\triangle)\lambda}+ \left(\frac{2R}{\triangle}+1\right) e^{-(R+2\triangle)\lambda}}{2} \\
&=\frac{1}{2}\left[2 - \lambda\triangle + O(\triangle^{2}) + 2Re^{-R\lambda}\left(-\lambda + \frac{3\triangle\lambda^{2}}{2} + O(\triangle^{2})\right) + e^{-R\lambda}\left(-2+\triangle\lambda + O(\triangle^{2})\right)\right] \\
&=1 - (1+\lambda R)e^{-R\lambda} + O(\triangle).
\end{align*}
So, $\left(\mathbb{E}\left[Y_{t}(x)\right] -  \mathbb{E}\left[Z_{t}(x)\right]\right)^{2}= 4\mu^{2}\lambda^{-4}(1+\lambda R)^{2}e^{-2R\lambda} + O(\triangle)$.
\\
On the other hand, since $\lim\limits_{R\rightarrow\infty}Re^{-R\lambda} = 0$ and  $\lim\limits_{R\rightarrow\infty}e^{-R\lambda} = 0$, $\left(\mathbb{E}\left[Y_{t}(x)\right] -  \mathbb{E}\left[Z_{t}(x)\right]\right)^{2}\stackrel{R\rightarrow \infty}{\longrightarrow}$\\ $4\mu^{2}\lambda^{-4}\left(1 - \lambda^{2}\triangle^{2}\left(1 - e^{-\lambda \triangle}\right)^{-2}\left(1 + e^{-\lambda\triangle}\right)/2\right)^{2}$.
\\
Next, we focus on the second term on the right-hand side of (18, MP):
\begin{align*}
\int_{\mathbb{R}\times\mathbb{R}}\left(k(x, t, \xi, s)-h(x, t, \xi, s)\right)^{2}\Var\left[L'\right]\mathrm{d}\xi\mathrm{d}s &= \tau^{2}\int_{\mathbb{R}\times\mathbb{R}}k^{2}(x, t, \xi, s) - 2k(x, t, \xi, s)h(x, t, \xi, s) + h^{2}(x, t, \xi, s)\mathrm{d}\xi\mathrm{d}s.
\end{align*}
The mixed term $\int_{\mathbb{R}\times\mathbb{R}}k(x, t, \xi, s)h(x, t, \xi, s) \mathrm{d}\xi\mathrm{d}s$ involves the parts of the approximated ambit set which lie within the true ambit set. By reference to Figure 2(ii) (MP), this is equal to:
\begin{align*}
&\sum_{j = 0}^{R/\triangle}\sum_{i = -j}^{j} e^{-\lambda j \triangle}\int_{A_{t}(x)}e^{-\lambda(t-s)}\mathbf{1}_{\left(t - (j+1)\triangle, t - j\triangle\right]}(s) \mathbf{1}_{\left[x + i\triangle- \frac{\triangle}{2}, x + i\triangle + \frac{\triangle}{2}\right)}(\xi)  \mathrm{d}\xi\mathrm{d}s \\
&=  \sum_{j = 0}^{R/\triangle} e^{-\lambda j \triangle} \left[\int_{t - j\triangle - \frac{\triangle}{2}}^{t - j\triangle}\int_{x-|t-s|}^{x+|t-s|}e^{-\lambda(t-s)}\mathrm{d}\xi\mathrm{d}s + \int_{t - (j+1)\triangle}^{t - j\triangle-\frac{\triangle}{2}}\int_{x - j\triangle - \frac{\triangle}{2}}^{x + j\triangle + \frac{\triangle}{2}}e^{-\lambda(t-s)}\mathrm{d}\xi\mathrm{d}s \right] \\
&= \sum_{j = 0}^{R/\triangle} \frac{e^{-2\lambda j \triangle}}{\lambda} \left\{2j\triangle + \frac{2}{\lambda}\left(1 - e^{-\frac{\lambda\triangle}{2}}\right) - (2j+1)\triangle e^{-\lambda\triangle} \right\}.
\end{align*}
After computing the other terms, we find that $\int_{\mathbb{R}\times\mathbb{R}}\left(k(x, t, \xi, s)-h(x, t, \xi, s)\right)^{2}\Var\left[L'\right]\mathrm{d}\xi\mathrm{d}s$ is given by:
\begin{align*}
&\tau^{2}\left[\frac{1}{2\lambda^{2}} + \sum_{j = 0}^{R/\triangle} e^{-2\lambda j \triangle} \left\{-\frac{4j\triangle}{\lambda} - \frac{4}{\lambda^{2}}\left(1 - e^{-\frac{\lambda\triangle}{2}}\right) + \frac{2(2j+1)\triangle}{\lambda} e^{-\lambda\triangle} + (2j+1)\triangle^{2} \right\} \right] \\
&= \tau^{2}\left[\frac{1}{2\lambda^{2}} + \sum_{j = 0}^{R/\triangle} e^{-2\lambda j \triangle} \left\{-\frac{4j\triangle}{\lambda} - \frac{4}{\lambda^{2}}\left(1 - e^{-\frac{\lambda\triangle}{2}}\right) + \frac{2(2j+1)\triangle}{\lambda} e^{-\lambda\triangle} + (2j+1)\triangle^{2} \right\} \right] \\
&= \tau^{2}\left[\frac{1}{2\lambda^{2}} -  \frac{1}{2\lambda^{2}}\left(2\frac{\left(1 - e^{-\lambda\triangle}\right)}{\lambda\triangle}-1\right)\left(e^{-2\lambda \triangle} - \left(\frac{R}{\triangle} + 1\right)e^{-2(R+\triangle)\lambda}+\frac{R}{\triangle} e^{-2(R+2\triangle)\lambda}\right)  \frac{4\lambda^{2}\triangle^{2}}{\left(1 - e^{-2\lambda \triangle}\right)^{2}} \right.\\
&\left.+  \left( \frac{e^{-\lambda\triangle}}{\lambda^{2}}+ \frac{\triangle}{2\lambda} - \frac{\left(1 - e^{-\frac{\lambda\triangle}{2}}\right)}{\lambda^{2}\left(\frac{\lambda\triangle}{2}\right)}\right)\frac{2\lambda\triangle}{1 - e^{-2\lambda\triangle}}\left(1 -e^{-2(R + \triangle)\lambda}\right)\right] \\
&= \frac{\tau^{2}}{2\lambda^{2}}(1+2R\lambda)e^{-2R\lambda} + O(\triangle),
\end{align*}
by using similar arguments involving series expansions and treating $R$ as fixed. On the other hand, if we fix $\triangle$ and let $R \rightarrow \infty$, this converges to:
\begin{equation*}
\tau^{2} \left[\frac{1}{2\lambda^{2}}\left(1 - \left(\frac{2\left(1 - e^{-\lambda\triangle}\right)}{\lambda\triangle}-1\right)\frac{4\lambda^{2}\triangle^{2} e^{-2\lambda \triangle}}{\left(1 - e^{-2\lambda \triangle}\right)^{2}} \right) + \left( \frac{e^{-\lambda\triangle}}{\lambda^{2}}+ \frac{\triangle}{2\lambda} - \frac{\left(1 - e^{-\frac{\lambda\triangle}{2}}\right)}{\lambda^{2}\left(\frac{\lambda\triangle}{2}\right)}\right)\frac{2\lambda\triangle}{1 - e^{-2\lambda\triangle}}\right].
\end{equation*} 
When we put the first and second terms of the MSE together, we have: 
\begin{align*}
& \mathbb{E}\left[\left|Y_{t}(x)-Z_{t}(x)\right|^{2}\right] = \left[\frac{4\mu^{2}}{\lambda^{4}}(1+\lambda R)^{2}+\frac{\tau^{2}}{2\lambda^{2}}(1+2\lambda R)\right]e^{-2R\lambda} + O(\triangle) \text{ for fixed } R, \\
&\text{and } \mathbb{E}\left[\left|Y_{t}(x)-Z_{t}(x)\right|^{2}\right] \stackrel{R\rightarrow \infty}{\longrightarrow}\frac{4\mu^{2}}{\lambda^{4}} \left(1- \left[ \frac{\lambda^{2}\triangle^{2}}{\left(1 - e^{-\lambda \triangle}\right)^{2}}\right]\left[\frac{1 + e^{-\lambda\triangle} }{2}\right]\right)^{2} \\
&+\tau^{2} \left[\frac{1}{2\lambda^{2}}\left(1 - \left(\frac{2\left(1 - e^{-\lambda\triangle}\right)}{\lambda\triangle}-1\right)\frac{4\lambda^{2}\triangle^{2} e^{-2\lambda \triangle}}{\left(1 - e^{-2\lambda \triangle}\right)^{2}} \right) + \left( \frac{e^{-\lambda\triangle}}{\lambda^{2}}+ \frac{\triangle}{2\lambda} - \frac{\left(1 - e^{-\frac{\lambda\triangle}{2}}\right)}{\lambda^{2}\left(\frac{\lambda\triangle}{2}\right)}\right)\frac{2\lambda\triangle}{1 - e^{-2\lambda\triangle}}\right].
\end{align*}
\\
The DG MSE and its limits can be computed similarly. Now, $\left(\mathbb{E}\left[Y_{t}(x)\right] -  \mathbb{E}\left[Z_{t}(x)\right]\right)^{2}$ is equal to:
\begin{align*}
&\left(\frac{2\mu}{\lambda^{2}}-  \mu\sum_{j = 0}^{R/\triangle}\sum_{i = 0}^{j} e^{-2\lambda j \triangle}\left[\int_{t-(j+1)\triangle}^{t - j\triangle}2(t-j\triangle-s)\mathrm{d}s  + 2\int_{t-(j+2)\triangle}^{t-(j+1)\triangle}\left(\triangle - (t-(j+1)\triangle-s)\right)\mathrm{d}s \right]\right)^{2} \\
&=\left(\frac{2\mu}{\lambda^{2}} -  \mu\sum_{j = 0}^{R/\triangle}2(j+1)\triangle^{2}e^{-\lambda j \triangle}\right)^{2} \\
&= \frac{4\mu^{2}}{\lambda^{4}}\left(1 - \left[\frac{\lambda^{2}\triangle^{2}}{\left(1 - e^{-\lambda \triangle}\right)^{2}}\right] \left(1 - \left(\frac{R}{\triangle} + 2\right)e^{-(R+\triangle)\lambda}+\left(\frac{R}{\triangle}+1\right) e^{-(R+2\triangle)\lambda}\right)\right)^{2}.
\end{align*}
\\
Next, we compute $\int_{\mathbb{R}\times\mathbb{R}}k(x, t, \xi, s)h(x, t, \xi, s) \mathrm{d}\xi\mathrm{d}s$. For the DG algorithm, the intersection of the approximated ambit set and the true ambit set is exactly equal to the approximated ambit set itself. Thus, by reference to Figure 4(ii) (MP), we can write:
\begin{align*}
&\int_{\mathbb{R}\times\mathbb{R}}k(x, t, \xi, s)h(x, t, \xi, s) \mathrm{d}\xi\mathrm{d}s \\
&= \sum_{j = 0}^{R/\triangle}\sum_{i = 0}^{j} e^{-\lambda j \triangle}\left[\int_{t-(j+1)\triangle}^{t - j\triangle}\int_{x+(2i-j)\triangle-(t-j\triangle-s)}^{x+(2i-j)\triangle+(t-j\triangle-s)}e^{-\lambda(t-s)}\mathrm{d}\xi\mathrm{d}s + \int_{t-(j+2)\triangle}^{t-(j+1)\triangle}\int_{x + (2i-j-1)\triangle+(t-(j+1)\triangle-s)}^{x+(2i-j+1)\triangle-(t-(j+1)\triangle-s)}e^{-\lambda(t-s)}\mathrm{d}\xi\mathrm{d}s \right] \\
&= \sum_{j = 0}^{R/\triangle}2(j+1)\triangle^{2}e^{-2\lambda j \triangle}.
\end{align*}
Upon calculating the other components, we have:
\begin{align*}
&\int_{\mathbb{R}\times\mathbb{R}}\left(k(x, t, \xi, s)-h(x, t, \xi, s)\right)^{2}\Var\left[L'\right]\mathrm{d}\xi\mathrm{d}s \\
&= \tau^{2}\left[\frac{1}{2\lambda^{2}} + 2\sum_{j = 0}^{R/\triangle}(j+1)e^{-2\lambda j \triangle} \left\{\triangle^{2} - \frac{2}{\lambda^{2}}\left(1 - e^{-\lambda\triangle}\right)^{2}\right\} \right] \\
&= \frac{\tau^{2}}{2\lambda^{2}}\left[1 -\left\{\frac{2}{\lambda^{2}\triangle^{2}}\left(1 - e^{-\lambda\triangle}\right)^{2}-1\right\}\left[\frac{4\lambda^{2}\triangle^{2}}{\left(1 - e^{-2\lambda \triangle}\right)^{2}}\left(e^{-2\lambda \triangle} - \left(\frac{R}{\triangle} + 1\right)e^{-2(R+\triangle)\lambda}+\frac{R}{\triangle} e^{-2(R+2\triangle)\lambda}\right) \right.\right. \\
&\left.\left. + \frac{4\lambda^{2}\triangle^{2}}{1 - e^{-2\lambda \triangle}} \left(1 -e^{-2(R + \triangle)\lambda}\right)\right]\right].
\end{align*}
By considering the behaviour of the two terms on the right-hand side of (18, MP) separately, we conclude that:
\begin{align*}
\mathbb{E}\left[\left|Y_{t}(x)-Z_{t}(x)\right|^{2}\right] &= \left[\frac{4\mu^{2}}{\lambda^{4}}(1+R)^{2}\lambda^{2}+\frac{\tau^{2}}{2\lambda^{2}}(1 + 2\lambda R)\right]e^{-2R\lambda} + O(\triangle) \text{ for fixed } R, \text{ and}\\
\mathbb{E}\left[\left|Y_{t}(x)-Z_{t}(x)\right|^{2}\right] &\rightarrow \frac{4\mu^{2}}{\lambda^{4}} \left(1- \left[ \frac{\lambda^{2}\triangle^{2}}{\left(1 - e^{-\lambda \triangle}\right)^{2}}\right]\right)^{2} +\frac{\tau^{2}}{2\lambda^{2}} \left[1 - \left\{\frac{2}{\lambda^{2}\triangle^{2}}\left(1 - e^{-\lambda\triangle}\right)^{2}-1\right\}\left[\frac{4\lambda^{2}\triangle^{2}}{\left(1 - e^{-2\lambda \triangle}\right)^{2}}\right]\right],
\end{align*}
as $R\rightarrow \infty$.
\end{proof}

\section{Additional simulation results} \label{sec:addsim}

In \cite{WP2015}, we proposed two estimation methods. The first is a spatio-temporal extension of the MM method in \cite{KAKE2013}. This uses the normalised variograms at single temporal and spatial lags to obtain $\hat{\lambda}$ and $\hat{c}$. The second method is a LS variation of the first technique. Instead of using single lags, we consider a range of lags and fit LS curves to the corresponding values of the normalised variograms. (We chose to use $15$ temporal and spatial lags in our simulation experiments.) In both methods, cumulants of the process were used to estimate the L\'evy basis parameters. We compared the methods in simulation experiments where data was generated using two simulation algorithms. While one generates data on the usual RG, the other generates data on a DG to mimick the edges of the integration set. Model parameters were estimated from the simulated data. A PC with the following characteristics was used in our simulation study: Intel$^{\circledR}$ Core\texttrademark i7-3770 CPU Processor @ 3.40GHz; 8GB of RAM; Windows 8.1 64-bit .
\\
The parameter estimates for the case where $L$ is Gaussian and $c = 1$ were shown via box plots in Section 5 of \cite{WP2015}. Unlike the other model parameters for which both inference methods gave reasonable estimates, interesting results were obtained for the parameter $c$. For data generated on the RG, it was observed that $c$ was consistently underestimated when the MM method was used. In comparison, for data generated on the DG, accurate values of $\hat{c}$ were achieved through both inference methods. For both the RG and the DG data, estimates of $c$ were also seen to exhibit higher variance when the LS approach was used. 
\\
In this section, we show that similar observations can be made when $c = 2$, $c = 0.5$ and $L$ is a IG, NIG or Gamma basis. These bases are defined through their distribution for any $E\in\mathcal{B}_{b}(S)$ in Table \ref{table:ThreeLevyBases} \citep{KAKE2013}. Let $f(x;\mathbf{\theta})$ denote the probability density function (PDF) of $X$ evaluated at $x$, given the parameters $\mathbf{\theta}$. The distributional parameters in Table \ref{table:ThreeLevyBases} correspond to the following PDFs:

\begin{enumerate}
\item[(i)] if $X \sim IG\left(\delta, \gamma\right)$, $f(x; \delta, \gamma) = \delta\left(2\pi x^{3}\right)^{-1/2} \exp\left({\delta\gamma - 2^{-1}\left(\delta^{2}/x + \gamma^{2}x\right)}\right)$, for $x > 0, \delta >0, \gamma >0$;
\item[(ii)] if $X \sim NIG\left(\alpha, \beta, \mu, \delta\right)$, let $K_{1}$ denote the modified Bessel function of third order and index 1. Then, for $x\in\mathbb{R}$, $f(x; \alpha, \beta, \mu, \delta) = \alpha\delta\left(\pi^{2}\left(\delta^2 + \left(x - \mu\right)^{2}\right)\right)^{-1/2}\exp{\left(\delta\sqrt{\alpha^{2}-\beta^{2}} + \beta\left(x-\mu\right)\right)} K_{1}\left(\alpha\sqrt{\delta^{2} + \left(x - \mu\right)^{2}}\right)$, where $\alpha, \beta, \mu$ and $\delta$ satisfy $\mu \in \mathbb{R}$, $\delta > 0$ and $0 \leq |\beta| < \alpha$;
\item[(iii)] if $X \sim Gamma\left(\alpha, \beta\right)$, $f\left(x; \alpha, \beta\right) = \beta^{\alpha}\Gamma\left(\alpha\right)^{-1}x^{\alpha-1}\exp{\left(-\beta x\right)}$, for $x > 0, \alpha > 0$ and $\beta > 0$. As usual, $\Gamma\left(\alpha\right)$ represents the Gamma function evaluated at $\alpha$.
\end{enumerate}

For each scenario, we generate $500$ data sets via the RG and DG algorithms, covering the $[0, 10c] \times [0, 10]$ space-time region, $\triangle_{t} = 0.05$, $p = q = 300$ and $\lambda = 1$. These are the simulation settings used for the Gaussian case with $c = 1$. The basis parameter values used are given in Table \ref{table:WorkedEgBasisParameters}. These are chosen so that the L\'evy seeds have a mean and standard deviation of approximately $0.2$ and $0.1$, allowing us to compare the results of one basis with the Gaussian and other non-Gaussian bases. The L\'evy seed distributions for the four bases are shown in Figure \ref{fig:TrueLevySeedDistr}. Although they share similar first two moments, there are differences in other features of the distributions such as the skewness. The cumulant expressions in terms of the distributional parameters which are needed for the MM and LS inference are shown in Table \ref{table:kkrelations}. 
\\
Figures \ref{fig:BPGau2}-\ref{fig:BPNIG1} correspond to the estimates obtained in the cases of $L$ Gaussian with $c = 2$, $L$ Gaussian with $c = 0.5$, $L$ IG, $L$ Gamma and $L$ NIG. In each figure, the estimates are arranged according to the parameters by row and the grid/inference method by column. 

\clearpage

\begin{figure}[tbp]
\centering
\caption{Chosen L\'evy seed distributions for simulated data.}
\label{fig:TrueLevySeedDistr}
\includegraphics[width = 3.2in, height = 3.4in, trim = 0.4in 0.4in 0.4in 0in]{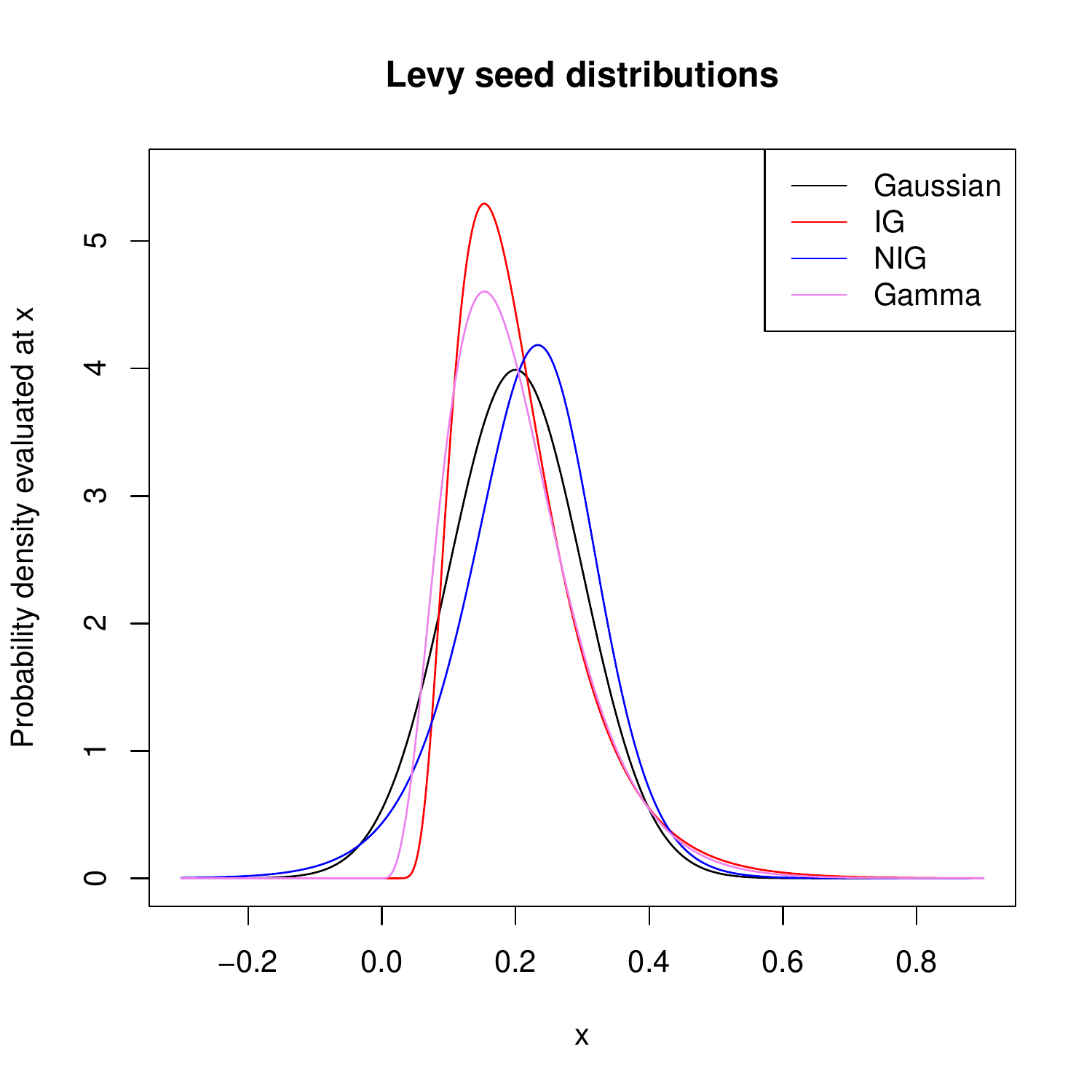}
\end{figure}

\begin{table}[htbp]
\centering
\caption{Distributions of the three homogeneous non-Gaussian L\'evy bases. Here, $E\in\mathcal{B}_{b}(S)$ and $\lambda_{d+1}$ is the Lebesgue measure on $\mathcal{B}(\mathbb{R}^{d+1})$.}
\begin{tabular}{lll}
\hline
\textbf{Basis} & \textbf{Parameters} & \textbf{Distribution} \\
\hline
IG & $\delta, \gamma$ & $L(E) \sim IG\big(\delta\lambda_{d+1}(E), \gamma\big)$ \\
NIG & $\alpha, \beta, \mu, \delta$ & $L(E) \sim NIG\big(\alpha, \beta, \mu\lambda_{d+1}(E), \delta\lambda_{d+1}(E)\big)$ \\
Gamma & $\alpha, \beta$ & $L(E) \sim Gamma\big(\alpha\lambda_{d+1}(E), \beta\big)$\\
\hline
\end{tabular}
\label{table:ThreeLevyBases}
\end{table}

\begin{table}[htbp]
\centering
\caption{True parameters for the IG, NIG and Gamma bases.}
\label{table:WorkedEgBasisParameters}
\begin{tabular}{l|l|l}
\textbf{IG} & \textbf{NIG} &\textbf{Gamma}\\
\hline
$\delta = 1$ & $\alpha = 20$ & $\alpha = 4.3$ \\
$\gamma = 4.8$ & $\beta = -5$ & $\beta = 21.5$ \\
& $\delta = 0.2$ & \\
& $\mu = 0.27$ &
\end{tabular}
\end{table}

From the plots corresponding to the RG algorithm, we see the tendency of the MM method to underestimate $c$ and notice that the LS estimates for $c$ have greater standard deviations than the MM estimates. For example, in the case of $L$ Gaussian and $c = 2$, $\hat{c}$ has a range of about $0.2$ under the MM estimation in Figure \ref{fig:BPGau2}(f) and and a range of about $0.7$ under the LS approach in Figure \ref{fig:BPGau2}(e). This increase in variance under the LS approach also had the effect of centering the $c$ estimates about their true values (given by the red lines). Note that for the NIG basis, $179$ invalid estimates for $\alpha$ and $204$ invalid estimates for $\delta$ were arose due to the parameter constraints and the mismatch between the empirical and theoretical cumulants of $L'$.
\\
Unlike for the RG, accurate estimates of $c$ are obtained from both inference methods for the DG data. We see this from the centering of $\hat{c}$ about the true values (denoted by the red lines) in each scenario. Increased variance in our $c$ estimates is also observed when we apply a LS approach to our estimation method. For example, for the case of $L$ Gaussian and $c = 2$, $\hat{c}$ has a range of about $1.8$ in Figure \ref{fig:BPGau2}(g) which is much larger than $0.14$, the range of $\hat{c}$ under the original MM method in Figure \ref{fig:BPGau2}(h).
\\
More comments on these observations can be found in Section 5 of \cite{WP2015}.

\clearpage

\begin{table}[tbp]
\centering
\caption{The cumulants $\kappa_{l}(L')$, for $l = 1, \dots, 4$, in terms of the L\'evy seed parameters when $L$ is IG, NIG and Gamma.}
\label{table:kkrelations}
\begin{tabular}{llll}
\hline
\textbf{Cumulant} &  \textbf{IG} &\textbf{NIG} & \textbf{Gamma}\\
\hline
$\kappa_{1}(L')$ & $\delta/\gamma$ & $\mu + \delta\beta/(\alpha^{2}-\beta^{2})^{1/2}$ & $\alpha/\beta$\\
$\kappa_{2}(L')$ & $\delta/\gamma^{3}$ & $\delta\alpha^{2}/(\alpha^{2}-\beta^{2})^{3/2}$ & $\alpha/\beta^{2}$\\
$\kappa_{3}(L')$& $3\delta/\gamma^{5}$ & $3\delta\beta\alpha^{2}/(\alpha^{2}-\beta^{2})^{5/2}$ & $2\alpha/\beta^{3}$\\
$\kappa_{4}(L')$  & $15\delta/\gamma^{7}$ & $3\delta(\alpha^{2} + 4\beta^{2})\alpha^{2}/(\alpha^{2}-\beta^{2})^{7/2}$ & $6\alpha/\beta^{4}$ \\
\hline
\end{tabular}
\end{table}

\begin{figure}[tbp]
\centering
\caption{Box plots of the parameter estimates from the LS and MM methods under the RG and DG algorithms in the case $c = 2$ and $L$ being a Gaussian basis. The true values are denoted by the red lines.}
\label{fig:BPGau2}
\includegraphics[width = 4.5in, height = 7in, trim = 0.4in 0.5in 0.4in 0in]{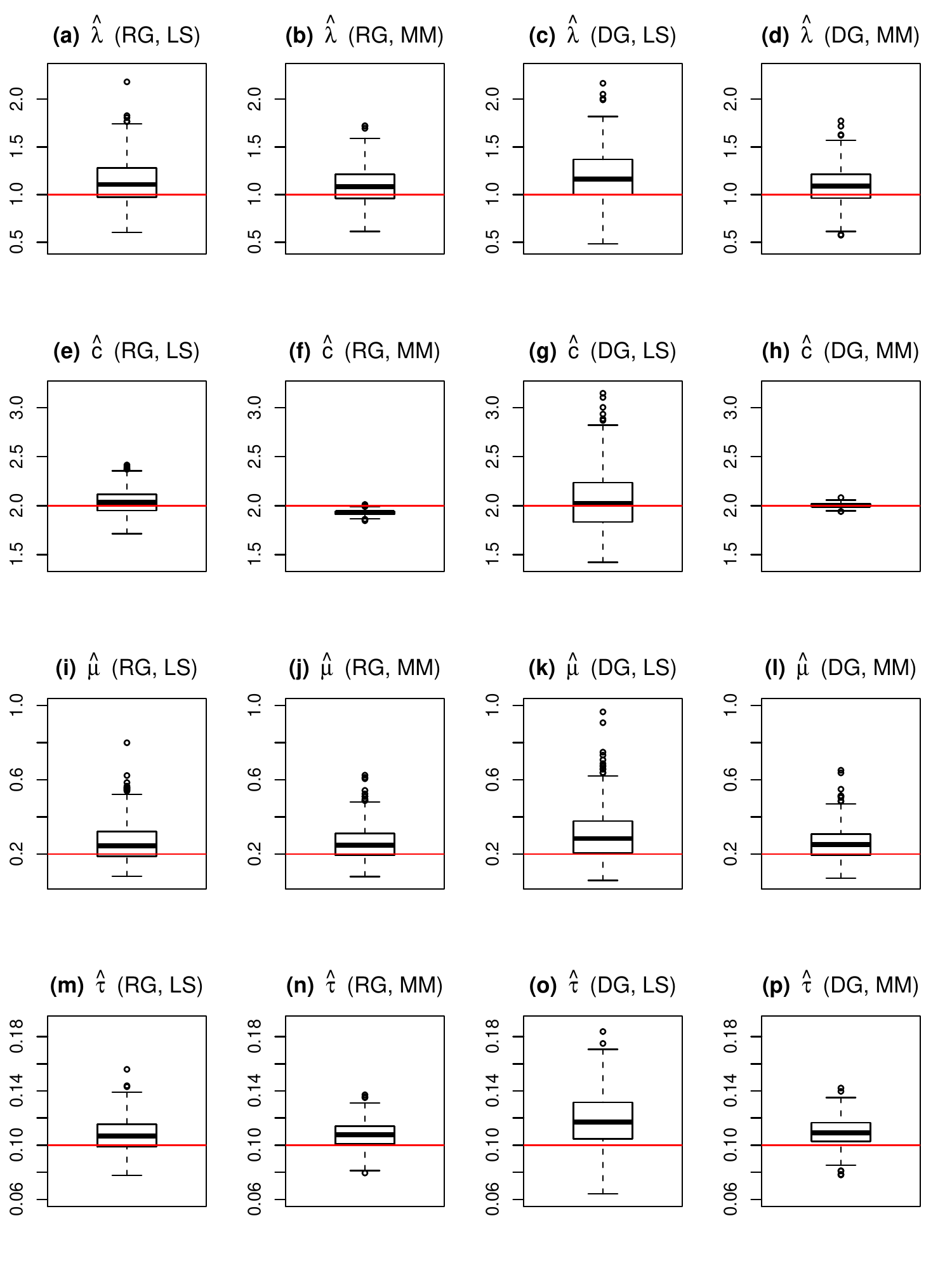}
\end{figure}

\begin{figure}[tbp]
\centering
\caption{Box plots of the parameter estimates from the LS and MM methods under the RG and DG algorithms in the case $c = 0.5$ and $L$ being a Gaussian basis. The true values are denoted by the red lines.}
\label{fig:BPGau05}
\includegraphics[width = 4.5in, height = 7in, trim = 0.4in 0.5in 0.4in 0in]{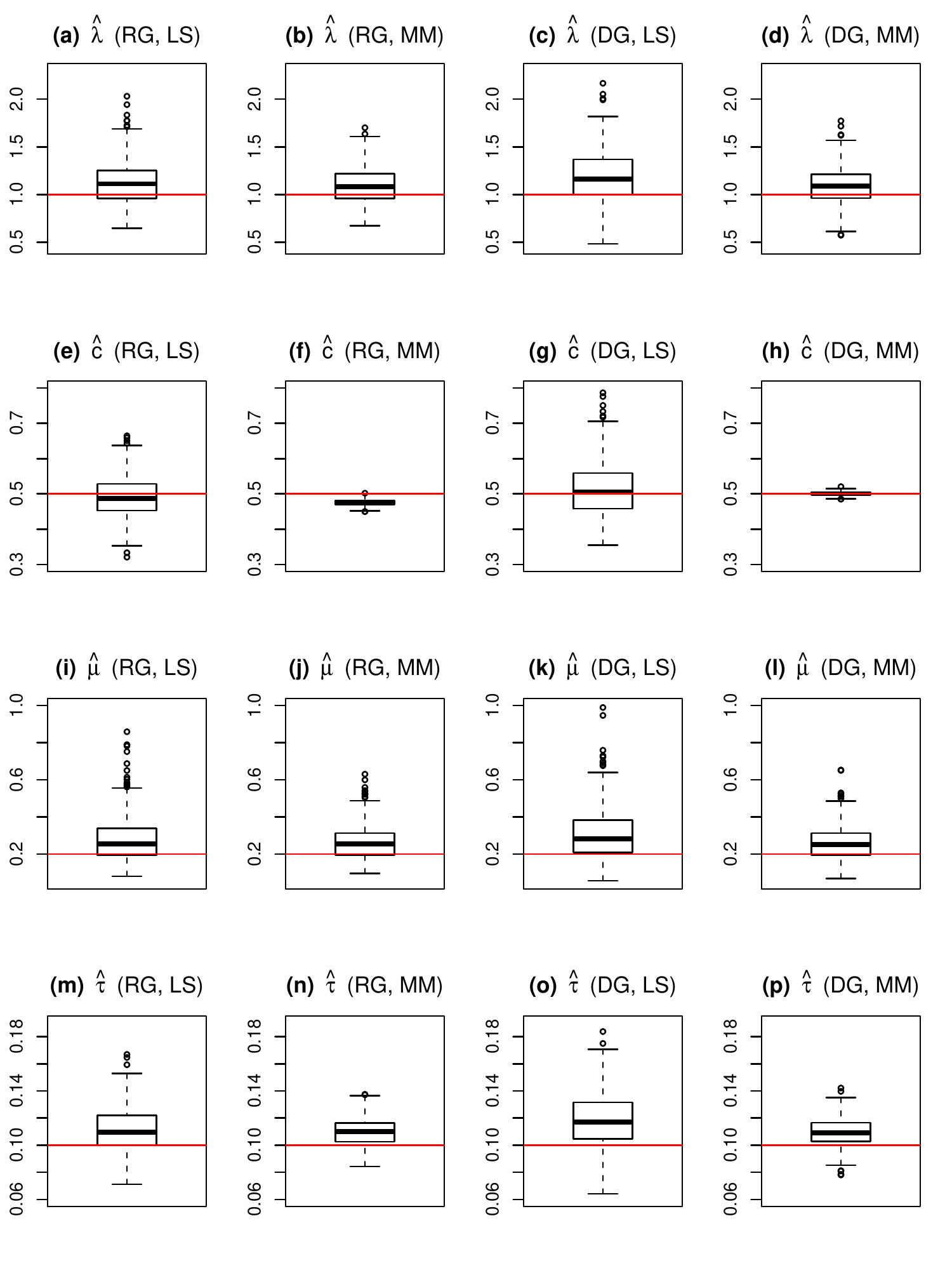}
\end{figure}

\begin{figure}[tbp]
\centering
\caption{Box plots of the parameter estimates from the LS and MM methods under the RG and DG algorithms in the case $c = 1$ and $L$ being an IG basis. The true values are denoted by the red lines.}
\label{fig:BPIG1}
\includegraphics[width = 4.5in, height = 7in, trim = 0.4in 0.5in 0.4in 0in]{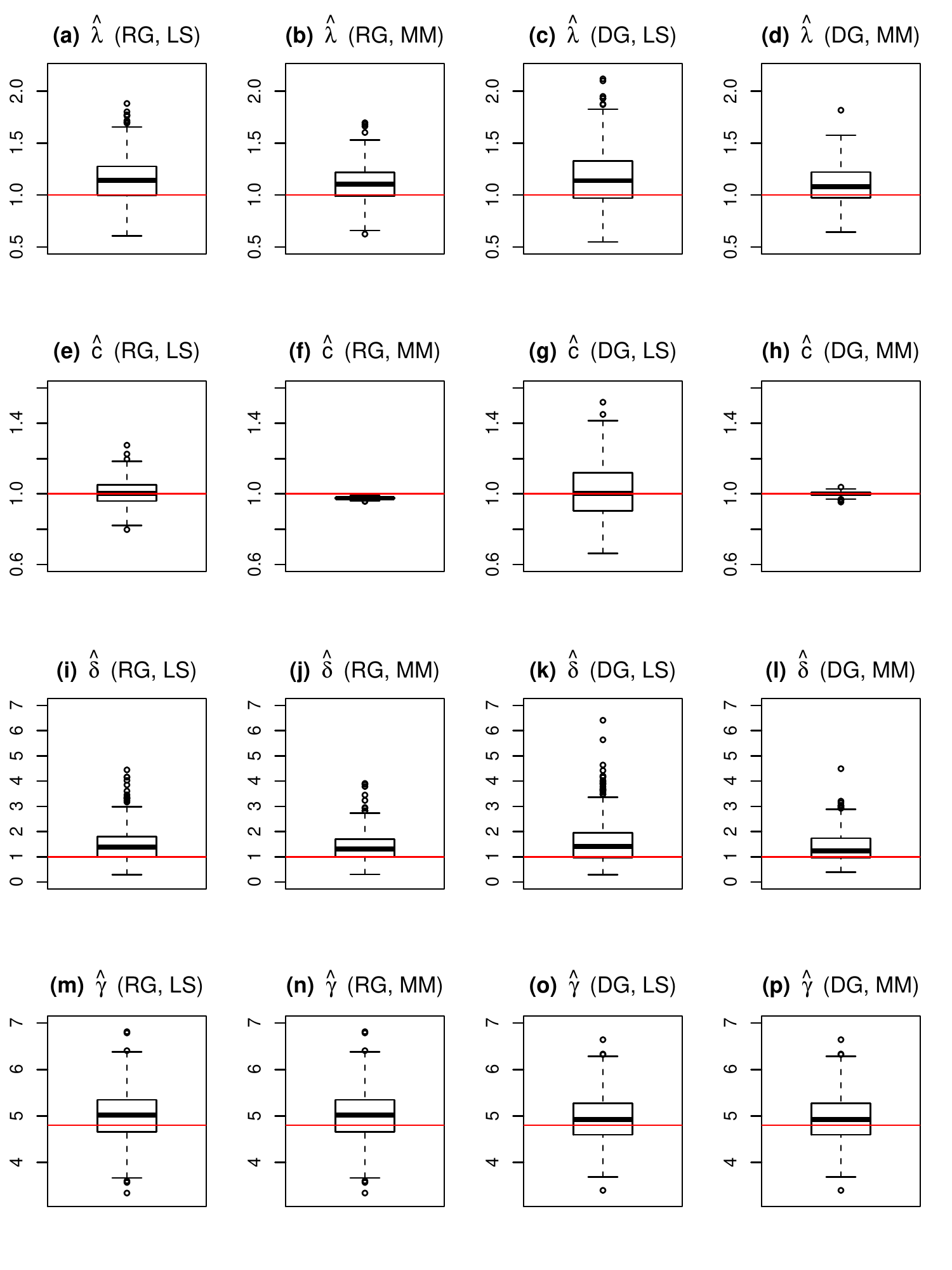}
\end{figure}

\begin{figure}[tbp]
\centering
\caption{Box plots of the parameter estimates from the LS and MM methods under the RG and DG algorithms in the case $c = 1$ and $L$ being a Gamma basis. The true values are denoted by the red lines.}
\label{fig:BPGamma1}
\includegraphics[width = 4.5in, height = 7in, trim = 0.4in 0.5in 0.4in 0in]{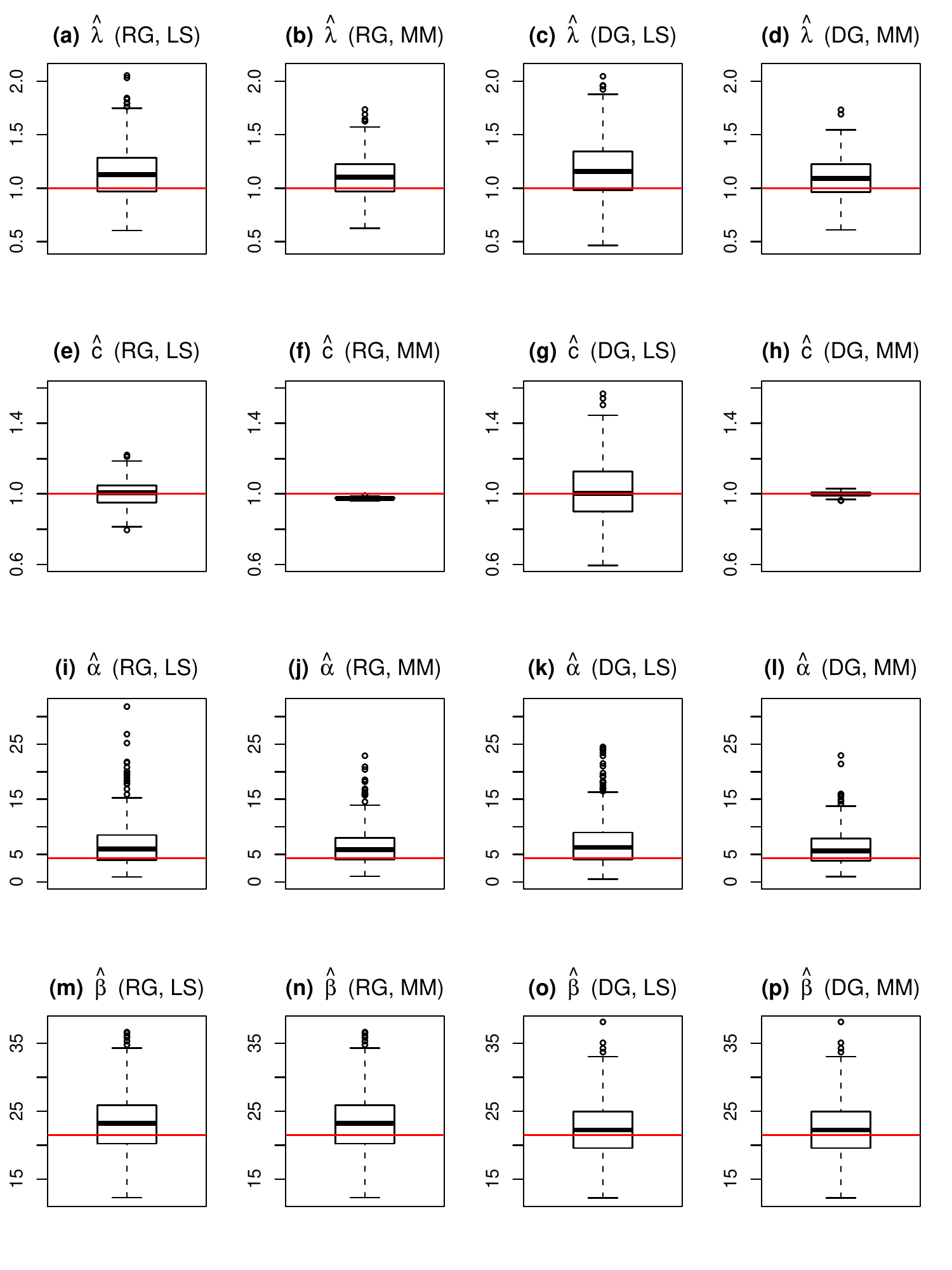}
\end{figure}

\begin{figure}[tbp]
\centering
\caption{Box plots of the parameter estimates from the LS and MM methods under the RG and DG algorithms in the case $c = 1$ and $L$ being a NIG basis. The true values are denoted by the red lines. The $179$ invalid estimates for $\alpha$ and $204$ invalid estimates for $\delta$ are omitted in the plots.}
\label{fig:BPNIG1}
\includegraphics[width = 4.5in, height = 9in, trim = 0.4in 0.5in 0.4in 0in]{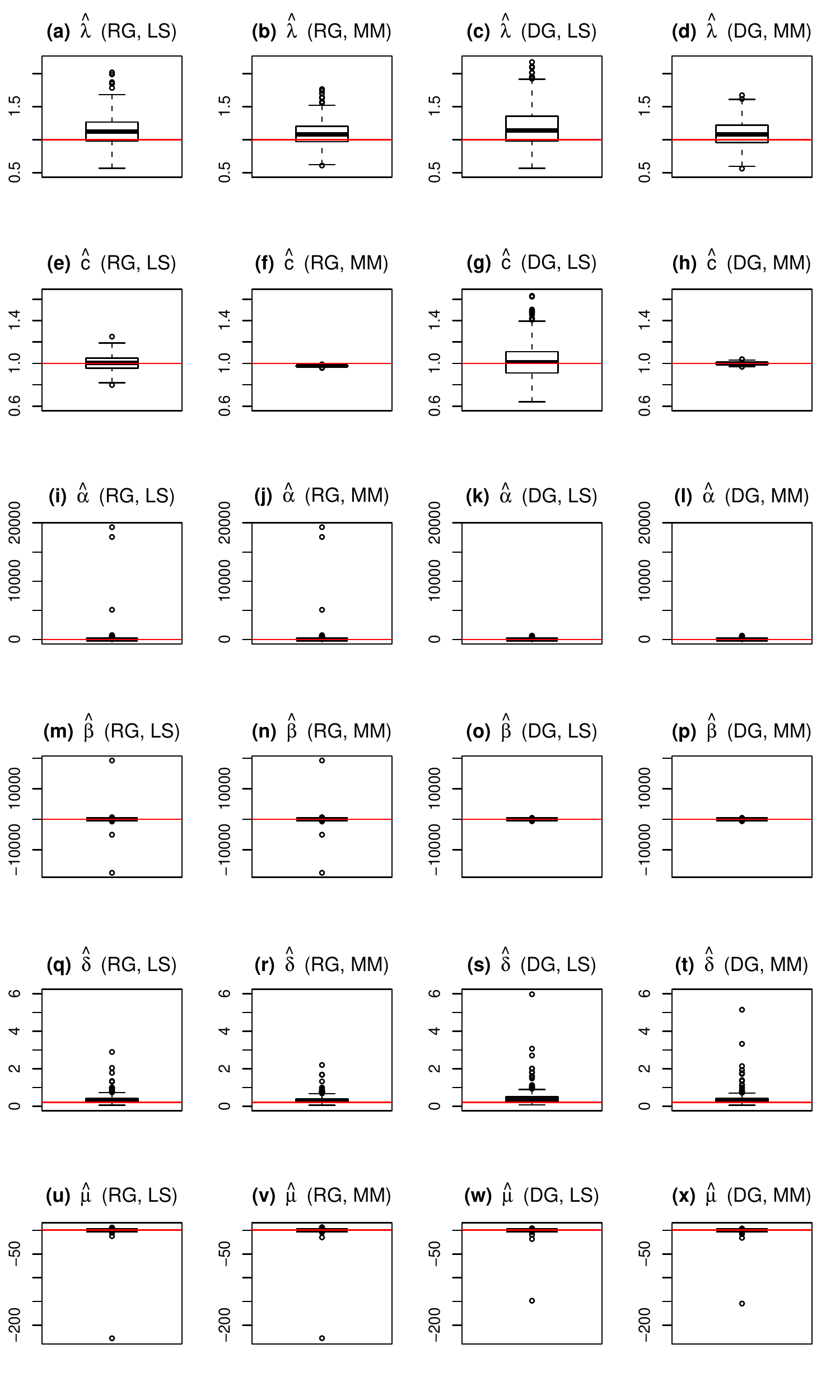}
\end{figure}

\clearpage

\section{Additional empirical results} \label{sec:addemp}

In this section, we show the heat plots of the data generated from the canonical $\mathrm{OU}_{\wedge}$ processes fitted to anomaly data of outgoing longwave radiation (OLR) from the International Research Institute for Climate and Society. The simulated data were generated using the DG algorithm and alternate coordinates were used for the plots in order to avoid the missing values. In each of the plots, the heat plot of the OLR data and the heat plot of the corresponding data run are placed side by side to facilitate comparison. The radiation anomalies are measured in watts per square metre. Due to the random nature of the simulations, we cannot judge the appropriateness of the fitted models based on individual runs. However, as diagonal lines seem to outline diamond-like structures in our simulated data, one can say that the fitted canonical $\mathrm{OU}_{\wedge}$ processes give linear approximations to the spatio-temporal dependencies observed in the OLR phenomena. 

\begin{figure}[here]
\centering
\caption{Heat plot of: (i) the radiation anomaly data; (ii) the Gaussian simulated data, seed 2.} \label{fig:hp1}
\includegraphics[width = 3.1in, height = 3.4in, trim = 0.5in 0.3in 0.1in 0in]{LWRadDataHeat.jpg}
\includegraphics[width = 3.1in, height = 3.4in, trim = 0.25in 0.3in 0.4in 0in]{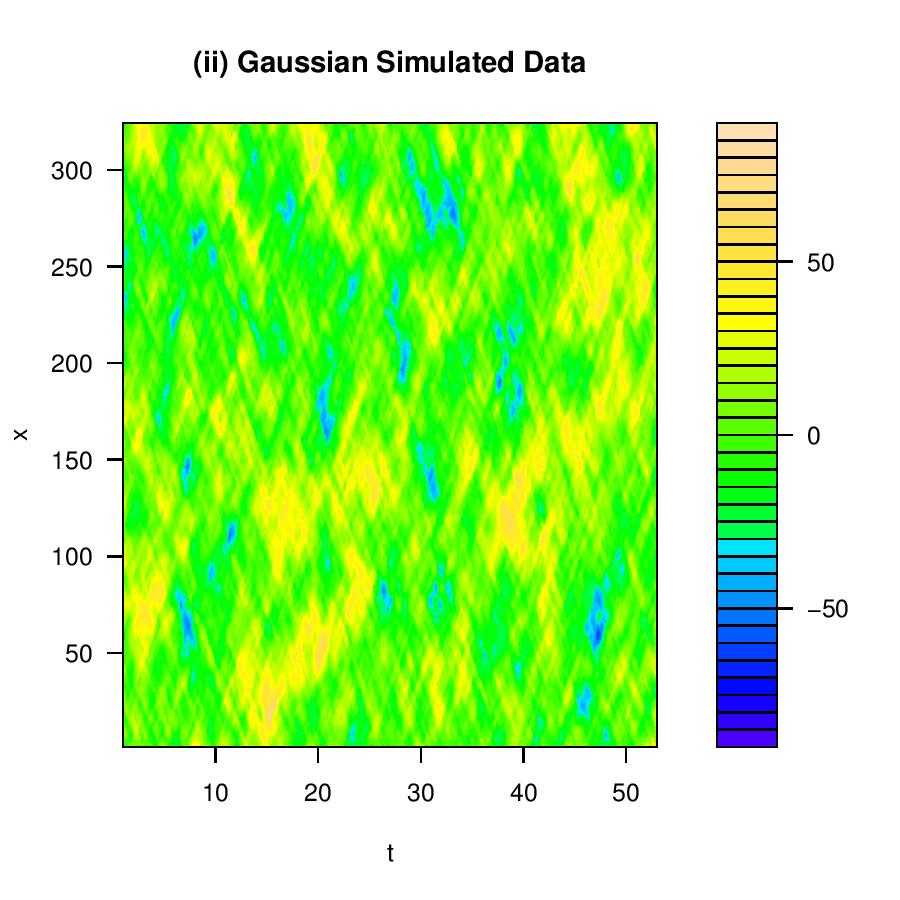}
\end{figure}

\begin{figure}[here]
\centering
\caption{Heat plot of: (i) the radiation anomaly data; (ii) the NIG simulated data, seed 2.} \label{fig:hp2}
\includegraphics[width = 3.1in, height = 3.4in, trim = 0.5in 0.3in 0.1in 0in]{LWRadDataHeat.jpg}
\includegraphics[width = 3.2in, height = 3.5in, trim = 0.25in 0.4in 0.4in 0in]{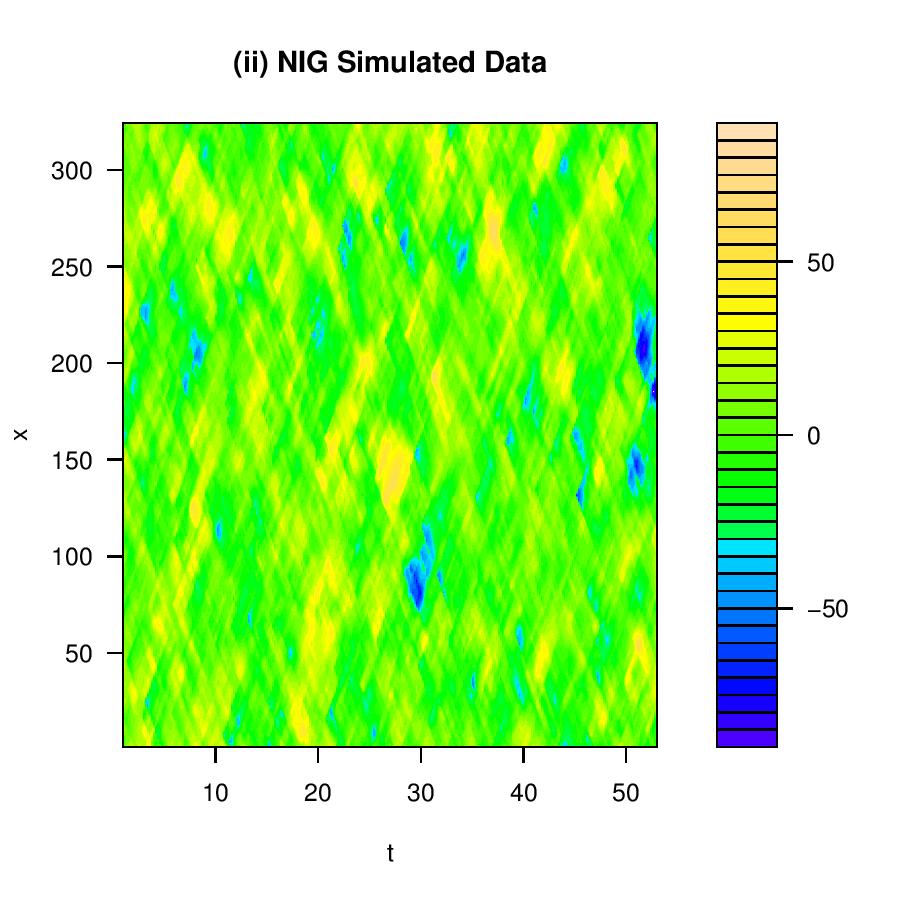}
\end{figure}

\clearpage

\begin{figure}[here]
\centering
\caption{Heat plot of: (i) the radiation anomaly data; (ii) the Gaussian simulated data, seed 3.}
\includegraphics[width = 3.1in, height = 3.4in, trim = 0.5in 0.3in 0.1in 0in]{LWRadDataHeat.jpg}
\includegraphics[width = 3.2in, height = 3.5in, trim = 0.25in 0.4in 0.4in 0in]{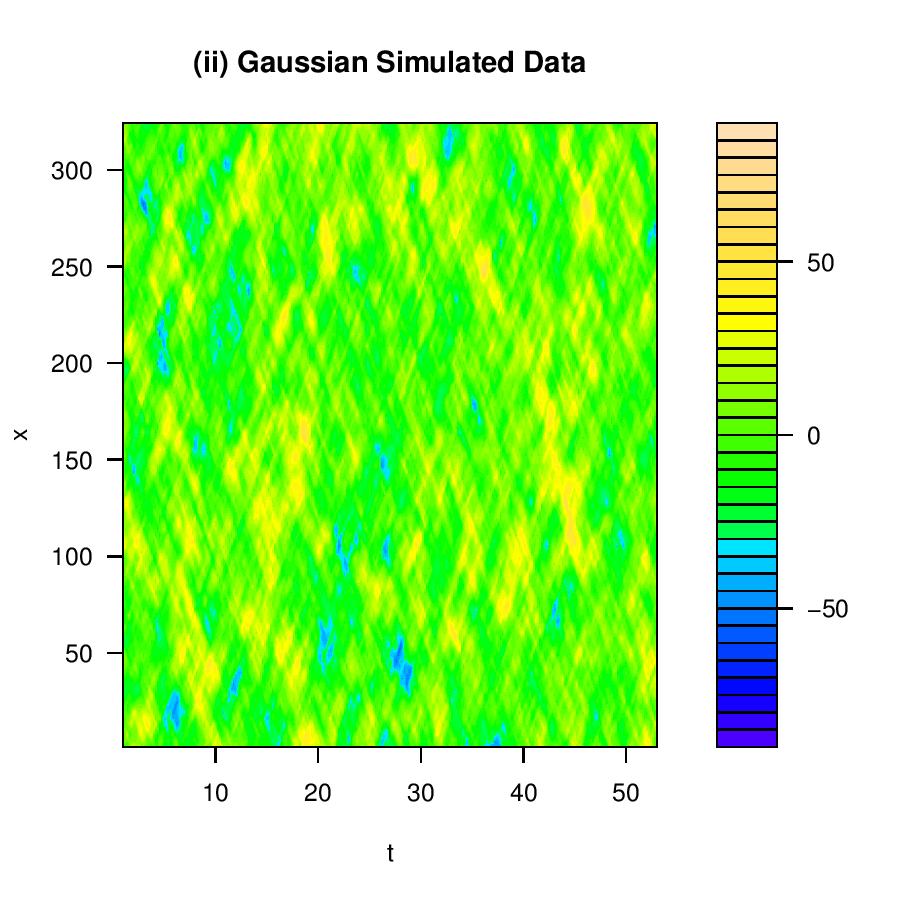}
\end{figure}

\begin{figure}[here]
\centering
\caption{Heat plot of: (i) the radiation anomaly data; (ii) the NIG simulated data, seed 3.}
\includegraphics[width = 3.1in, height = 3.4in, trim = 0.5in 0.3in 0.1in 0in]{LWRadDataHeat.jpg}
\includegraphics[width = 3.2in, height = 3.5in, trim = 0.25in 0.4in 0.4in 0in]{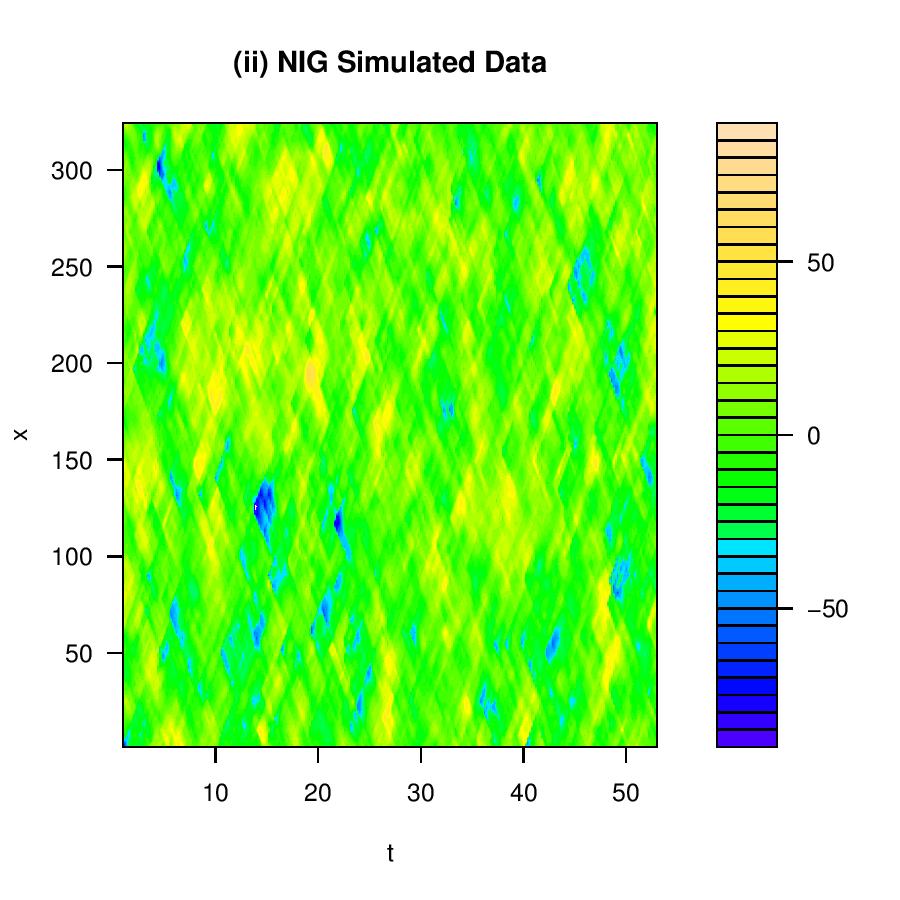}
\end{figure}

\clearpage

\begin{figure}[here]
\centering
\caption{Heat plot of: (i) the radiation anomaly data; (ii) the Gaussian simulated data, seed 4.}
\includegraphics[width = 3.1in, height = 3.4in, trim = 0.5in 0.3in 0.1in 0in]{LWRadDataHeat.jpg}
\includegraphics[width = 3.2in, height = 3.5in, trim = 0.25in 0.4in 0.4in 0in]{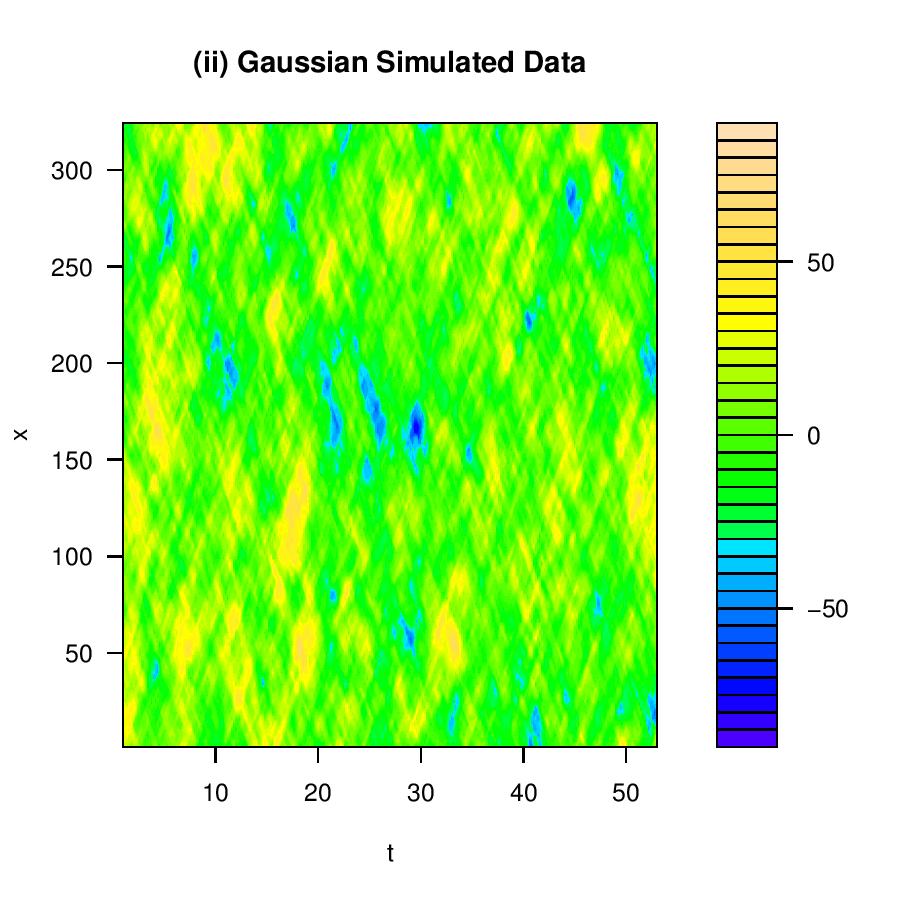}
\end{figure}

\begin{figure}[here]
\centering
\caption{Heat plot of: (i) the radiation anomaly data; (ii) the NIG simulated data, seed 4.}
\includegraphics[width = 3.1in, height = 3.4in, trim = 0.5in 0.3in 0.1in 0in]{LWRadDataHeat.jpg}
\includegraphics[width = 3.2in, height = 3.5in, trim = 0.25in 0.4in 0.4in 0in]{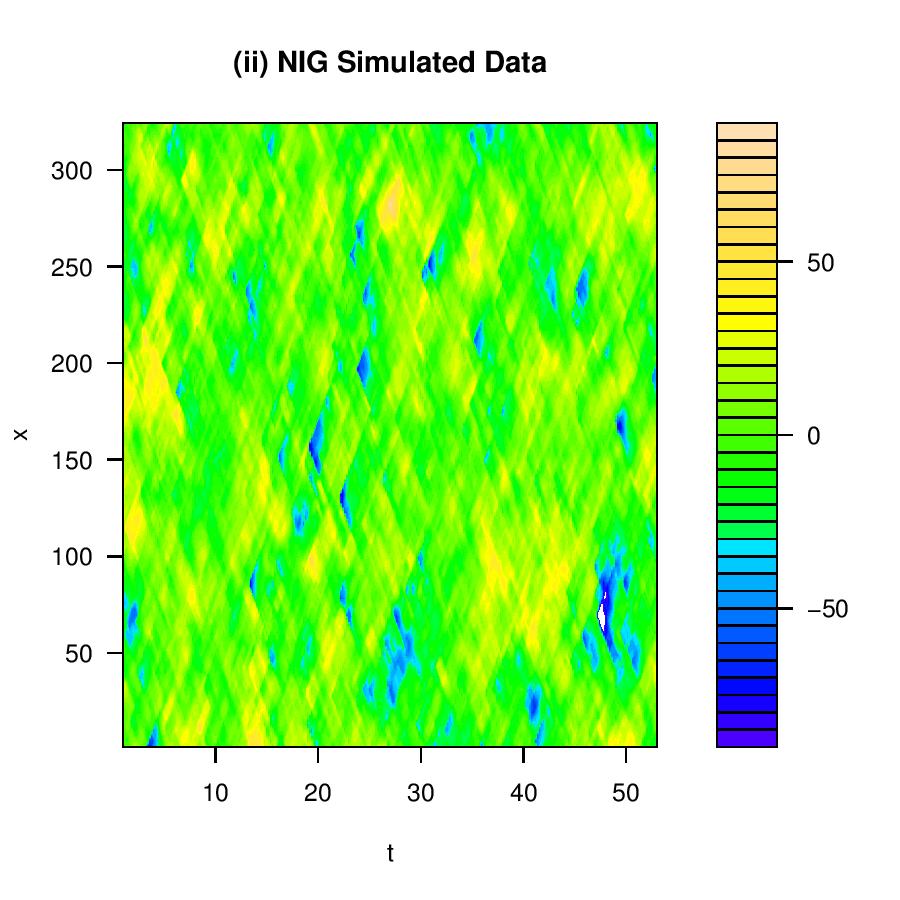}
\end{figure}

\clearpage 

\begin{figure}[here]
\centering
\caption{Heat plot of: (i) the radiation anomaly data; (ii) the Gaussian simulated data, seed 5.}
\includegraphics[width = 3.1in, height = 3.4in, trim = 0.5in 0.3in 0.1in 0in]{LWRadDataHeat.jpg}
\includegraphics[width = 3.2in, height = 3.5in, trim = 0.25in 0.4in 0.4in 0in]{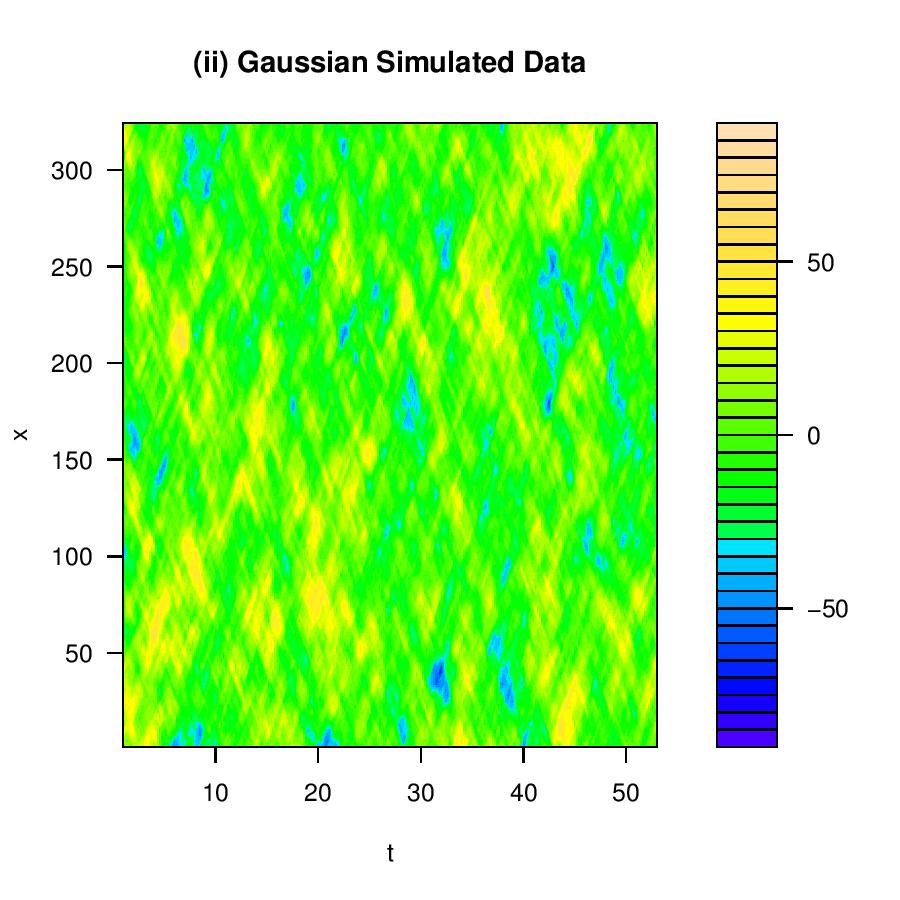}
\end{figure}

\begin{figure}[here]
\centering
\caption{Heat plot of: (i) the radiation anomaly data; (ii) the NIG simulated data, seed 5.}
\includegraphics[width = 3.1in, height = 3.4in, trim = 0.5in 0.3in 0.1in 0in]{LWRadDataHeat.jpg}
\includegraphics[width = 3.2in, height = 3.5in, trim = 0.25in 0.4in 0.4in 0in]{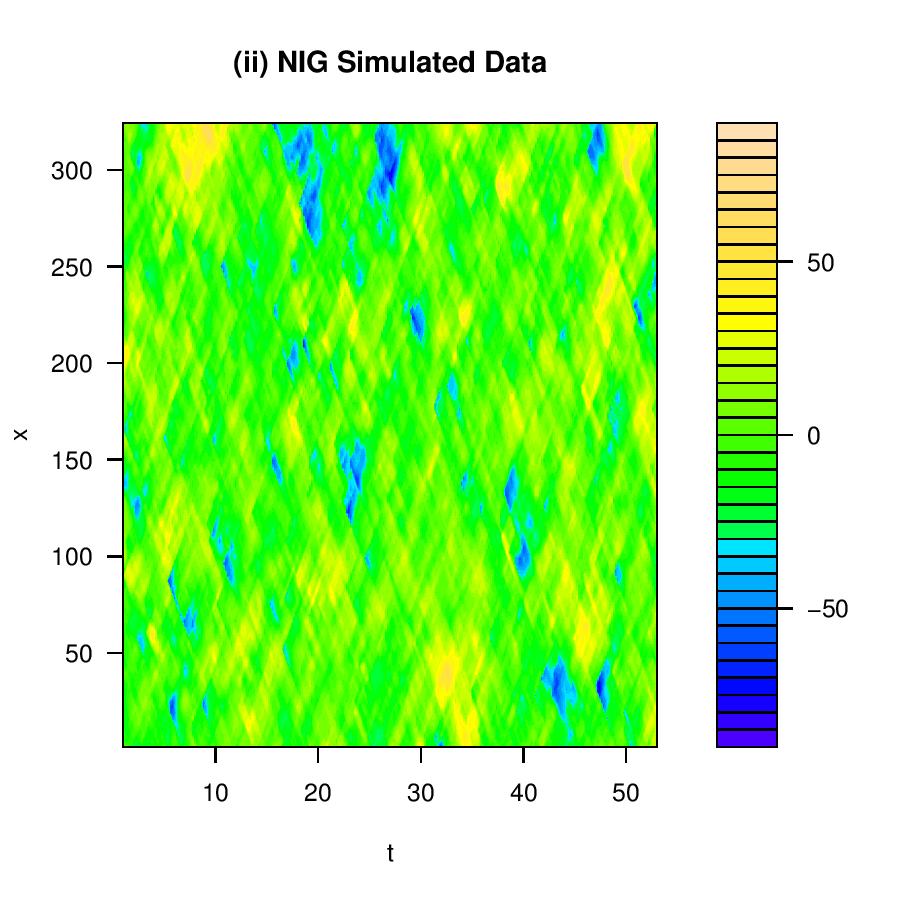}
\end{figure}

\clearpage

\begin{figure}[here]
\centering
\caption{Heat plot of: (i) the radiation anomaly data; (ii) the Gaussian simulated data, seed 6.}
\includegraphics[width = 3.1in, height = 3.4in, trim = 0.5in 0.3in 0.1in 0in]{LWRadDataHeat.jpg}
\includegraphics[width = 3.2in, height = 3.5in, trim = 0.25in 0.4in 0.4in 0in]{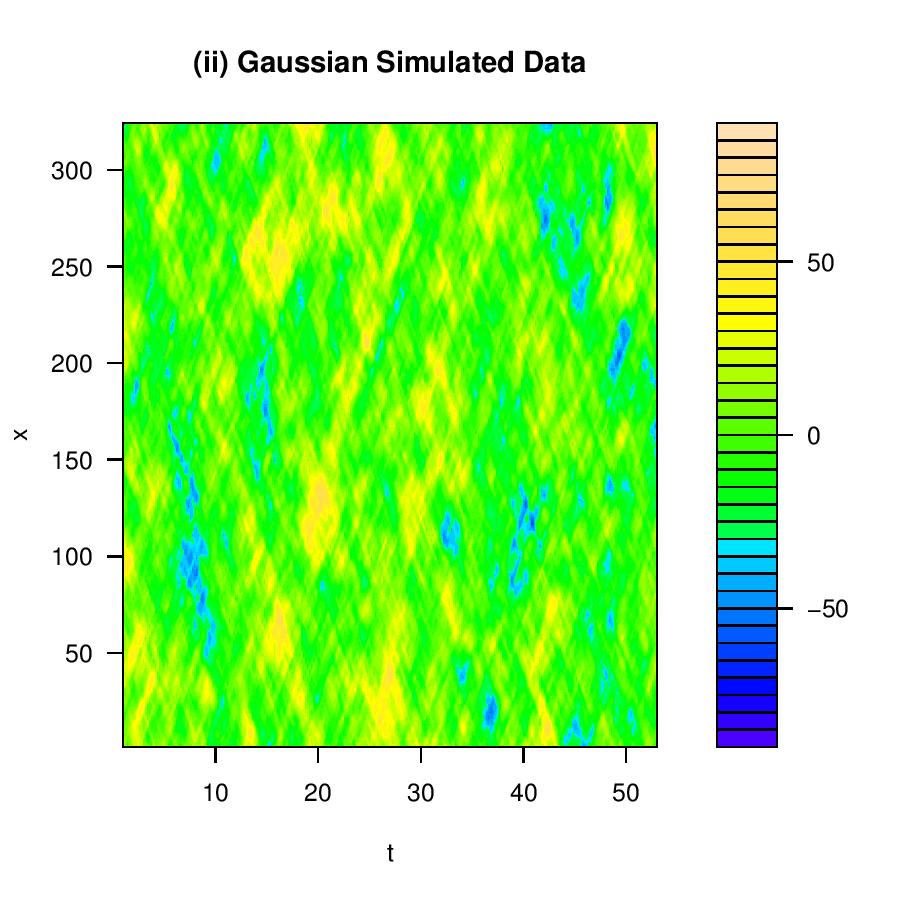}
\end{figure}

\begin{figure}[here]
\centering
\caption{Heat plot of: (i) the radiation anomaly data; (ii) the NIG simulated data, seed 6.}
\includegraphics[width = 3.1in, height = 3.4in, trim = 0.5in 0.3in 0.1in 0in]{LWRadDataHeat.jpg}
\includegraphics[width = 3.2in, height = 3.5in, trim = 0.25in 0.4in 0.4in 0in]{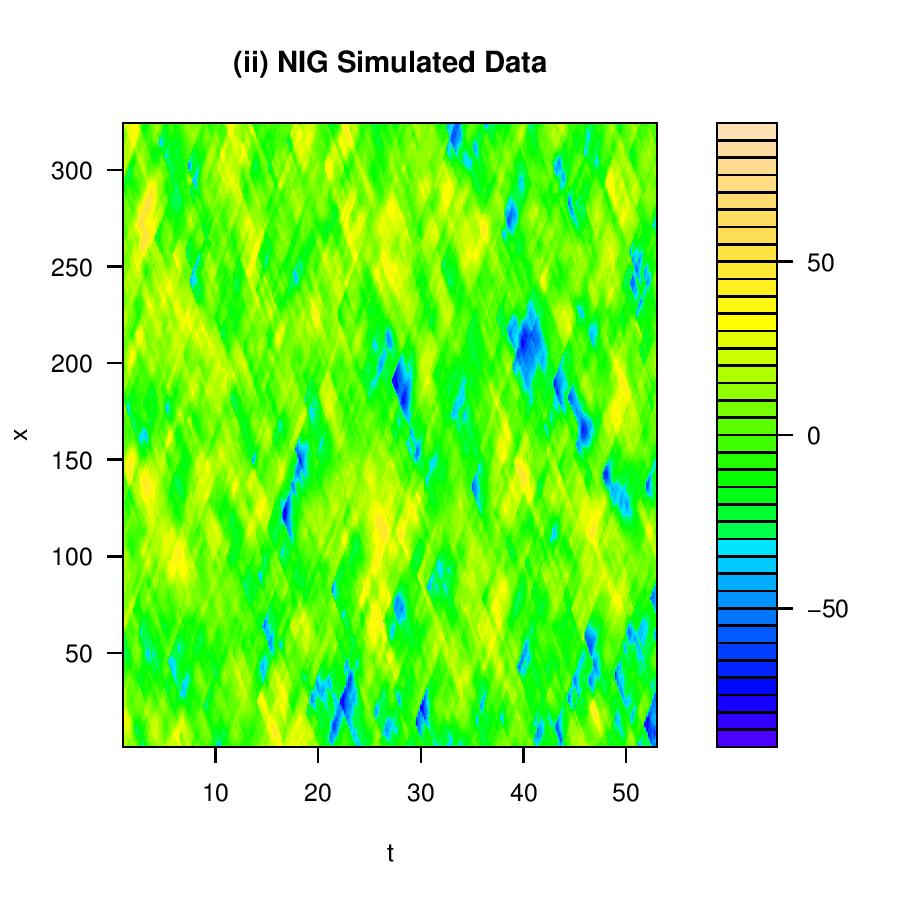}
\end{figure}

\clearpage

\begin{figure}[here]
\centering
\caption{Heat plot of: (i) the radiation anomaly data; (ii) the Gaussian simulated data, seed 7.}
\includegraphics[width = 3.1in, height = 3.4in, trim = 0.5in 0.3in 0.1in 0in]{LWRadDataHeat.jpg}
\includegraphics[width = 3.2in, height = 3.5in, trim = 0.25in 0.4in 0.4in 0in]{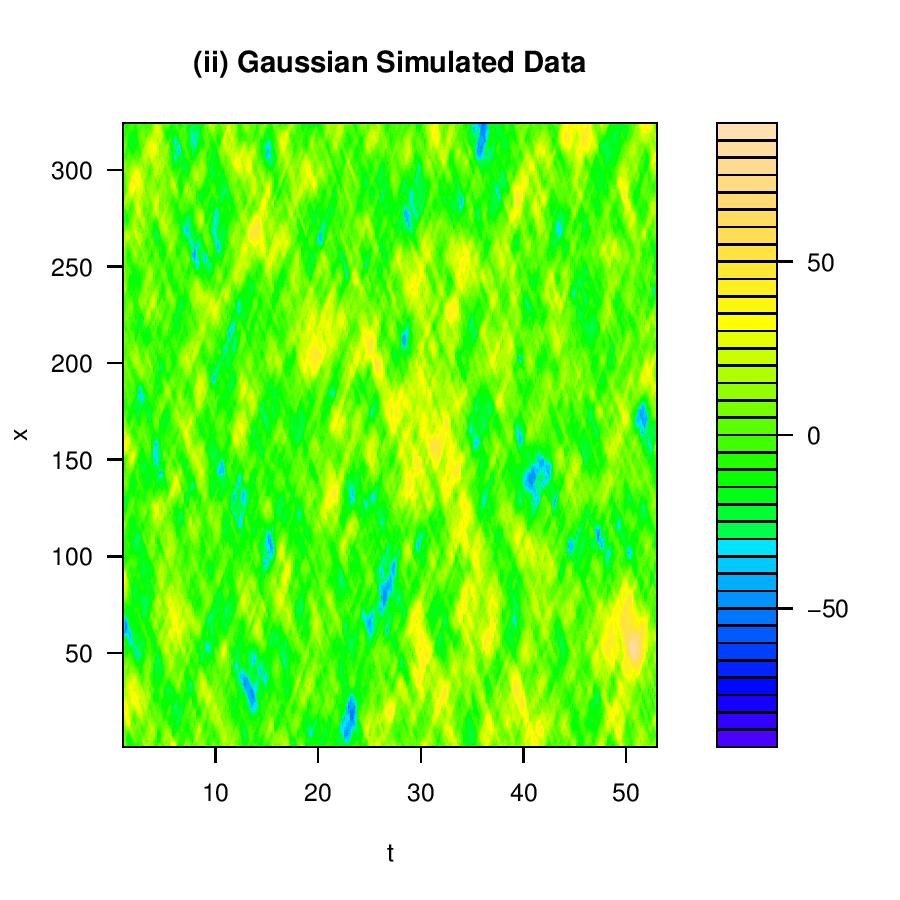}
\end{figure}

\begin{figure}[here]
\centering
\caption{Heat plot of: (i) the radiation anomaly data; (ii) the NIG simulated data, seed 7.}
\includegraphics[width = 3.1in, height = 3.4in, trim = 0.5in 0.3in 0.1in 0in]{LWRadDataHeat.jpg}
\includegraphics[width = 3.2in, height = 3.5in, trim = 0.25in 0.4in 0.4in 0in]{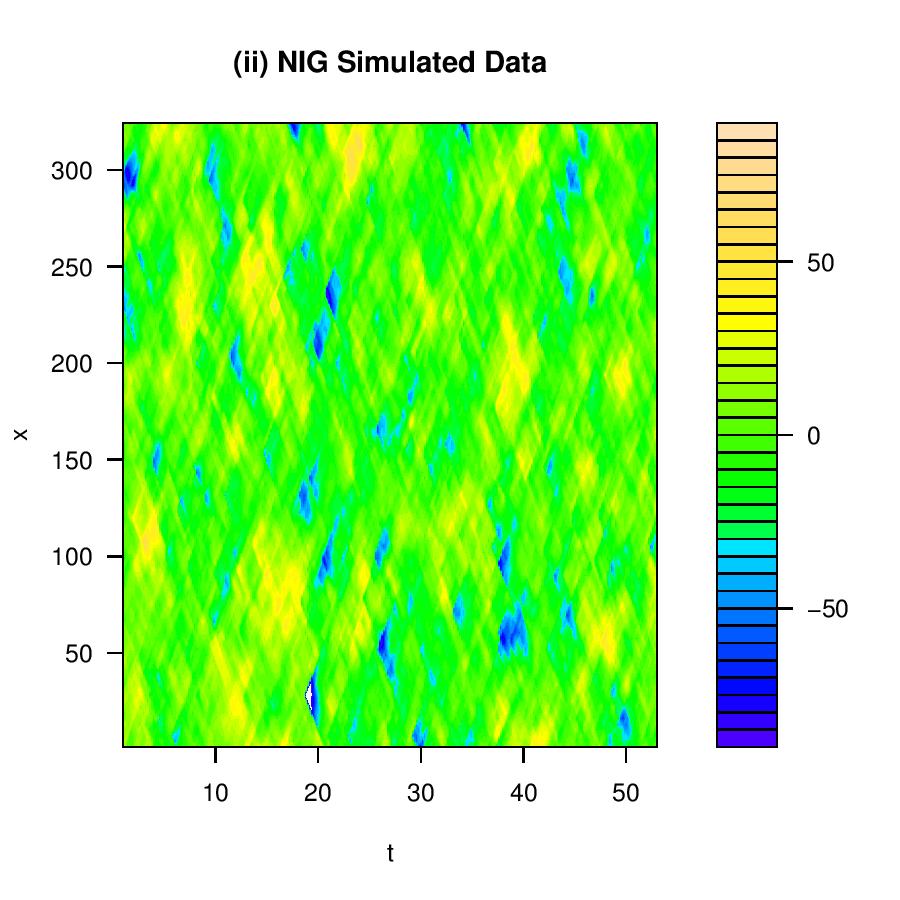}
\end{figure}

\clearpage

\begin{figure}[here]
\centering
\caption{Heat plot of: (i) the radiation anomaly data; (ii) the Gaussian simulated data, seed 8.}
\includegraphics[width = 3.1in, height = 3.4in, trim = 0.5in 0.3in 0.1in 0in]{LWRadDataHeat.jpg}
\includegraphics[width = 3.2in, height = 3.5in, trim = 0.25in 0.4in 0.4in 0in]{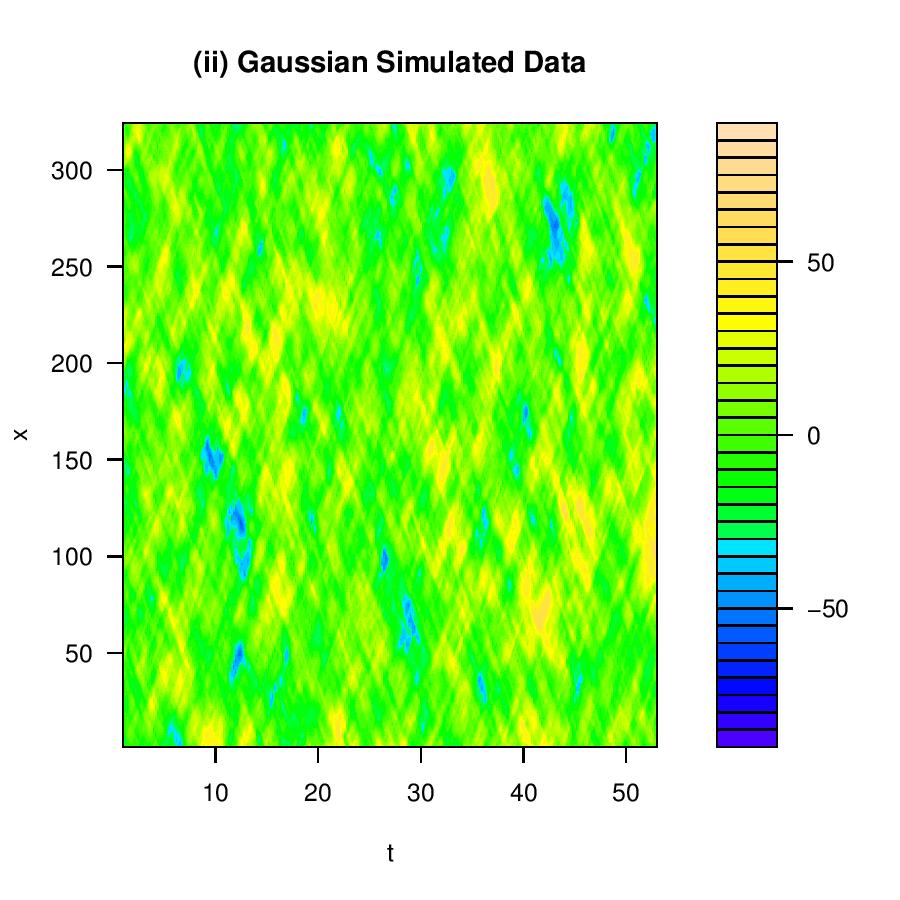}
\end{figure}

\begin{figure}[here]
\centering
\caption{Heat plot of: (i) the radiation anomaly data; (ii) the NIG simulated data, seed 8.}
\includegraphics[width = 3.1in, height = 3.4in, trim = 0.5in 0.3in 0.1in 0in]{LWRadDataHeat.jpg}
\includegraphics[width = 3.2in, height = 3.5in, trim = 0.25in 0.4in 0.4in 0in]{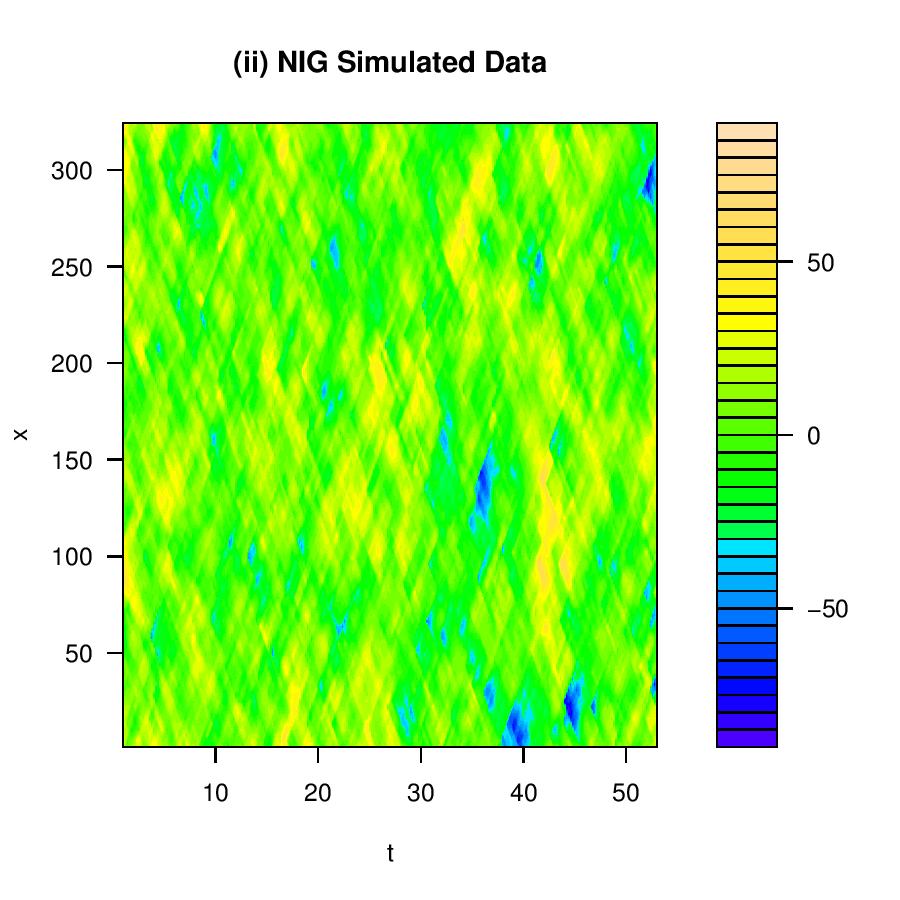}
\end{figure}

\clearpage 

\begin{figure}[here]
\centering
\caption{Heat plot of: (i) the radiation anomaly data; (ii) the Gaussian simulated data, seed 9.}
\includegraphics[width = 3.1in, height = 3.4in, trim = 0.5in 0.3in 0.1in 0in]{LWRadDataHeat.jpg}
\includegraphics[width = 3.2in, height = 3.5in, trim = 0.25in 0.4in 0.4in 0in]{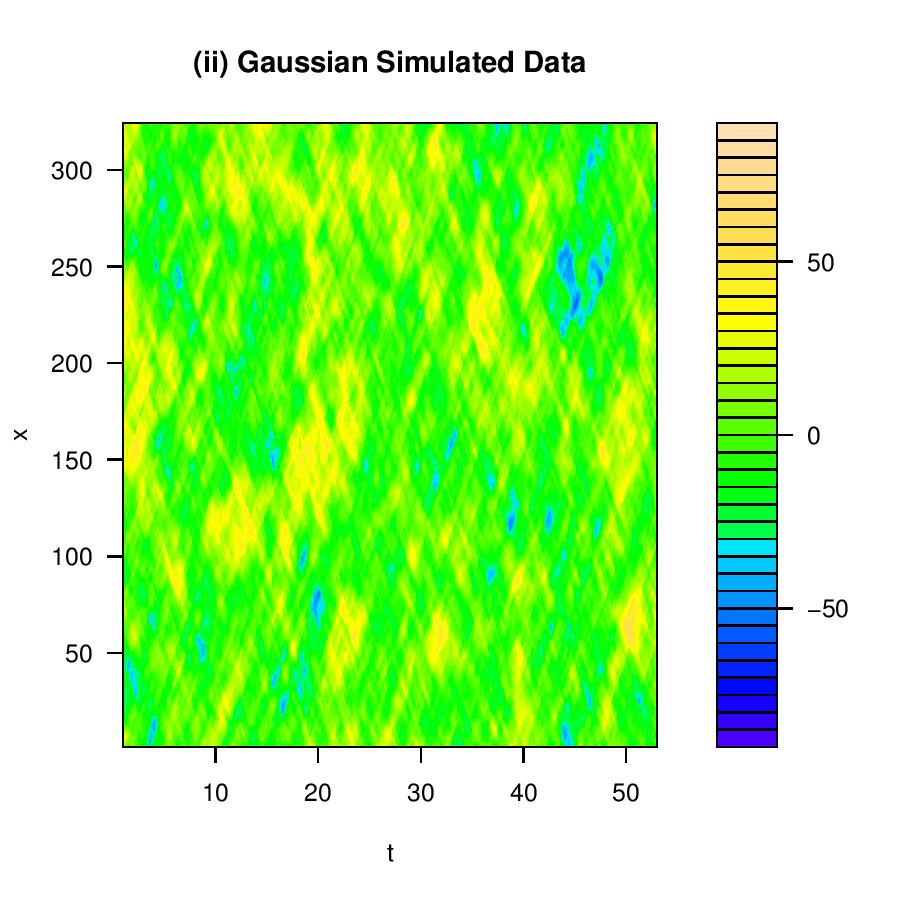}
\end{figure}

\begin{figure}[here]
\centering
\caption{Heat plot of: (i) the radiation anomaly data; (ii) the NIG simulated data, seed 9.}
\includegraphics[width = 3.1in, height = 3.4in, trim = 0.5in 0.3in 0.1in 0in]{LWRadDataHeat.jpg}
\includegraphics[width = 3.2in, height = 3.5in, trim = 0.25in 0.4in 0.4in 0in]{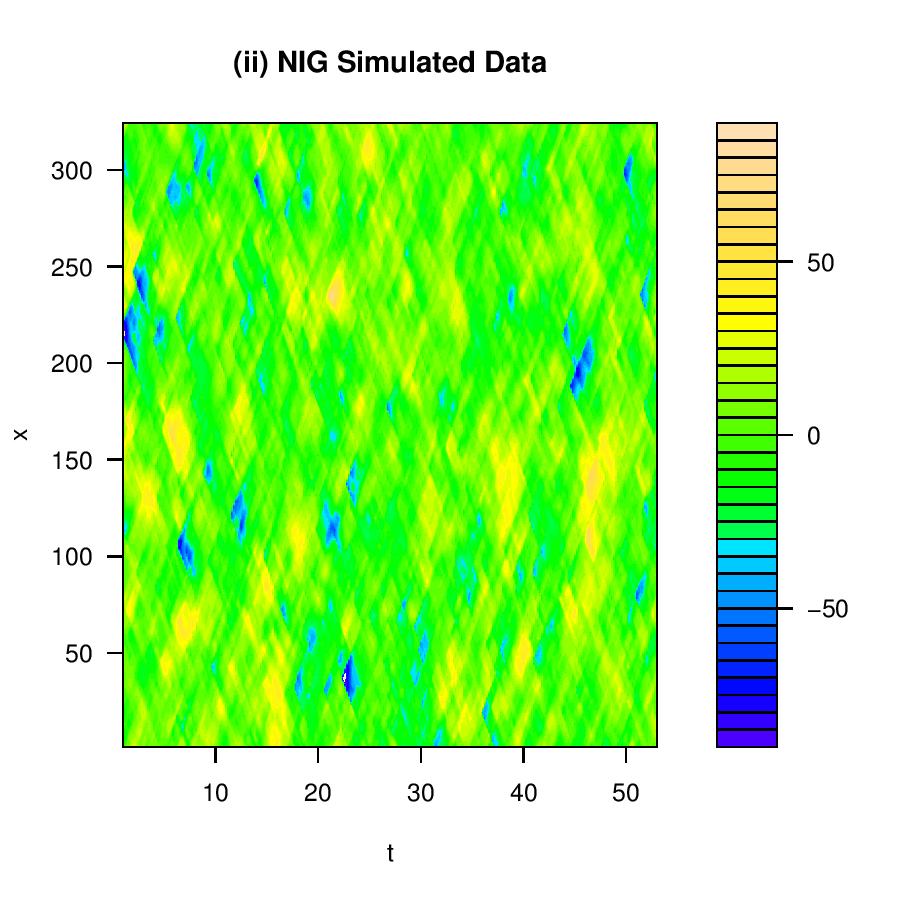}
\end{figure}

\clearpage

\begin{figure}[here]
\centering
\caption{Heat plot of: (i) the radiation anomaly data; (ii) the Gaussian simulated data, seed 10.}
\includegraphics[width = 3.1in, height = 3.4in, trim = 0.5in 0.3in 0.1in 0in]{LWRadDataHeat.jpg}
\includegraphics[width = 3.2in, height = 3.5in, trim = 0.25in 0.4in 0.4in 0in]{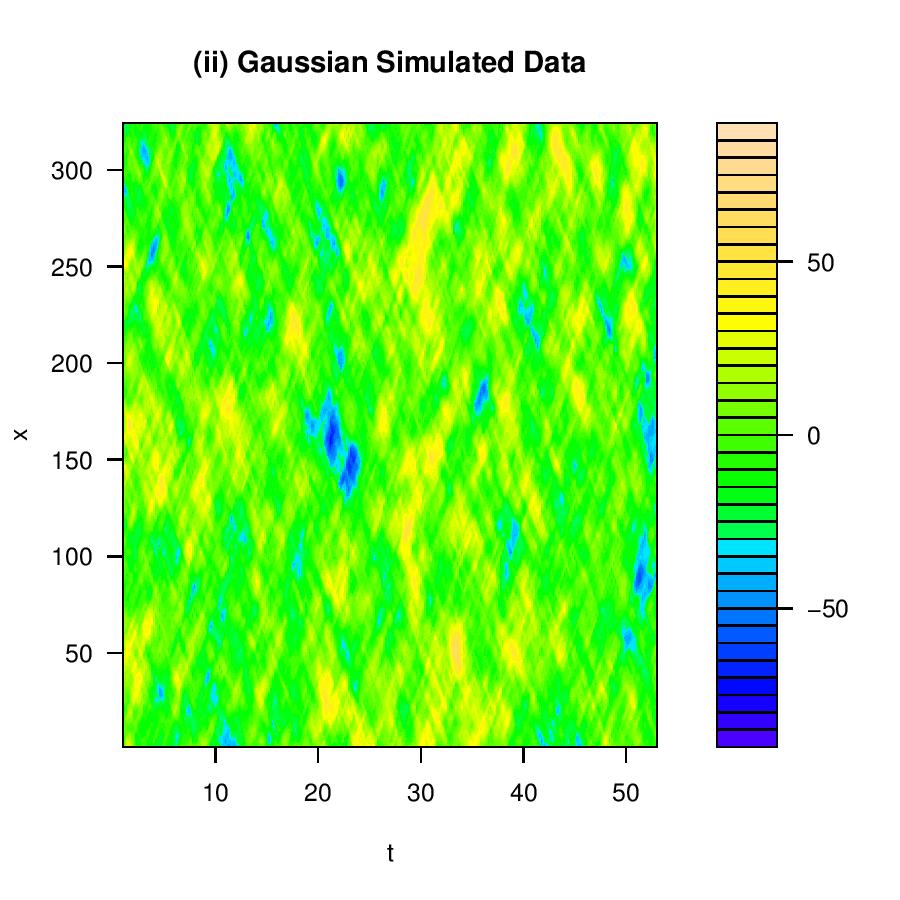}
\end{figure}

\begin{figure}[here]
\centering
\caption{Heat plot of: (i) the radiation anomaly data; (ii) the NIG simulated data, seed 10.}
\includegraphics[width = 3.1in, height = 3.4in, trim = 0.5in 0.3in 0.1in 0in]{LWRadDataHeat.jpg}
\includegraphics[width = 3.2in, height = 3.5in, trim = 0.25in 0.4in 0.4in 0in]{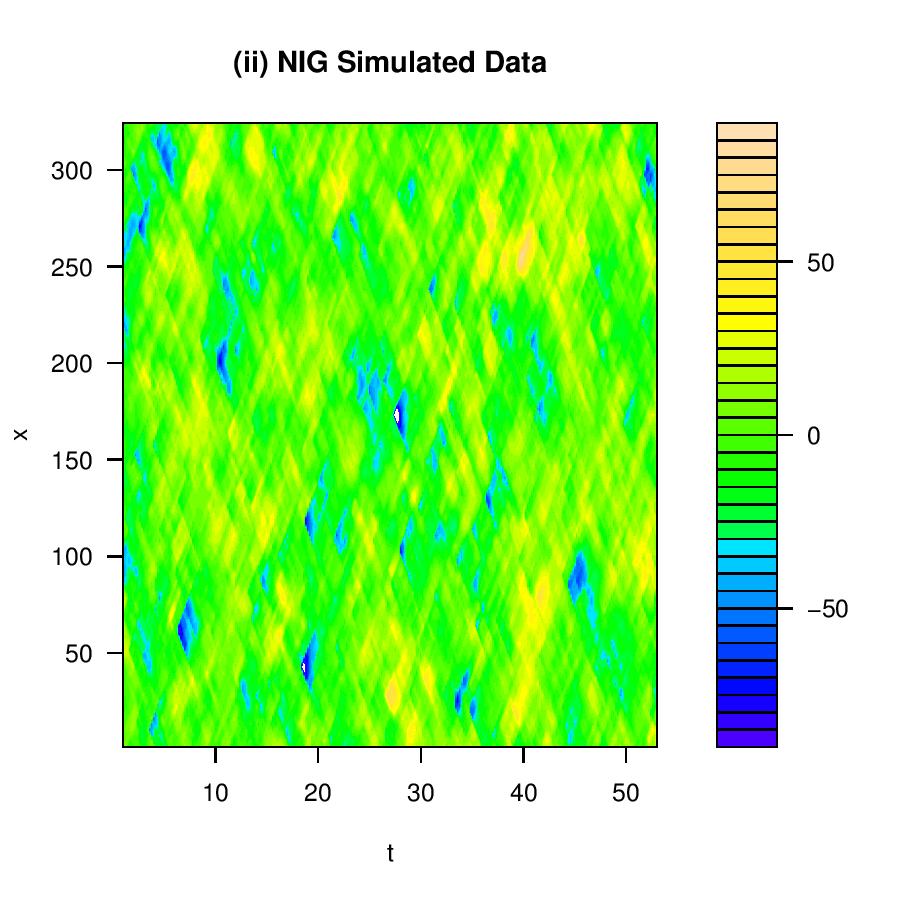}
\end{figure}

\clearpage

\bibliographystyle{agsm} 
\bibliography{refs}